\documentclass[format=acmsmall,screen=true,nonacm]{acmart} 

\acmBooktitle{}

\makeatletter
\let\@authorsaddresses\@empty
\makeatother

\usepackage[utf8]{inputenc}
\usepackage{enumitem}
\usepackage[section]{placeins}
\usepackage{hyperref}
\hypersetup{
	colorlinks,
	linkcolor={black!30!blue},
	citecolor={black!30!blue},
	urlcolor={black!30!blue}
}

\usepackage{array,amsmath,amsthm,amsfonts,graphicx}
\usepackage{tikz,tikz-cd}
\usetikzlibrary{backgrounds,fit,decorations.pathreplacing,shapes.misc,decorations.text,calc,positioning}  
\usepackage{multirow,blkarray}
\usepackage{braket}
\usepackage{subcaption}
\usepackage{color}
\usepackage{longtable}

\allowdisplaybreaks

\newcommand{\subfigqsaqc}[2]{
	\begin{subfigure}[b]{0.33\textwidth}
		\centering\resizebox{\textwidth}{!}
		{\renewcommand{\eps}{\epsgr}\input{figs/QSAQC/#1}\renewcommand{\eps}{\epsgr}}
		\caption{Step #1#2}
		\label{fig:qsaqc:#1}
\end{subfigure}}

\newtheorem{manualtheoreminner}{Theorem}
\newenvironment{manualtheorem}[1]{%
	\renewcommand\themanualtheoreminner{#1}%
	\manualtheoreminner
}{\endmanualtheoreminner}

\newtheorem{manualpropinner}{Proposition}
\newenvironment{manualprop}[1]{%
	\renewcommand\themanualpropinner{#1}%
	\manualpropinner
}{\endmanualpropinner}

\newsavebox{\bigmat}
\newsavebox{\tbQS}

\definecolor{blacier}{rgb}{0.22745098039,0.55686274509,0.7294117647}
\definecolor{lowgreen}{rgb}{0.40390625,0.6109375,0.09109375}
\definecolor{lowgray}{gray}{0.5}

\def\blacier#1{{\textcolor{blacier}{#1}}}
\def\lowgreen#1{{\textcolor{lowgreen}{#1}}}
\def\RC#1{{\textcolor{red}{#1}}}

\newcommand{\cal}[1]{\mathcal{#1}}
\newcommand{\calA}{{\ensuremath{{\cal A}}}}
\newcommand{\calB}{{\ensuremath{{\cal B}}}}
\newcommand{\calC}{{\ensuremath{{\cal C}}}}
\newcommand{\calD}{{\ensuremath{{\cal D}}}}

\newcommand{\calG}{{\ensuremath{{\cal G}}}}
\newcommand{\calH}{{\ensuremath{{\cal H}}}}
\newcommand{\calI}{{\ensuremath{{\cal I}}}}
\newcommand{\calK}{{\ensuremath{{\cal K}}}}

\newcommand{\calN}{{\ensuremath{{\cal N}}}}
\newcommand{\calO}{{\ensuremath{{\cal O}}}}
\newcommand{\calQ}{{\ensuremath{{\cal Q}}}}
\newcommand{\calS}{{\ensuremath{{\cal S}}}}
\newcommand{\calT}{{\ensuremath{{\cal T}}}}
\newcommand{\calV}{{\ensuremath{{\cal V}}}}
\newcommand{\calW}{{\ensuremath{{\cal W}}}}
\newcommand{\calX}{{\ensuremath{{\cal X}}}}

\newcommand{\In}{\texttt{I}}
\newcommand{\Out}{\texttt{O}}
\newcommand{\ovva}{\overline{a}}
\newcommand{\ovvap}{\overline{a}'}
\newcommand{\ovvq}{\overline{q}}
\newcommand{\ovvqp}{\overline{q}'}
\newcommand{\ovvb}{\overline{b}}
\newcommand{\ovvbp}{\overline{b}'}
\newcommand{\ovvp}{\overline{p}}
\newcommand{\ovvpp}{\overline{p}'}
\newcommand{\eps}{\varepsilon}
\newcommand{\epsgr}{\textcolor{lowgray}{\varepsilon}}

\def\subspace#1{\overline{\calH}_#1}

\newtheorem{remark}{Remark}[section]

\DeclareMathOperator{\SET}{SET}
\DeclareMathOperator{\Conn}{Conn}
\newcommand{\sector}[6]{
	\draw[] (#1,#2) -- (#1,#2+1) -- (#1+2,#2+1) [rounded corners] -- (#1+2,#2) --cycle;
	\draw[] (#1,#2+1) [rounded corners] -- (#1,#2+2) -- (#1+2,#2+2) [sharp corners] -- (#1+2,#2+1) -- cycle;
	\draw (#1+2,#2+1.8) [rounded corners] coordinate[] (C) 
	--  (#1 + 2.3 ,#2+1) coordinate[] (A) 
	-- (#1 + 2,#2+0.2) coordinate[] (B);
	
	\node[align=center] at (#1+1,#2+.5) {\huge #3};
	\node[align=center] at (#1+1,#2+1.5) {\huge #6};
	
	\node[inner sep= 0pt] at (#1-0.26,#2+2) (name#4) {\Huge #4};
	\node[inner sep= 0pt] at (A) (outarrow#4) {};
	\node[inner sep= 0pt] at (#1,#2+1) (inarrow#4) {};
	\node[inner sep= 0pt] at (#1+.5,#2) (bottomleft#4) {};
	\node[inner sep= 0pt] at (#1+1,#2) (bottommid#4) {};
	\node[inner sep= 0pt] at (#1+1.5,#2) (bottomright#4) {};
	\node[inner sep= 0pt] at (#1+.5,#2+2) (topleft#4) {};
	\node[inner sep= 0pt] at (#1+1,#2+2) (topmid#4) {};
	\node[inner sep= 0pt] at (#1+1.5,#2+2) (topright#4) {};	
}

\newcommand{\sectorIOWQ}[9][]{
    \draw[] (#2+2,#3+1.8) [rounded corners] coordinate[] (C) 
	--  (#2 + 2.7 ,#3+1) coordinate[] (A) 
	-- (#2 + 2,#3+0.3) coordinate[] (B);
	\node at (barycentric cs:A=1.5,B=1,C=1) {#8};
	\node at (#2-0.25,#3+2) (name) {#9}; 
	
	\draw[] (#2,#3) --  (#2,#3+1) -- (#2+1,#3+1) -- (#2+1,#3)[rounded corners]  -- cycle;
	\draw[] (#2+1,#3) -- (#2+1,#3+1) -- (#2+2,#3+1) [rounded corners] -- (#2+2,#3) [sharp corners] -- cycle;
	\draw[] (#2,#3+1) [rounded corners] -- (#2,#3+2) [sharp corners] -- (#2+1,#3+2) -- (#2+1,#3+1) -- cycle;
	\draw[] (#2+1,#3+1) -- (#2+1,#3+2) [rounded corners] -- (#2+2,#3+2) [sharp corners] -- (#2+2,#3+1) -- cycle;
	\node[] at (#2+.5,#3+.5) {#4};
	\node[] at (#2+1.5,#3+.5) {#5};
	\node[] at (#2+.5,#3+1.5) {#6};
	\node[] at (#2+1.5,#3+1.5) {#7};
	
	\node[inner sep= 0pt] at (A) (outarrow#1) {};
	\node[inner sep= 0pt] at (#2,#3+1) (inarrow#1) {};
}

\newcommand{\fleche}[4][]{
	\draw[{Round Cap}-stealth,color=#4,line width=1mm,#1] (#2) to (#3);
}

\newcommand{\flechecontrol}[5]{
    \draw[-stealth,color=#5,line width=0.75mm] (#1).. controls (#3,#4) ..(#2);
}

\tikzset{cross/.style={cross out, draw, 
		minimum size=2*(#1-\pgflinewidth), 
		inner sep=0pt, outer sep=0pt}}

\def\drawpolygon#1,#2;{
    \begin{pgfonlayer}{background}
        \filldraw[line width=20,join=round      ](#1.center)foreach\A in{#2}{--(\A.center)}--cycle;
        \filldraw[line width=19,join=round,white](#1.center)foreach\A in{#2}{--(\A.center)}--cycle;
    \end{pgfonlayer}
}

\newcommand{\rectsec}[3]{
	\draw[dashed] (#1-0.3,0.3) rectangle (#1+0.3,-#2+0.7);
    \node (#3) at (#1,0.5)  {#1}; 
}

\newcommand{\rectpos}[6]{
{\draw[dashed] (#1-0.3,0.3-#3) rectangle (#1+0.3,-#2+0.7-#3);}
    \node (#6) at (#1-0.3,-#2*0.5-#3*0.5+0.5) {};
    \node (#6r) at (#1+0.3,-#2*0.5-#3*0.5+0.5) {};
    \node (#5) at (#1,0.5-#3)  {#4} ; }

\newcommand{\flechedouble}[5]{
    \draw[stealth-stealth,color=#5,line width=0.75mm,color=gray] (#1).. controls (#3,#4) ..(#2);
}

\definecolor{box_color}{rgb}{1,1,1}
\definecolor{data_color}{RGB}{0,0,255}
\definecolor{dark_box_color}{rgb}{0,0,0}

\newcommand{\drawName}[3]{
    \node[draw,fill=dark_box_color!25] at (#1,#2) {
    \Large #3};
}

\newcommand{\drawData}[4]{
    \node[draw,fill=data_color!10,minimum size=1cm,circle] (#4#3) at (#1,#2) {
    \large $ #3 $};
}

\newcommand{\adresse}[4]{

    \ifthenelse{#4=1}{
    \drawName{#1+1}{#2+1.5}{#3};
    \drawData{#1+0.5}{#2+0.5}{V}{#3};
    \drawData{#1+1.5}{#2+0.5}{H}{#3};
    }{
    \ifthenelse{#4=2}{
    \drawName{#1+1.5}{#2+1}{#3};
    \drawData{#1+0.5}{#2+0.5}{V}{#3};
    \drawData{#1+0.5}{#2+1.5}{H}{#3};
    }{
    \ifthenelse{#4=3}{
    \drawName{#1+1}{#2+0.5}{#3};
    \drawData{#1+1.5}{#2+1.5}{V}{#3};
    \drawData{#1+0.5}{#2+1.5}{H}{#3};
    }{
    \ifthenelse{#4=4}{
    \drawName{#1+0.5}{#2+1}{#3};
    \drawData{#1+1.5}{#2+1.5}{V}{#3};
    \drawData{#1+1.5}{#2+0.5}{H}{#3};
    }}}}
}

\hypersetup{
	pdftitle = {Addressable quantum gates},
	pdfauthor = {Pablo Arrighi,
		Christopher Cedzich,
		Marin Costes,
		Ulysse Rémond,
		Benoît Valiron},
	pdfsubject = {Quantum computing},
	pdfkeywords = {quantum computing, indefinite causal orders, models for quantum computing, addressable quantum gates, models of distributed quantum computing, anonymous, identity-oblivious, alpha equivalence, name-invariance, covariance, markovian dynamic graphs, timed quantum circuits, clocked quantum circuits, quantum causal graph dynamics, quantum controllable quantum circuits, quantum FPGA}
}


\begin{document}

\hypersetup{
	colorlinks,
	linkcolor={black!30!blue},
	citecolor={black!30!blue},
	urlcolor={black!30!blue}
}

\title{Addressable quantum gates}

\author[P. Arrighi]{Pablo Arrighi}
\email{pablo.arrighi@universite-paris-saclay.fr}
\affiliation{%
  \institution{Université Paris-Saclay, INRIA, CNRS, LMF}
  \city{Gif-sur-Yvette}
  \country{France}
  \postcode{91190}
}
\affiliation{%
  \institution{IXXI, Lyon}
  \city{Lyon}
  \country{France}
}

\author[C. Cedzich]{Christopher Cedzich}\affiliation{Quantum Technology Group, Heinrich Heine Universit\"at D\"usseldorf, Universit\"atsstr. 1, 40225 D\"usseldorf \country{Germany}}\email{cedzich@hhu.de}{}{}
\author[M. Costes]{Marin Costes}
\email{marin.costes@ens-paris-saclay.fr}{}{}{}
\author[U. Rémond]{Ulysse Rémond}
\email{ulysse.remond@ens-paris-saclay.fr}{}{}{}
\author[B. Valiron]{Beno\^it Valiron}\email{benoit.valiron@monoidal.net}{}{}{}
\affiliation{ Université Paris-Saclay,  CNRS, Centrale Supélec, LMF, 91190 Gif-sur-Yvette \country{France}}

\date{\today}	

\begin{abstract}
We extend the circuit model of quantum computation so that the wiring between gates is soft-coded within registers inside the gates. The addresses in these registers can be manipulated and put into superpositions. This aims at capturing indefinite causal orders and making their geometrical layout explicit: we express the quantum switch and the polarizing beam-splitter within the model. In this context, our main contribution is a full characterization of the anonymity constraints. Indeed, the names used as addresses should not matter beyond the wiring they describe, i.e. quantum evolutions should commute with ``renamings''. We show that these quantum evolutions can still act non-trivially upon the names. We specify the structure of ``nameblind'' matrices.
\end{abstract}

\maketitle

\section{Introduction}
\indent The standard paradigm for quantum computation is the coprocessor model. In this model, quantum evolution is controlled by a purely classical device ---a conventional computer. Quantum computation is described as the list of elementary instructions ---the quantum gates--- sent to the coprocessor: the so-called quantum circuit. This representation has long been thought of as the {most} viable model of quantum computation, and it has been successful in realizing sizeable complexity gains on numerous useful algorithms.

Several other models of quantum computation have been designed to offer other quantum computing possibilities, compared to the usual circuits (wire/gate) vision ---notably : one-way computing \cite{1wqc}, quantum walks \cite{Kempe}, adiabatic computing \cite{AdiabaticQC}, hybrid models and many more, some of which already proved their practical use time and time again.

However, even when sticking to the wire/gate point of view, one quickly notices that in the coprocessor model only the data is quantum. The control flow, i.e., the order in which the gates are to be applied, is classically determined and definite.
In other words, the wiring between the gates is fixed and, albeit quantum, the data flows through the circuit in a definite, classical manner. 

Quantum mechanics allows for more: in \cite{ValironQSwitch} it was argued by constructing a new elementary circuit, the so-called ``quantum switch'', that classically ordered gates are not the only possible paradigm of quantum computation. Instead, the quantum switch behaves like a \emph{quantum} test: given a qubit $q$ and a single instance of gates $U$ and $V$, the operation $\textsc{Switch}(q)(U)(V)$ realizes
\begin{equation}\label{eq:qswitch}
\begin{array}{l}
\ket{0}_q\otimes\ket{\psi}\longmapsto \ket{0}_q\otimes(UV\ket{\psi}),
\\
\ket{1}_q\otimes\ket{\psi}\longmapsto \ket{1}_q\otimes(VU\ket{\psi}).
\end{array}
\end{equation}

\indent Adding the quantum switch to a circuit renders the orders of gates \emph{indefinite} \cite{ValironQSwitch,OreshkovQuantumOrder,Goswami_2018}: the order in which $U$ and
$V$ are applied depends on the state of the ``control'' qubit $q$. Hence, adding $\textsc{Switch}$ enables data \emph{and} computation to be in
superposition. Notably, the quantum switch is not realizable in the coprocessor model, when only one instance of $U$ and $V$ is available~ \cite{ValironQSwitch}. Yet, it was shown to be realizable in experiments \cite{QswitchExp1,QswitchExp2,exp-N-switch,FirstExpICO}. 
Moreover, it allows for gains in terms of communication \cite{guerinExponentialCommunicationComplexity2016,Chiribella_perfectQCommunication,experimentalSuperiorComplexity} and algorithmic complexity compared to textbook circuits~ \cite{QswitchCompAdv}.

In the recent years, several attempts have been made to encompass indefinite causal orders within a consistent model of quantum computing.
A considerable strand of works focuses on the compositionality of indefinite causal orders, via sums of diagrams in the process theory of quantum mechanics \cite{2020_diagrammatic_ICO}, or higher-order circuits 
featuring type annotations for that specific purpose \cite{HigherOrder}, some for routing per subspace \cite{RoutedQC}, other to enforce the time-ordering of packets \cite{CausalBoxes}.
In all these proposals however, the flow of information may be superposed, but it remains static. A more recent strand of works works by distinguishing data qubits from external control qubits \cite{wechs2021quantum,Quantum-Shannon-Theory}. The latter are used to encode the order in which gates are to be applied upon the former. Because the control qubits can be manipulated coherently and conditionally, the flow of information is both dynamical and quantum, i.e. ``second-quantised'' in the language of \cite{Quantum-Shannon-Theory}. 

In this paper we introduce a new model of quantum computation, not by extending textbook quantum circuits in a ad hoc manner (e.g. adding quantum switches \cite{QswitchCompAdv}, polarizing beam splitters \cite{PBScalculus,PBScalculus:coh} etc.) but by introducing a single new principle encompassing all of these constructs and hopefully more. The basic idea is that gates are no longer hardwired to one another. Instead, the wiring is soft-coded: each gate contains a register storing addresses of the gates it targets. We therefore call such gates \emph{Addressable Quantum Gates} (AQG). Those address registers are quantum, hence target addresses can be superposed, and so output data may be sent to different places, in a superposition. Moreover, addresses can interact with the data and be entangled with it. Thereby, addresses may be sent to other gates in an \emph{Addressable Quantum Circuit} (AQC), which naturally leads to evolving superpositions of causal orders. That is, in contrast to the models in \cite{2020_diagrammatic_ICO,HigherOrder,RoutedQC,CausalBoxes} that also encompass indefinite causal orders, our model allows for quantum dynamic routing. Furthermore, compared to \cite{wechs2021quantum} we allow for vacuum states and the increased compositionality that follows. Moreover, we express the geometry not through external control qubits, but via address registers placed directly within each gates, which allows for a clearer geometrical layout, in the spirit of process algebras.

Besides naturally encompassing indefinite causal orders, the AQC model also extends distributed quantum computing in several aspects. Historically, the search for distributed models of quantum computing has mainly focused on static networks and on classically evolving networks \cite{CQP-Gay}. Recently, the notion of a Quantum Causal Graph Dynamics (QCGD) \cite{ArrighiQCGD} has borne the premises of a fully quantum distributed model of computation, where the connectivity of the network is allowed to evolve coherently and find itself in a superposition. The corresponding mathematical framework \cite{ArrighiQNT}, however, is rather abstract and axiomatic. A model of computation should be constructive. This is the case of the hands-on, bottom up approach presented here: we provide a concrete, composable model of distributed quantum computation featuring local synchronous updates as in the \textsc{Local} model, as well as quantum evolving connectivity in the spirit of classical markovian graph dynamics models \cite{MarkovianDynGraphs}.

Classically it is well-known that the computing power of distributed algorithms crucially depends on the availability of a unique identifier at each node \cite{FraigniaudOblivious}. From a more physical point of view, we may indeed consider that these identifiers, a.k.a., ``names'' are fiducial, serving just to describe the network or laying out a coordinate system upon it, when really only the geometry should matter. The limitations brought by this requirement (cf. `anonymous networks' \cite{AnonymousNetworks}, `identifier-oblivious' algorithms \cite{FraigniaudOblivious0} in distributed computing, `renaming-invariance' \cite{ArrighiNamesInQG} or `covariance' in the Physics literature) are thoroughly studied here in the context of our model.

We remark that the questions of 1) the actual gain over the textbook quantum circuit model brought by indefinite causal orders in general and the quantum switch in particular, and 2) whether the current physical implementations of the quantum switch actually realize this gain, is subject of ongoing scientific discussions. Perhaps one of the strongest attempts to settle this debate is the notion of the time-delocalized system implementation of indefinite causal order developed in \cite{OreshkovTD1,OreshkovTD2}. To some extent, the controversy surrounding this ``gain'' has to do with how the gain is defined and quantified. For instance, we could be counting the number of ``usages'' of gates as in standard oracle complexity. Likewise, we could be counting the number of ``physical implementations" of gates, supposing that each of them requires much experimental care. It is not obvious that the quantum switch, at least in its current experimental implementations, provides any ``usage gain''. It does if we only count as ``usage'' the number of times the gate is traversed by data, but it is not clear why one should not count those times it is traversed by the ``vacuum state'' \cite{Quantum-Shannon-Theory}. On the other hand, it seems clear that, contrary to standard quantum circuits, these experiments factor out the number of physical implementations of gates. In that sense, our model provides some clarification, because each physical implementation of a gate is in one-to-one correspondence with one addressable quantum gate: the reusing of physically implemented resources is made explicit. We leave for future work question whether AQC implementations of indefinite causal orders correspond to (or help express) time-delocalized implementations.

Our contributions are the following:
\begin{itemize}[noitemsep,topsep=0pt,parsep=0pt,partopsep=0pt]
\item A novel computational paradigm ---called \emph{Addressable Quantum Circuits}--- generalizing textbook quantum circuits, so as to make the flow of information quantum (Sect. \ref{Section:model}).

\item Showing that this quantum flow of information allows for indefinite causal order. We demonstrate this below by encoding the quantum switch and the polarizing beam splitter (Sect. \ref{whatcando}).

\item A restriction of the model so that the actual names used as addresses do not matter beyond the wiring they describe, i.e. `disallowing pointer arithmetic' in the manner of Java...

\item ...which led us to solve the following linear algebra question as a natural byproduct: say that a matrix acts on the Hilbert space spanned by $\ket{abc...}$, $\ket{acb...}$, $\ket{bac...}$ in a way that commutes with renaming $a$ into $b$ etc. What is this matrix allowed to do? Can it still create superpositions? We provide a full, constructive characterization of such ``nameblind'' matrices. Whilst strongly structured, they still admit non-trivial behaviors. It is our main technical contribution (Sect.\,\ref{sec:renamings}).

\item The compositionality of the model is formalized (Appendix \ref{annex:operations}). The physicality of the model is discussed with comparison with other frameworks (Section \ref{annex:related}).
\end{itemize}

\noindent Notations are summarized in Appendix \ref{annex:notations}.

\section{The model}\label{Section:model}

\indent Quantum algorithms, whether they are distributed or not, are consist of two steps. Firstly, a quantum circuit is operated on the input quantum data --- possibly in an iterative fashion. In the second step, the resulting quantum state is measured. \emph{Addressable Quantum Circuits} (AQCs) represent a novel description of the first part. More precisely, each iteration of an AQC is one step of propagation of quantum data in the circuit, where gates can be applied in a superposition. Such physical locations in the laboratory are abstracted into \emph{sectors} in the AQC formalism.
 
Each sector is uniquely identified by an \emph{address} taken in a finite set \calA ---see Def. \ref{def:addresses}. 
Several sectors can be grouped together to form \emph{gates}, which determine the level of granularity at which ``local'' unitary \emph{gate operators} are applied. 
These concepts are formally introduced in this section. To establish familiarity, we encode the usual Bell state creation circuit as an AQC. This requires four sectors with addresses $\calA = \{ 1 , 2 , 3, 4 \}$ ---see Fig. \ref{fig:maxeprsect}. 

\begin{figure}[tbh]
    \centering
    \label{fig:epr2}
    \begin{subfigure}[b]{0.6\textwidth}
    \centering
    \resizebox{0.9\textwidth}{!}{

\begin{tikzpicture}[thick]
    %
    \tikzstyle{operator} = [draw,fill=white,minimum size=1.5em] 
    \tikzstyle{phase} = [fill,shape=circle,minimum size=5pt,inner sep=0pt]
    \tikzstyle{oplus} = [draw,fill=white,shape=circle,minimum size=10pt,inner sep=0pt]
    \tikzstyle{plus} = [draw,fill=white,shape=circle,minimum size=10pt,inner sep=0pt]
    \tikzstyle{surround} = [fill=blue!10,thick,draw=black,rounded corners=2mm]
    %
    \node at (0,0) (q1) {$\ket{0}$};
    \node at (0,-1) (q2) {$\ket{0}$};
    %
    \rectsec{1}{2};
    \node[operator] (op11) at (2,0) {H} edge [-] (q1);
    \rectsec{2}{2};
    %
    \rectsec{3}{2};
    \node[phase] (phase11) at (3,0) {} edge [-] (op11);
    \node[oplus] (phase12) at (3,-1) {} edge [-] (q2);
    \node[cross=4pt,rotate=45] (toto) at (3,-1) {} edge [-] (q2);
    \draw[-] (phase11) -- (phase12);
    \rectsec{4}{2}{num4};
    \node (end1) at (5,0) {} edge [-] (phase11);
    \node (end2) at (5,-1) {} edge [-] (phase12);
    %
    \draw[decorate,decoration={brace},thick] (5,0.2) to
	node[midway,right] (bracket) {$\frac{\ket{00}+\ket{11}}{\sqrt{2}}$}
	(5,-1.2);
    %
    \begin{pgfonlayer}{background} 
    \node[surround] (background) [fit = (q2) (num4) (bracket)] {};
    \end{pgfonlayer}

\end{tikzpicture}}
    \caption{Bell state creation as a textbook quantum circuit}
    \label{fig:maxeprsect}
    \end{subfigure}
    \begin{subfigure}[b]{0.6\textwidth}
    \centering
    \resizebox{0.6\textwidth}{!}{

\begin{tikzpicture}[thick]
    %
    \tikzstyle{operator} = [draw,fill=white,minimum size=1.5em] 
    \tikzstyle{phase} = [fill,shape=circle,minimum size=5pt,inner sep=0pt]
    \tikzstyle{oplus} = [draw,fill=white,shape=circle,minimum size=10pt,inner sep=0pt]
    %
    %
    \rectpos{1.5}{2}{1}{1}{num1}{rec1};
    \node[operator] (op01) at (1.5,-1) {I};
    \node[operator] (op02) at (1.5,-2) {I};
    
    \node[operator] (op11) at (0,0) {H};
    \node[operator] (op12) at (0,-1) {I};
    \rectpos{0}{2}{0}{2}{num2}{rec2};
    %
    \rectpos{2.7}{2}{0}{3}{num3}{rec3};
    \node[phase] (phase11) at (2.7,0) {} ;
    \node[oplus] (phase12) at (2.7,-1) {} ;
    \node[cross=4pt,rotate=45] (toto) at (2.7,-1) {} ;
    \draw[-] (phase11) -- (phase12);
    \rectpos{4.5}{2}{0.5}{4}{num4}{rec4};
    \node[operator] (op41) at (4.5,-0.5) {I};
    \node[operator] (op42) at (4.5,-1.5) {I};
    \node (end1) at (5,0) {} ;
    \node (end2) at (5,-1) {} ;
    \flechecontrol{op02}{rec2}{-1.3}{-2.3}{black};
    \flechecontrol{rec2r}{rec3}{1.5}{0}{black};
    \flechecontrol{rec3r}{rec4}{3.3}{-0.3}{black};
\end{tikzpicture}}
    \caption{Bell state creation as an addressable quantum circuit
    }
    \label{fig:epr3}
    \end{subfigure}
    \caption{Dashed boxes represent sectors.  
    (a) Textbook quantum circuits require a background layer that implies a time ordering, here shown explicitly in light purple. (b) The AQC formalism frees the information flow from the background surface. Each sector has a target space, holding a target address, as represented by the arrows. Addresses could have been taken in any other integer set instead.}
    \label{fig:epr}
\end{figure}
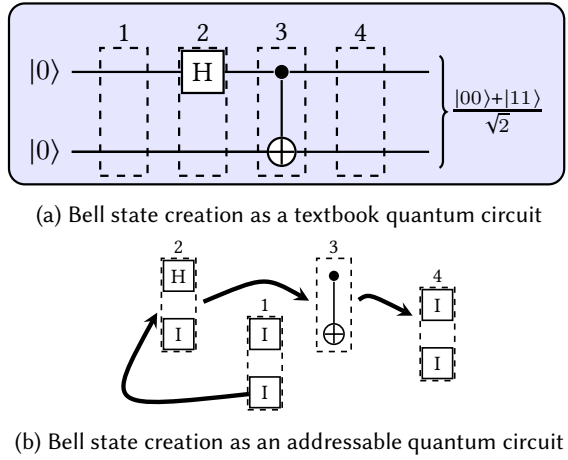

\subsection{Addresses and sectors} \label{sect:address}

In the following $\calH_\calX$ denotes the Hilbert space whose canonical orthonormal basis is $\{\ket{e}\}_{e\in \calX} $. To each address we associate a sector, i.e. a connectable device which stores quantum data. 
The main novel feature of our formalism is its ability to manipulate and superpose addresses, and thus the connectivity between sectors. To this end, each sector has a \emph{target address space} $\calH_\calT$ to determine where the output data is to be sent. Additionally, it can store other addresses in its \emph{stored address spaces} $\calH_\calW$. By locally acting on these spaces we may shuffle addresses, possibly into a superposition.

\begin{definition}[Target address space, stored address space]\label{def:targetstored}
Let the set of addresses $\calA$ be a finite set of integers. A \emph{target address space} is a Hilbert space $\calH_\calT$, where $\calT \equiv \{a|a\in \calA^{\leq1}\} = A^?$ is the set of words on $\calA$ of length at most $1$. In contrast, basis states of a \emph{stored address space} $\calH_\calW$ can contain multiple addresses, but each address at most once. Thus, $\calH_\calW$ is a Hilbert space specified by $\calW \equiv \{\ovva = a_1\cdots a_n|a_i\in\calA,a_i\neq a_j\:\forall i,j\} = \calA^-$, the set of words on $\calA$ with all different letters.
\end{definition}

\begin{example}\label{example:ket32}
    With $\calA =  \{2,3\}$, the vector 
$ (\ket{2} +e^{-i\frac{ \pi}{8}} \ket{3} +\ket{3 \hspace{5pt} 2}+\ket{2 \hspace{5pt} 3})/2$
is a valid state of $\calH_{\calW}$.
\end{example}
For the Bell state creation circuit, just like for any textbook quantum circuit encoding, the stored address spaces are unused and the target address spaces remain fixed.
Specifying the target addresses for each sector makes the underlying background layer redundant: we might just as well dismantle the circuit as in Fig. \ref{fig:epr3}. In other examples with dynamically evolving wirings, the circuit may even be in a superposition of several arrow configurations ---see e.g. Fig. \ref{fig:QSAQC}. Such superpositions of wirings induce superpositions of orders of applications of gates on the quantum data stored in $\calH_\calQ$:

\begin{definition}[Data space] \label{def:dataspace}
Let ${\cal D}$ be a finite set of possible {\em data values}. The set of words on data of length at most $n$ is denoted by ${\cal Q} = {\cal D}^{\leq n}$. Since ${\cal D}$ is finite, so is ${\cal Q}$. The Hilbert space $\calH_{{\cal Q}}$ is referred to as a {\em data space}.
\end{definition}

Each sector contains two storage spaces : 
one for the input and one for the output, and both can store addresses and data. Each of these subspaces may be empty, in which case we denote its state by $\eps$. Sectors with empty data spaces correspond to vacuum states \cite{Quantum-Shannon-Theory}. The purpose of such output and input spaces is explicitated in Sec. \ref{evolution}.

\begin{definition}[Input and output space, sector space]\label{def:sectors}
An \emph{input (resp. output) space} is a Hilbert space $\calH_\calI$ (resp. $\calH_\calO$) of the form $\mathcal{H}_{\calW}\otimes \calH_{{\calQ}}$.
A \emph{sector space} is a Hilbert space of the form
$\calH_\calS = \calH_\calT\otimes \calH_\calI\otimes \calH_\calO$. 
To each address $a\in\calA$ is associated a sector space $\calH_\calS^a$.
\end{definition}

We write basis states of input (resp. output) spaces as $\ket{\overline{a},\overline{q}}_{\calI}$ (resp. $\ket{\overline{a},\overline{q}}_{\calO}$), where $\overline{a}\in\calW$  and $\overline{q}\in \calQ$. Basis states of a sector may be expressed as $\ket{t}_\calT\ket{\ovva}_\calI\ket\ovvq_\calI\ket{\ovvap}_\calO\ket\ovvqp_\calO  \in \calH_\calT \otimes \calH_{\calW_I} \otimes \calH_{\calQ_I} \otimes \calH_{\calW_O} \otimes \calH_{\calQ_O}$ or, by reordering the tensor factors as $\ket{t}_\calT\ket{\ovva}_\calI\ket{\ovvap}_\calO\ket\ovvq_I\ket\ovvqp_\calO \in \calH_\calT \otimes \calH_{\calW_I}\otimes \calH_{\calW_O} \otimes \calH_{\calQ_I} \otimes \calH_{\calQ_O}$ ---see Fig \ref{fig:sectors}.

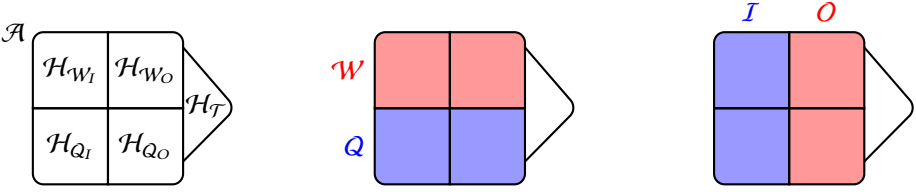
\begin{figure}[ht]
    \centering
    \begin{subfigure}[b]{0.3\textwidth}
    \centering
    {\begin{tikzpicture}[thick]
\sectorIOWQ{0}{0}{$\calH_{\calQ_I}$}{$\calH_{\calQ_O}$}{$\calH_{\calW_I}$}{$\calH_{\calW_O}$}{$\calH_\calT$}{$\calA$}
\end{tikzpicture}}
    \label{fig:sector}
    \end{subfigure}
    \begin{subfigure}[b]{0.33\textwidth}
    \centering
    {\newcommand{\rectan}[3]{
\draw[color=#3,line width=1mm] (0.05,#1+0.05) rectangle (1.95,#1+0.95) node at (-0.26,#1+0.50) {#2};
}

\begin{tikzpicture}[thick]
\fill[color= blue, opacity= .4] (0,0) --  node[opacity=1,left] {$\calQ$} (0,1) -- (2,1)[rounded corners] -- (2,0)  -- cycle;
\fill[color= red, opacity= .4] (0,1)[rounded corners] --  node[opacity=1,left] {$\calW$} (0,2) -- (2,2)[sharp corners] -- (2,1)  -- cycle;

\sectorIOWQ{0}{0}{}{}{}{}{}{} 
\end{tikzpicture}}
    \label{fig:sectorAD}
    \end{subfigure}
    \begin{subfigure}[b]{0.33\textwidth}
    \centering
    {\begin{tikzpicture}[thick]
\fill[color= blue, opacity= .4] (0,0)[rounded corners] --  (0,2)[sharp corners] -- node[opacity=1,above] {$\calI$} (1,2) -- (1,0)[rounded corners]  -- cycle;
\fill[color= red, opacity= .4] (1,0) -- (1,2)[rounded corners] -- node[opacity=1,above] {$\calO$}  (2,2) -- (2,0)[sharp corners]  -- cycle;
\sectorIOWQ{0}{0}{}{}{}{}{}{}
\end{tikzpicture}}
    \label{fig:sectorIO}
    \end{subfigure}
    \caption{ (a) A sector. \hspace{10pt} (b) Stored addresses vs. quantum data. \hspace{15pt} (c) Input vs. output spaces.\\[5pt]
    When all \textbf{output} (resp. input) spaces are empty, we only represent the \textbf{input} (resp. output) spaces, as in Fig. \ref{fig:QSAQC}.}
    \label{fig:sectors} 
\end{figure}

\begin{example}\label{example:AQorthog}
With $\calA = \{ 1 , 2 , 3, 4 \}$ as in Fig. \ref{fig:maxeprsect} and $\mathcal Q = \{00,01,10,11,\eps\}$, the states $\ket{\eps,10}$, $\ket{3,\eps}$ and $\ket{23,11}$ are basis states for $\calH_\calI$ and $\calH_\calO$. E.g. $\ket{\eps,11} + \ket{2,\eps}$ and $\ket{2,11}$ are orthogonal.

Moreover, $\ket{3}_\calT\ket{2,01}_\calI\ket{\eps,\eps}_\calO$ is a basis state of $\calH_S$. 
\end{example}

The state of the quantum system at any point of its evolution is described by a sum of tensor products of states of each of its sectors. E.g. in the beginning of the execution of the Bell state creation circuit ---see Fig. \ref{fig:maxeprsect}, the quantum system is in the state $\ket{\psi_0}$:
\begin{equation}\label{eq:psi0}
\begin{array}{rrrclclc}
     \ket{\psi_0} =&&\ket{2}^1_\calT & \otimes & \ket{\eps^a,00^d}^1_\calI & \otimes & \ket{\eps^a,\eps^d}^1_O  & \hspace{2em}\text{(sector 1)} \\
     &\otimes & \ket{3}^2_\calT & \otimes & \ket{\eps^a,\eps^d}^2_\calI & \otimes & \ket{\eps^a,\eps^d}^2_O  &  \hspace{2em}\text{(sector 2)} \\
    & \otimes & \ket{4}^3_\calT & \otimes&  \ket{\eps^a,\eps^d}^3_\calI & \otimes & \ket{\eps^a,\eps^d}^3_O & \hspace{2em}\text{(sector 3)}  \\
     &\otimes & \ket{\eps}^4_\calT & \otimes & \ket{\eps^a,\eps^d}^4_\calI  & \otimes & \ket{\eps^a,\eps^d}^4_O & \hspace{2em}\text{(sector 4)}   
\end{array}\\
\end{equation}

\subsection{Addressable Gates}

A \emph{gate} is composed of multiple sectors, and thus can receive information from and send information to multiple sources and destinations. This grouping of sectors into gates defines a notion of locality for the evolution: during a ``scattering step'' a gate acts upon its constituent sectors with a local unitary.

\begin{definition}[Gates, Gate Operators]\label{def:gates}
A \emph{gate} $g$ is a (non-empty) subset of a set of addresses \calA, and the \emph{set of gates} $\calG$ is a partition of $\calA$. The \emph{gate space} of a gate $g$ consists of $|g|$ sectors and is given by $\calH^g=\bigotimes_{a\in g} \calH^a_\calS$. 
A \emph{gate operator} on $g\in\mathcal G$ is a unitary $S^g:\calH^g\to\calH^g$ which 
preserves the addresses of the target and stored spaces, i.e. 
$\bra{\ovvap\ \ovvqp}S^g\ket{\ovva\ \ovvq}\neq 0$ implies that $\ovvap$  has the same letters as $\ovva$, where $\ovva$ and $\ovvap$ (resp. $\ovvq, \ovvqp$) list all target and stored addresses (resp. all data) in sectors of $g$.
\end{definition}

In general, gates act as the identity on the empty state $\eps$. The gate operators of the Bell state creation circuit are given by
\begin{equation}\label{eq:eprgateoperator}S^{\{1\}} = F = S^{\{4\}}, 
     \hspace{20pt}  S^{\{2\}} = F  ((H\otimes I)_{\calI} \otimes {I}_{\calO}), 
     \hspace{20pt}  S^{\{3\}} = F  ({\textsc{Cnot}}_{\calI} \otimes {I}_\calO),  
\end{equation}
     where $F$ swaps the input and output systems. Note that neither data nor addresses are manipulated in sectors $1$ and $4$, yet they are required to compose this AQC with other AQCs ---see Sect. \ref{annex:operations}.

The overall state space of an AQC is the tensor product of its sector spaces together with an additional global condition: to guarantee the reversibility of the transport step that will be defined below, we require that each address appears at most once.

\begin{definition}[Skeleton, Addressable Quantum Circuit]\label{def:AQC}
A \emph{skeleton} $K$ is a tuple $(\calA,\calG,\calQ,S)$ where:
\begin{itemize}
    \item $\calA$ is the set of addresses, 
    \item $\calG$ is the set of gates,
    \item $\calQ$ is the set of data ---consisting of words on data of length at most $n$,
    \item $S$ is a function associating a gate operator $S^g$ to each gate $g \in \calG$.
\end{itemize}
This specifies the sector spaces as $\calH_\calS= \calH_\calT\otimes\calH_{\calW}\otimes\calH_{\calQ}\otimes\calH_{\calW}\otimes\calH_{\calQ}$.\\
The \emph{circuit Hilbert space} $\calH$ is defined as $\bigotimes_{a \in \calA} \calH_\calS^a$ restricted to superpositions of basis states where addresses appear at most once.
An \emph{Addressable Quantum Circuit} is a tuple ($K$,$\ket{\psi}$) that specifies both the dynamics and the current state, i.e. $K$ is a skeleton and $\ket\psi$ is a state of $\calH$.
\end{definition}

For example, consider the skeleton of the Bell state creation circuit shown in Fig. \ref{fig:epr3}. It is specified by $\calA = \{ 1 , 2 , 3, 4 \}$, $\calG = \{ \{1\},\{2\},\{3\},\{4\}\}$, $\calQ$ induced by $\calD = \{0,1\}$ and $n = 2$, and $S$ as in \eqref{eq:eprgateoperator}. In Ex. \ref{example:bell_exec} below we choose the initial state of \eqref{eq:psi0}. Reordering tensor factors as above, we obtain $\ket{\psi_0}=\ket{2,3,4,\eps}_\calT \ket{{00},\eps,\eps,\eps}_\calI \ket{\eps,\eps,\eps,\eps}_\calO $.

\subsection{Evolution}\label{evolution}
The overall evolution alternates a scattering step during which the gate operators are applied to their respective gates, and a transport step that swaps the output space of each sector with the input space of the sector it targets. The evolution is unitary i.e. it does not involve measurements.

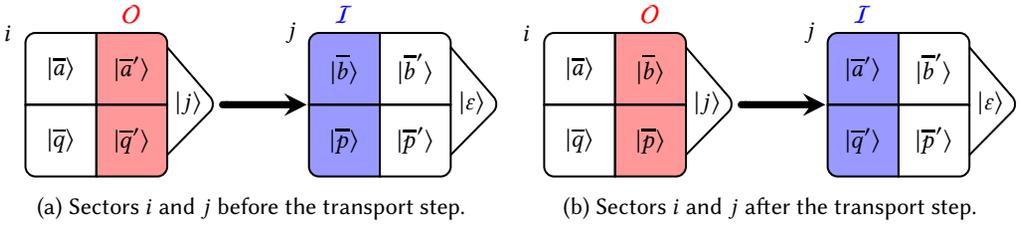
\begin{figure}[htb]
    \centering
    \begin{subfigure}[b]{0.49\textwidth}
    \centering
    \resizebox{\textwidth}{!}
    {\begin{tikzpicture}[thick]
\fill[color= red, opacity= .4] (1,0) -- (1,2)[rounded corners] -- node[opacity=1,above] {$\calO$}  (2,2) -- (2,0)[sharp corners]  -- cycle;
\sectorIOWQ[i]{0}{0}{$\ket \ovvq$}{$\ket \ovvqp$}{$\ket \ovva$}{$\ket \ovvap$}{$\ket j$}{$i$}
\fill[color= blue, opacity= .4] (4,0)[rounded corners] --  (4,2)[sharp corners] -- node[opacity=1,above] {$\calI$} (5,2) -- (5,0)[rounded corners]  -- cycle;
\sectorIOWQ[j]{4}{0}{$\ket \ovvp$}{$\ket \ovvpp$}{$\ket \ovvb$}{$\ket \ovvbp$}{$\ket \eps$}{$j$}
\fleche{outarrowi}{inarrowj}{black}
\end{tikzpicture}}
    \caption{Sectors $i$ and $j$ before the transport step.}
    \label{fig:sectorEvolv1}
    \end{subfigure}
    \begin{subfigure}[b]{0.49\textwidth}
    \centering
    \resizebox{\textwidth}{!}
    {\begin{tikzpicture}[thick]
\fill[color= red, opacity= .4] (1,0) -- (1,2)[rounded corners] -- node[opacity=1,above] {$\calO$}  (2,2) -- (2,0)[sharp corners]  -- cycle;
\sectorIOWQ[i]{0}{0}{$\ket \ovvq$}{$\ket \ovvp$}{$\ket \ovva$}{$\ket \ovvb$}{$\ket j$}{$i$}
\fill[color= blue, opacity= .4] (4,0)[rounded corners] --  (4,2)[sharp corners] -- node[opacity=1,above] {$\calI$} (5,2) -- (5,0)[rounded corners]  -- cycle;
\sectorIOWQ[j]{4}{0}{$\ket \ovvqp$}{$\ket \ovvpp$}{$\ket \ovvap$}{$\ket \ovvbp$}{$\ket \eps$}{$j$}
\fleche{outarrowi}{inarrowj}{black}
\end{tikzpicture}}
    \caption{Sectors $i$ and $j$ after the transport step.}
    \label{fig:sectorEvolv2}
    \end{subfigure}
    \caption{Effect of the transport step upon two connected sectors ---see example \ref{example:transport}. Here and elsewhere in the paper arrows represent the content of each target space. If the state of the target space is a superposition of multiple addresses, one arrow points to each target.}
    \label{fig:transp-explanation}
\end{figure}

\begin{definition}[Evolution]\label{def:Evolution}
An addressable quantum circuit $((\calA,\calG,\calQ,S),\ket{\psi})$ evolves in one time step into the addressable quantum circuit $((\calA,\calG,\calQ,S),\ket{\psi'})$, where $\ket{\psi'}=G\ket{\psi}$ with $G=TS$ and:
\begin{description}
\item[Scattering $S$] is the application of $\bigotimes_g S^g$, i.e. the gate operators $S^g$ are applied simultaneously on their corresponding gate spaces.
\item[Transport $T$] is a permutation of basis states, extended linearly to $\calH$. 
The permutation maps synchronously each triple of the form $\ket{j}^i_\calT \ket{\overline x}^i_\calO \ket{\overline y}^j_\calI$ to $\ket{j}^i_\calT \ket{\overline y}^i_\calO \ket{\overline x}^j_\calI$, leaving the rest (of the basis state) unchanged. 
That is, the transport step maps a basis state $\ket{j}^i_\calT \ket{\overline y_i}^i_\calI\ket{\overline x_i}^i_\calO\otimes\ket{\alpha}^j_\calT \ket{\overline y_j}^j_\calI\ket{\overline x_j}^j_\calO\otimes\bigotimes_{k\not\in\{i,j\}}\ket\psi^k$ to $\ket{j}^i_\calT \ket{\overline y_i}^i_\calI\ket{\overline y_j}^i_\calO\otimes\ket{\alpha}^j_\calT \ket{\overline x_i}^j_\calI\ket{\overline x_j}^j_\calO\otimes\bigotimes_{k\not\in\{i,j\}}\ket\psi^k$. 
Colloquially, when Sector $i$ targets Sector $j$, $i$'s output space is swapped with $j$'s input space ---see Fig.~\ref{fig:transp-explanation}.
\end{description}
\end{definition}
The necessity of the requirement that each address appears at most once becomes clear now: it is essential for the transport step to be well-defined by avoiding that a sector is targeted by more than one other sector, as the output space of a sector $i$ are swapped with the input space of the sector targeted by $i$. 
The addresses stored in output and input spaces, referencing some future sector that would be addressed later, flow throughout the system just as the quantum data does. It is only required to store in input or output spaces the addresses of the sectors that will be addressed by a quantum superposition of different sectors.

\begin{example}\label{example:bell_exec}
Let us describe step-by-step the Bell state creation circuit of Fig. \ref{fig:epr}. We drop the stored address space which are always empty. In Ex. \ref{example:epr:ts} we further break down the evolution into $S$ and $T$. The blue and red terms (and, later, arrows ---see Fig. \ref{fig:qsaqc})
are superposed, and get entangled:
\begin{align*}
\textstyle\ket{2,3,4,\eps}_\calT &\textstyle\otimes \ket{{00},\eps,\eps,\eps}_\calI \otimes \ket{\eps,\eps,\eps,\eps}_\calO \vphantom{\frac{1}{\sqrt{2}}}\\ 
&\textstyle\longrightarrow_G\quad\ket{2,3,4,\eps}_\calT \otimes \ket{\eps,{00},\eps,\eps}_\calI  \otimes \ket{\eps,\eps,\eps,\eps}_\calO \vphantom{\frac{1}{\sqrt{2}}}\\ 
&\textstyle\longrightarrow_G\quad \ket{2,3,4,\eps}_\calT\otimes\frac{1}{\sqrt{2}}\left(\ket{\eps,\eps,{\color{red} 00},\eps}_\calI + \ket{\eps,\eps,{\color{blue} 10},\eps}_\calI \right) \otimes \ket{\eps,\eps,\eps,\eps}_\calO\\ 
&\textstyle\longrightarrow_G\quad \ket{2,3,4,\eps}_\calT \otimes \frac{1}{\sqrt2}\left(\ket{\eps,\eps,\eps,{\color{red} 00}}_\calI + \ket{\eps,\eps,\eps,{\color{blue} 11}}_\calI\right) \otimes \ket{\eps,\eps,\eps,\eps}_\calO \\ 
&\textstyle\longrightarrow_G\quad \ket{2,3,4,\eps}_\calT \otimes \ket{\eps,\eps,\eps,\eps}_\calI \otimes \frac1 {\sqrt2}\left(\ket{\eps,\eps,\eps,{\color{red} 00}}_\calO + \ket{\eps,\eps,\eps,{\color{blue} 11}}_\calO\right)\tag{$\ast$}\label{eq:step5}\\
&\textstyle\longrightarrow_G\quad \ket{2,3,4,\eps}_\calT \otimes \ket{\eps,\eps,\eps,\eps}_\calI \otimes \frac1{\sqrt2}\left(\ket{\eps,\eps,{\color{red} 00},\eps}_\calO + \ket{\eps,\eps,{\color{blue} 11},\eps}_\calO\right)\\
&\cdots&
\end{align*}
The finite execution of usual circuits corresponds to a finite sequence of iterations in the otherwise potentially infinite evolution of AQCs. E.g. in order to match the corresponding textbook circuit, the AQC above needs to be stopped after step \eqref{eq:step5}.
\end{example}

\subsection{Composability of AQC}\label{sec:corecomposability}

The fact that the connectivity of gates is soft-coded in the AQC model, leaves room for several ways of composing AQCs. In Appendix. \ref{annex:operations} we propose formal definitions for:
\begin{description}
\item[Parallel composition] of AQCs, corresponding to executing them in parallel without interaction.
\item[Concatenation] of AQCs, corresponding to successively applying them to the quantum data.
\item[Connection] of AQCs, alllowing for a more general connectivity without assuming an intrinsic (temporal) order between AQCs.
\end{description}
The last of these even allows for indefinite causal order. For now let us show how indefinite causal order may happen within one AQC.

\section{What can it do?} \label{whatcando}

Using the Bell state creation circuit as a guiding example, we illustrated how textbook quantum circuits can be implemented in this framework. We argued that textbook quantum circuits only require hardwired one-sector gates in the sense that the connectivity of the graph is static and definite, i.e. $\calH_{\calT}$ is in a fixed basis state and $\calH_{\calW_I}$ and $\calH_{\calW_O}$ are empty.

As argued in the introduction, there exist physical devices that allow for superpositions of causal orders. These devices lie beyond the scope of hardwired circuits, in that they cannot be described by the textbook circuit framework without extending the workspace or waiving knowledge about the internal workings of the gates by using all-encompassing gates. We now describe how such causal superpositions arise naturally by superposing addresses in AQCs.

\subsection{The quantum switch}\label{sec:qs}

The quantum switch described in \eqref{eq:qswitch} applies unitaries $U$ and $V$ to $m$ qubits in a superposition of causal orders that depends on a control qubit ---see Fig.~\ref{fig:QS}. We can encode \textsc{Switch} as the AQC 
of Fig.~\ref{fig:QSAQC}, where
$\calA = \{ 1 , 2 , 3, 4, 5, 6 \}$, $\calG = \{ \{1\},\{2,5\},\{3\},\{4\},\{6\}\}$, 
and $\calQ= \{0,1\}^{\leq m+1}$ contains $m$ controlled qubits plus one control qubit. The gates $\{3\}$ and $\{4\}$ apply unitary operators $U$ and $V$ to the $m$ controlled qubits, whereas the gate $\{2,5\}$ implements the quantum switch itself.

As initial state we choose $\ket {\psi_0} = \ket{2}_\calT^1 \ket{\overline{q} c}_{\calQ_I}^1  \ket{345}_{\calW_O}^2 \ket{6}_\calT^5 $, where $\ket{c}=\alpha\ket 0 + \beta \ket 1\in\calH_\calD $ denotes the  
control qubit and $\ket{\overline{q}}$ are 
the $m$ controlled qubits. All non specified spaces are empty, i.e. they are in state $\ket\eps$.

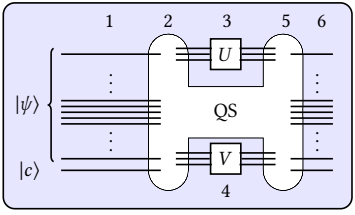
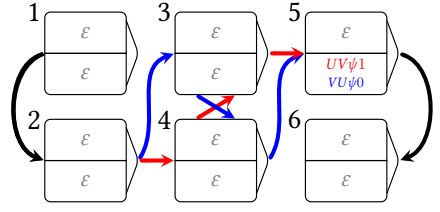
\begin{figure}[htb]
	\savebox{\tbQS}{\resizebox{0.34\textwidth}{!}{\begin{tikzpicture}[thick]
    %
    \tikzstyle{operator} = [draw,fill=white,minimum size=1.5em] 
    \tikzstyle{phase} = [fill,shape=circle,minimum size=5pt,inner sep=0pt]
    \tikzstyle{oplus} = [draw,fill=white,shape=circle,minimum size=10pt,inner sep=0pt]
    \tikzstyle{plus} = [draw,fill=white,shape=circle,minimum size=10pt,inner sep=0pt]
    \tikzstyle{surround} = [fill=blue!10,thick,draw=black,rounded corners=2mm]
    %
    \node at (0,0) (q1) {};
    \node at (1,-0.4) (vdots1) {\vdots};
    \node at (0,-0.8) (q4) {};
    \node at (0,-0.9) (q5) {};
    \node at (0,-1) (q3) {};
    \node at (0,-1.1) (q6) {};
    \node at (0,-1.2) (q7) {};
    \node at (0,-1.8) (q8) {};
    \node at (1,-1.4) (vdots2) {\vdots};
    \node at (0,-2) (q2) {};
    %
    \node[] (sec21) at (2,0) {} edge [-] (q1);
    \node (sec24) at (2,-0.8) {} edge [-] (q4);
    \node (sec25) at (2,-0.9) {}  edge [-] (q5);
    \node[] (sec23) at (2,-1) {} edge [-] (q3);
    \node (sec26) at (2,-1.1) {}  edge [-] (q6);
    \node (sec27) at (2,-1.2) {}  edge [-] (q7);
    \node (sec28) at (2,-1.8) {}  edge [-] (q8);
    \node[] (sec22) at (2,-2) {} edge [-] (q2);
    \node[above = 0.2 of sec21] (num2) {2};
    \node[left = 0.6 of num2] (num1)  {1};
    \node[operator] (sec3) at (3,0) {$U$} edge [] (sec21);
    \node at (3.11,0.1) (sec301) {};
    \node at (3.11,-0.1) (sec302) {};
    \node at (2.89,0.1) (sec303) {};
    \node at (2.89,-0.1) (sec304) {};
    \node[right = 0.6 of num2] (num3)  {3};
    \node[right = 0.6 of num3] (num5)  {5};
    \node[right = 0.2 of num5] (num6)  {6};
    \node[operator] (sec4) at (3,-1.8) {$V$} edge [-] (sec28);
    \node[] (sec401) at (3.11,-1.7) {};
    \node[] (sec402) at (3.11,-1.9) {};
    \node[] (sec403) at (2.89,-1.7) {};
    \node[] (sec404) at (2.89,-1.9) {};
    \node at (3,-2.35) (num4)  {4};
    \node[] (sec5) at (3,-1) {QS};
    %
    \node[] (sec51) at (4,0) {} edge [-] (sec3);
    \node[] (sec5101) at (4,0.1) {} edge [-] (sec301);
    \node[] (sec5102) at (4,-0.1) {} edge [-] (sec302);
    
    \node at (4.6,-0.4) (vdots1) {\vdots};
    \node at (4,-0.8) (sec54) {};
    \node at (4,-0.9) (sec55) {};
    \node[] (sec53) at (4,-1) {};
    \node at (4,-1.1) (sec56) {};
    \node at (4,-1.2) (sec57) {};
    \node at (4.6,-1.4) (vdots2) {\vdots};
    \node at (4,-1.8) (sec58) {}edge [-] (sec4);
    \node[] (sec5801) at (4,-1.7) {} edge [-] (sec401);
    \node[] (sec5802) at (4,-1.9) {} edge [-] (sec402);
    \node[] (sec52) at (4,-2) {} ;
    \node[] (end1) at (5,0) {} edge [-] (sec51);
    \node at (5,-0.8) (end4) {} edge [-] (sec54);
    \node at (5,-0.9) (end5) {} edge [-] (sec55);
    \node[] (end3) at (5,-1) {} edge [-] (sec53);
    \node at (5,-1.1) (end6) {} edge [-] (sec56);
    \node at (5,-1.2) (end7) {} edge [-] (sec57);
    \node at (5,-1.8) (end8) {} edge [-] (sec58);
    \node[] (end2) at (5,-2) {} edge [-] (sec52);
    
    
    \node[] (sec2801) at (2,-1.7) {} edge [-] (sec403);
    \node[] (sec2802) at (2,-1.9) {} edge [-] (sec404);
    \node[] (sec5801) at (2,0.1) {} edge [-] (sec303);
    \node[] (sec5802) at (2,-0.1) {} edge [-] (sec304);
    
    \draw[decorate,decoration={brace,mirror},thick] (0,0.1) to
	node[midway,left] (bracket) {$\ket{\psi}$ \hspace{5pt}}
	(0,-1.85);
    \node[left] (contr) at (0,-2) {$\ket{c}$ \hspace{5pt}};
    %
    \begin{pgfonlayer}{background} 
    \node[surround] (background) [fit = (end1) (end2) (num5) (num4) (bracket),inner sep=2pt] {};
    \end{pgfonlayer}
    %
    
    \drawpolygon sec21, sec22, sec26, sec56, sec52, sec51, sec55, sec25;
\end{tikzpicture}}}
    \begin{subfigure}[t]{0.48\textwidth}
        \centering
        \usebox{\tbQS}
        \caption{Textbook-like representation of the Quantum Switch.}
        \label{fig:QS}
    \end{subfigure}
    \hfill
    \begin{subfigure}[t]{0.48\textwidth}
        \centering
        \raisebox{\dimexpr.5\ht\tbQS-.5\height}{
        \resizebox{0.88\textwidth}{!}{\renewcommand{\eps}{\epsgr}\begin{tikzpicture}[thick]

\sector{0}{1.3}{$\eps$}{1}{$\eps$}{$\eps$}
\sector{0}{-1.3}{$\eps$}{2}{$\eps$}{$\eps$}
\sector{3.2}{1.3}{$\eps$}{3}{$\eps$}{$\eps$}
\sector{3.2}{-1.3}{$\eps$}{4}{$\eps$}{$\eps$}
\sector{6.4}{1.3}{\large
\textcolor{red}{$UV\psi1$}\\ \textcolor{blue}{$VU\psi0$}}
    {5}{$\eps$}{$\eps$}
\sector{6.4}{-1.3}{$\eps$}{6}{$\eps$}{$\eps$}
\draw[-stealth,line width=1mm,bend right=70] (inarrow1) to (inarrow2);
\fleche[bend right=70]{inarrow1}{inarrow2}{black};
\fleche[out=60, in= 200]{outarrow2}{inarrow3}{blue}
\fleche{outarrow2}{inarrow4}{red}
\fleche{topleft4}{bottomright3}{red}	
\fleche{bottomleft3}{topright4}{blue}
\fleche[out=60, in= 200]{outarrow4}{inarrow5}{blue}
\fleche{outarrow3}{inarrow5}{red}
\fleche[bend left=70]{outarrow5}{outarrow6}{black}
\end{tikzpicture}\renewcommand{\eps}{\epsgr}}
    }
        \caption{AQC representation of the Quantum Switch during its execution. Red and blue data and targets are superposed.}
        \label{fig:QSAQC}
    \end{subfigure}
    \caption{Superposing causal orders by superposing addresses.}
\end{figure}

Let us explain the behavior of this AQC, see Fig. \ref{fig:qsaqc}. For the sake of brevity, when specifying scatterings by their action on basis states, non-represented spaces are understood to be empty, i.e. in state $\eps$. Moreover, scatterings are understood to act as the identity on unspecified basis states. Sector $2$ initially contains addresses $3$, $4$ and $5$ as output stored addresses. When the data enters this sector, the gate operator $S^{\{2,5\}}$ superposes the addresses $3$ or $4$ in the target space  (Fig. \ref{fig:qsaqc:1}) according to the control qubit, e.g. $\ket \eps_{\calT}^2 \ket {345}_{\calW_\calO}^2 \ket{\ovvq c}_{\calQ_\calI}^2 \mapsto \alpha\ket 3_{\calT}^2 \ket{45}_{\calW_\calO}^2 \ket{\ovvq 0}_{\calQ_\calO}^2 + \beta\ket 4_{\calT}^2 \ket{35}_{\calW_\calO}^2 \ket{\ovvq 1}_{\calQ_\calO}^2 $. It is then 
transported to sectors $3$ (resp. $4$) as shown in Fig. \ref{fig:qsaqc:2}. There,  $U$ (resp. $V$) is applied, and the target space of the respective sector is filled by the first address of the stored address space (Fig. \ref{fig:qsaqc:3}). The data continues its journey, and $V$ (resp. $U$) is applied in Sector $4$ (resp. $3$). The two branches meet in Sector $5$ (\ref{fig:QSAQC}). The data halts here, while the target addresses fold back and the quantum switch reverts to its former, non entangled state. Then all data flow to Sector $6$. The mathematical details are found in Appendix \ref{annex:QS}.

The addresses flow back within the AQC from steps $5$ to $8$. This means that the addresses flow from the input space of a targeted sector to the output space of a targeting sector. These steps are essential to guarantee that the data can interfere outside of the AQC, i.e. that no intrication remains between quantum data and the quantum switch.
It is crucial to note that one cannot initialize this AQC with the address of Sector $6$, which is always accessed from Sector $5$, in the input or output spaces of Sector $2$, since otherwise the AQC would still be entangled with the output data after the addresses flowed back.

\begin{figure}[H]
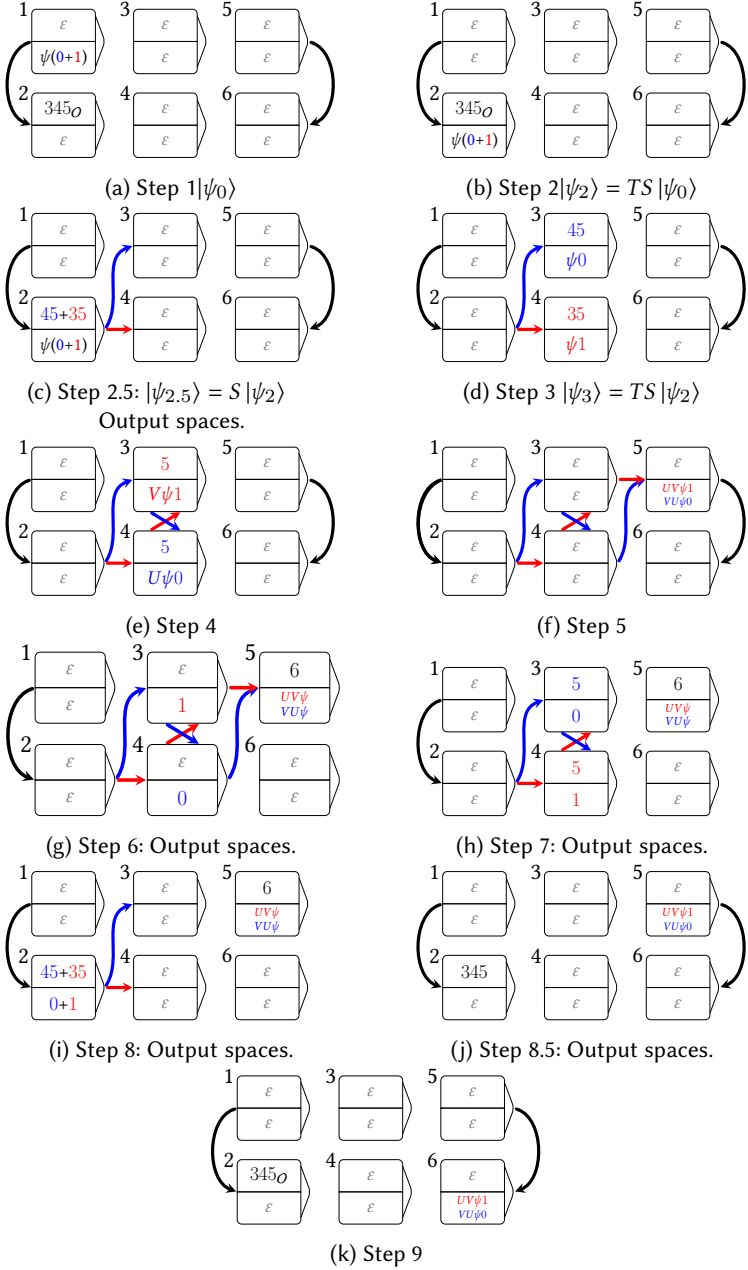

	\begin{center}
		\subfigqsaqc{1}{$\ket {\psi_0}$}\hspace{20pt}
		\subfigqsaqc{2}{$\ket{\psi_2}=TS\ket{\psi_0}$}
		\subfigqsaqc{2.5}{: $\ket{\psi_{2.5}}=S\ket{\psi_2}$\\\centering Output spaces.}\hspace{20pt}
		\subfigqsaqc{3}{   $\ket{\psi_3}=TS\ket{\psi_2}$\\ ~}
		\subfigqsaqc{4}{}\hspace{20pt}
		\subfigqsaqc{5}{}
		\subfigqsaqc{6}{: Output spaces.}\hspace{20pt}
		\subfigqsaqc{7}{: Output spaces.} 
		\subfigqsaqc{8}{: Output spaces.}\hspace{20pt}
		\subfigqsaqc{8.5}{: Output spaces.}
		\subfigqsaqc{9}{}
		\caption{Each step corresponds to one application of the global evolution operator $G=TS$. 
			The step $n.5$ for any integer $n$ corresponds to applying the operator $S$ to the state of step $n$. 
			The control qubit is $\ket c = \alpha\ket 0+\beta\ket 1$. Blue represents a scalar coefficient of $\alpha$ before the basis state, red represents a coefficient of $\beta$ and black represents $1$. When nothing is specified in caption, only input spaces are represented, {whereas "Output spaces" in caption indicates that only the output spaces are represented. A} $\calO$ in subscript (as in Figs. \ref{fig:qsaqc:1}, \ref{fig:qsaqc:2}, and \ref{fig:qsaqc:9}) indicates that a specific element is in the output space.}
		\label{fig:qsaqc}
	\end{center}
\end{figure}

The gate operators $S^{  \{3\}}$ and  $S^{  \{4\}}$ act both on a single sector space $\calH_\calS$ as they are one-sector gates. Both can be decomposed into 1) a gate acting on $\calH_\calS$ handling the swaps between the subspaces of the sector ($M_{34}$ in Appendix \ref{annex:m34}), and 2) a gate applying $U$ or $V$ on the respective data space $\calH_{\calQ_\calO}$, where ${\calQ_\calO}={\calQ^{\leq m+1}}$. The subsystems whose causal order is in superposition are therefore $\calH_{\calQ_\calO}^3$ and $\calH_{\calQ_\calO}^4$, see also \cite{OreshkovTD1}. In other words, the gate applies $U$ (or $V$) on a state consisting of $0$ qubits or $m+1$ qubits. $U$ (or $V$) is any unitary acting on $m$ qubits as it does not take into account the first qubit, i.e., the control qubit that was used in the gate $\{2,5\}$. We refer the reader to Appendix \ref{annex:m34} for details.

\subsection{The polarizing beam splitter}

In a Polarizing Beam Splitter (PBS), birefringent material separates an incident beam according to its polarization ---see Fig. \ref{fig:PBS:vh}.
A transmitted photon moves toward a device denoted by A that rotates the polarization angle of the light, and after that, moves toward a similar device denoted by B that rotates the polarization in a different way.

By using two PBS, the reflected version of the photon encounters B first, then A. The end result of the polarization in this case is different. There are two quantum paths: either A occurs before B, or B occurs before A, yielding an indefinite causal order. Indefinite causal orders are related to the time order of the evolution of the underlying quantum system. In the first case, A causally influences B so that if A had not occurred, then B’s input and output would be totally different. Conversely, in the second case, B causally influences A.

PBS are thus used in the literature to implement indefinite causal orders \cite{Goswami_2018,PBScalculus,PBScalculus:coh}, as in PBS the physical location of quantum data is itself quantum. Therefore, AQCs are well-suited to represent such devices, unlike textbook circuits which are ignorant towards~locations.

Here, we consider a PBS where horizontally polarized photons $(H)$ pass through the bifringent material, while vertically polarized ones $(V)$ are reflected without phase shift.
Photons as the physical carriers of data are represented by Fock states
with a finite number of particles, and are further distinguished by their polarization. We represent this PBS as an AQC $((\calA,\calG,\calQ,S),\ket{\psi_0})$, with $\calD = \{ 0_V,1_V,\dots m_V\} \times \{ 0_H, 1_H,\dots m_H \}$ for some integer $m$ being the total number of photons present in the circuit, $\calQ=\calD^{\leq1}$,  $\calA = \{ 1 ,\dots, 8 \}$, and $\calG = \{ \{1,2,3,4\},\{5\},\{6\},\{7\},\{8\}\}$. The first gate $S^{\{1,2,3,4\}}$ corresponds to the PBS itself, as shown in Fig. \ref{fig:PBS:vh}.
Sectors $5-8$ provide the context in which the the PBS is evolving and thus play the same role as the Sectors $1$ and $4$ of the Bell state creation AQC and the Sectors $1$ and $6$ in the \textsc{Switch} AQC.
In this AQC, photons can enter the PBS from any direction. The gate operator $S^{\{1,2,3,4\}}$ makes them travel to the right sector depending on their polarization (Fig. \ref{fig:PBS:vh}). They exit or enter the PBS during the transport step (Fig. \ref{fig:PBS:quiescent}).

We choose the initial state $\ket {\psi_0} =  \bigotimes_{i=5}^{8}\ket{i-4}_\calT^i\ket{i}_\calT^{i-4}\ket{\psi_i}^i_{\calQ_I}$ where the $\ket{\psi_i}$ represent entry photons ---see Fig. \ref{fig:PBS:quiescent}, or a superposition of such $\ket{\psi_0}$. The content of the target spaces entails that each entry of the PBS may communicate only with the sector in front of it.

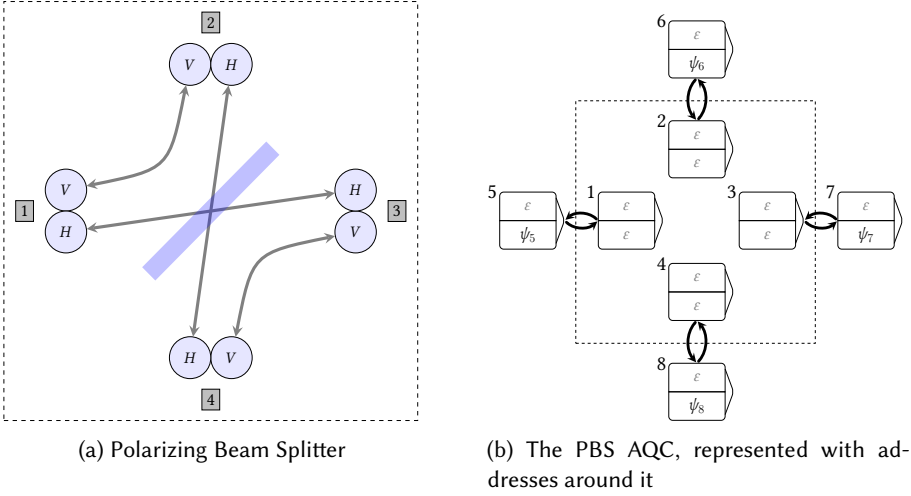
\begin{figure}[h]
    \centering
    \begin{subfigure}[t]{0.4\textwidth}
        \centering
        \resizebox{\textwidth}{!}
        {\renewcommand{\eps}{\epsgr}\begin{tikzpicture}
    \draw[dashed] (0,6) rectangle (10,-4);
    \adresse{0}{0}{1}{4};
    \adresse{4}{4}{2}{1};
    \adresse{8}{0}{3}{2};
    \adresse{4}{-4}{4}{3};
    \flechedouble{1V}{2V}{4}{2}{dark_box_color};
    \flechedouble{1H}{3H}{5}{1}{dark_box_color};
    \flechedouble{2H}{4H}{5}{1}{dark_box_color};
    \flechedouble{4V}{3V}{6}{0}{dark_box_color};
        
    \draw[color=blue,line width=0.45cm,opacity=0.25] (3.5,-0.5)--(6.5,2.5);
\end{tikzpicture}\renewcommand{\eps}{\epsgr}}
        \caption{Polarizing Beam Splitter}
        \label{fig:PBS:vh}
    \end{subfigure}
    \hspace{20pt}
    \begin{subfigure}[t]{0.4\textwidth}
        \centering
        \resizebox{\textwidth}{!}
        {\renewcommand{\eps}{\epsgr}\begin{tikzpicture}[thick]
\draw[dashed] (-3.3,-3.3) rectangle (5.2,5.2);
\sector{-2.5}{0}{$\eps$}{1}{$\eps$}{$\eps$}
\sector{-6}{0}{$\psi_5$}{5}{$\eps$}{$\eps$}
\sector{0}{2.5}{$\eps$}{2}{$\eps$}{$\eps$}
\sector{0}{6}{$\psi_6$}{6}{$\eps$}{$\eps$}

\sector{2.5}{0}{$\eps$}{3}{$\eps$}{$\eps$}

\sector{6}{0}{$\psi_7$}{7}{$\eps$}{$\eps$}

\sector{0}{-2.5}{$\eps$}{4}{$\eps$}{$\eps$}

\sector{0}{-6}{$\psi_8$}{8}{$\eps$}{$\eps$}

\fleche[bend right= 40]{bottommid6}{topmid2}{black}
\fleche[bend right= 40]{topmid2}{bottommid6}{black}
\fleche[bend right= 40]{topmid8}{bottommid4}{black}
\fleche[bend right= 40]{bottommid4}{topmid8}{black}

\fleche[bend right= 40]{outarrow3}{inarrow7}{black}
\fleche[bend right= 40]{inarrow7}{outarrow3}{black}
\fleche[bend right= 40]{inarrow1}{outarrow5}{black}
\fleche[bend right= 40]{outarrow5}{inarrow1}{black}

\end{tikzpicture}\renewcommand{\eps}{\epsgr}}
        \caption{The PBS AQC, represented with addresses around it}
        \label{fig:PBS:quiescent}
        \end{subfigure}
     \label{fig:PBS}
    \caption{$H$ and $V$ indicate dataspaces corresponding to photons of horizontal (resp. vertical polarization). The dashed square represents the multi-sector PBS gate.}
\end{figure}

The gate operator of the PBS gate $S^{\{1,2,3,4\}}$ acts on general states of the Sectors $1$ to $4$ as
\begin{equation}\label{eq:pbs}
S^{\{1,2,3,4\}}:\bigotimes_{i=1}^{4}\ket{i+4}_\calT\ket{n_i}_V\ket{m_i}_H\longmapsto\bigotimes_{i=1}^{4}\ket{i+4}_\calT\ket{n_{v(i)}}_V\ket{m_{h(i)}}_H,
\end{equation}
where $v$ (resp. $h$) is the function associating to a sector's vertical (resp. horizontal) polarization the sector indicated in Fig. \ref{fig:PBS:vh}:  
$v(i) = i \text{ mod } 2 + 1 + 2 \lfloor i/2\rfloor$ and $h(i) = (i + 1 ) \text{ mod } 4 + 1$.

\section{Renamings}\label{sec:renamings}

So far, in our model a single arbitrarily chosen integer serves both as a name and as the address of a sector. Without further restrictions, this can lead to physically odd behavior: 
For example, one easily constructs gate operators whose action is controlled by the address stored in the target address space, see Fig.~\ref{fig: Not_RI}. 
This is quite strange, since in our model we wish to think of sectors merely as different physical locations and of addresses as just encoding the underlying geometry, i.e. the wirings of the gates. 

A standard way to implement this independence is to require invariance under \emph{renamings}, 
i.e. under permutations of addresses. This is akin to requesting identifier-obliviousness in distributed algorithms \cite{FraigniaudOblivious0}. We call a gate operator \emph{renaming-invariant}, or simply \emph{nameblind} if it commutes with all such permutations. For instance, after renaming $x\mapsto x+100\text{~mod~} |\calA|$, the gate operator $S^{\{1\}}$ that was applied at location $\{1\}$ is applied at location $\{101\ \text {mod}\ |\calA|\}$. 
Thus, as renamings reshuffle ``physical locations'' and their associated gate operators, renaming invariance is a global condition. 
In Sect. \ref{namedaqg}, after a detour via a larger Hilbert space in which physical locations and their addresses are explicitly separated, we deduce an equivalent but local condition on gate operators. 

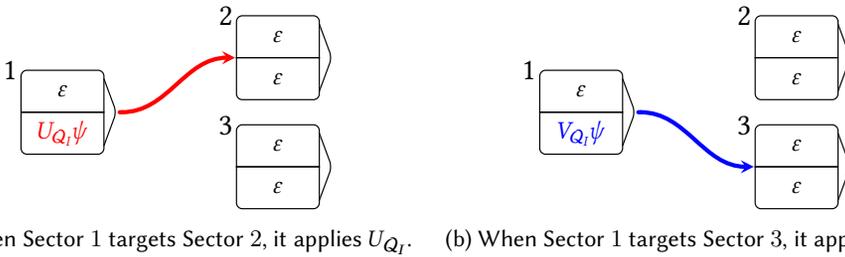
\begin{figure}[ht!]
	\begin{center}
		\begin{subfigure}[b]{0.49 \textwidth}
			\begin{center}
				\resizebox{0.65\textwidth}{!}{\begin{tikzpicture}[thick]
\sector{0}{0}{\textcolor{red}{$U_{\calQ_ I}\psi$}}{1}{$\eps$}{$\eps$}
\sector{5.2}{1.3}{$\eps$}{2}{$\eps$}{$\eps$}
\sector{5.2}{-1.3}{$\eps$}{3}{$\eps$}{$\eps$}
\fleche[in=180, out=0]{outarrow1}{inarrow2}{red}
\end{tikzpicture}}
				\caption{When Sector $1$ targets Sector $2$, it applies $U_{\calQ_ I}$.}
			\end{center}
		\end{subfigure}
		\begin{subfigure}[b]{0.49 \textwidth}
			\begin{center}
				\resizebox{0.65\textwidth}{!}{\begin{tikzpicture}[thick]
\sector{0}{0}{\textcolor{blue}{$V_{\calQ_ I}\psi$}}{1}{$\eps$}{$\eps$}
\sector{5.2}{1.3}{$\eps$}{2}{$\eps$}{$\eps$}
\sector{5.2}{-1.3}{$\eps$}{3}{$\eps$}{$\eps$}
\fleche[in= 180,out =0]{outarrow1}{inarrow3}{blue}
\end{tikzpicture}}
				\caption{When Sector $1$ targets Sector $3$, it applies $V_{\calQ_ I}$.}
			\end{center}
		\end{subfigure}
		\caption{Example of a gate operator that applies different operators depending on the targeted sector and is hence not ``renaming invariant''. Formally this operator could be written as $\ket{2}\bra{2}_\calT\otimes {U}_{\calQ_ I} + \ket{3}\bra{3}_\calT\otimes {V}_{\calQ_ I}$.}
		\label{fig: Not_RI}
	\end{center}
\end{figure}

A practical advantage of working with nameblind gate operators is the reduced number of cases to be specified. For instance, a gate operator manipulating $6$ addresses acts over $6!\times 7$ basis states of the address spaces if the target space is empty (ordering the six addresses and separating them amongst input and output) plus $6\times 6!$ such basis states if the target space is occupied, which gives a total of $9360$ ---see Appendix \ref{annex:addresses}. If the gate operator is nameblind, it suffices to specify its behaviour on $13$ basis states (fixing an order of the addresses and distributing them target, input and output) cases to specify it entirely. Renaming invariance is thus a very restrictive condition, and one may wonder whether it leaves any room for superposing addresses at all. Surprisingly, it does. We will fully characterize the structure of nameblind gate operators in Sect. \ref{sec:nameblind}.

An instructive analogy to renaming invariance from daily life is that of postal delivery, where a letter has to be delivered to a certain house with some house number. This house number is fairly arbitrary and might change due to bureaucratic processes. To guarantee the success of the mailman, when that happens, one has to make that change consistently, i.e. also on the letter and even in the phone books. In other words, if  only the house number changes, the letter will not arrive
. Below, we further distinguish renamings into ``internal'' and ``external'' ones w.r.t. a gate. Such renamings can be thought of as changing the house number only in a particular street and everywhere except in that street, respectively.

We remark that the idea of invariance under renamings of pointers is standard in the realm of classical programming languages. For instance, it is standard to consider pointers (or references) as opaque in high-level programming languages such as Java, or OCaml. From a semantics perspective, one way to model such opaque names is to use nominal sets: a variable that is precisely invariant under renaming. Unlike for nominal sets, the set of names in the current setting is finite. 

\subsection{Named AQC}\label{namedaqg}
In the formalism of AQC we use addresses in two distinct ways, with two different notations. 
Take for instance the state of a single sector $\ket{s}^i=\ket{t}^i_\calT\ket\ovva_\calW^i\ket\ovvq_\calQ^i$. 
On the one hand, addresses are ``physical locations''. This is denoted by the superscript $i$ (think of it like a house).
On the other hand, addresses are also used to target the flow of information. Those are the $t$ (think of it like the address of another house on an envelope) and the $\ovva$ (think of it like an address book from which the gate operator can choose). These are denoted as states of certain registers. 

In order to discuss renamings in the AQC context, we need to carefully separate the two uses of addresses, and put them on an equal footing. We achieve this by extending the Hilbert space by a ``name space'' that is induced by $\calN\equiv\calT$. Letting $\ket{t}^i_\calT\ket\ovva_\calW^i\ket\ovvq_\calQ^i \mapsto \ket{t}^i_\calT\ket{\ovva}^i_\calW \ket{\ovvq}^i_\calQ \ket{i}^i_\calN$ leads to the modified formalism which we refer to as \emph{Named Addressable Quantum Circuit} (NAQC) ---see \ref{annex:naqc:def}. 
In that larger space, when we rename addresses, we no longer need to modify the order in which the corresponding sectors occur. Renamings act locally upon each sector without shuffling them around: 
\begin{equation*}
	R\bigg(\bigotimes_i\ket{\ovva\ovvq}^i\ket{i}^i\bigg)  = \bigotimes \ket{R(\ovva)\ovvq}^i\ket{R(i)}^i.
\end{equation*}
In other words, superscript $i$ is the physical location (the house) at which a gate $S^{i}$ will always be applied, before and after a renaming. $\ket{R(i)}^i$ is the name of the physical location (the number on the house). Any occurrence of $R(i)$ within some $\ket{R(\ovva)\ovvq}^j$ still points toward the same physical location (address of the house). The identification between AQC and NAQC thus extends to 
\begin{equation}\label{eq:equivalence}
    \ket{t}^i_\calT\ket\ovva_\calW^i\ket\ovvq_\calQ^i \Longleftrightarrow  \{ \ket{R(t)}^i_\calT\ket{R(\ovva)}^i_\calW \ket{\ovvq}^i_\calQ \ket{R(i)}^i_\calN\ \ |\ \ R \text{ is a renaming}\}
\end{equation}
i.e. each AQC $A$ is associated to the equivalence class NAQCs whose renaming into the canonical name and then projection in the AQC space yields $A$.

\begin{definition}[Renaming and nameblind NAQC]
A bijection over $\cal{A}$ is called a \textbf{renaming}. It naturally extends pointwise to $\calW$ and $\calT$, and linearly to $\calH_{\calW}$, $\calH_{\calT}$, and to any Hilbert space by acting non-trivially only on its address (and  name) spaces. We call a renaming $R$ \textbf{external} to a gate $g$ if it acts trivially on $g$, i.e. if $a\in g\Rightarrow R(a)=a$, i.e. it only manipulates \emph{external addresses} $a \in \calA\setminus g$.
Moreover, we call an NAQC $((\calA,\calG,\calQ,S,\calN),\ket{\psi})$ \textbf{nameblind} if $S$ commutes with every renaming $R$, i.e. if $RS=SR$.
\end{definition}

Nameblindness can be tested locally on each gate by demanding that $\forall g \in \calG, RS^g = S^gR$. This allows us to lift the notion of nameblindness back to AQCs:

\begin{theorem}[Nameblind AQC]\label{thm:isomorphisme}
    Equivalence classes of nameblind NAQCs w.r.t. renaming are isomorphic to AQCs $((\calA,\calG,\calQ,S),\ket{\psi})$ for which for all $g$ in $\calG$, and for all renamings $E$ external to $g$, $S^gE=ES^g$.
    Such AQCs are therefore referred to as the \textbf{nameblind} AQCs.
\end{theorem}

Summarizing, an AQC is nameblind if its gate operators commute with all external renamings. Next, we show that such gate operators are made of blocks of `nameblind matrices', i.e. matrices that commute with all renamings, and study the structure of these. 

\subsection{Nameblind matrices}\label{sec:nameblind}

Gate operators $S$ of nameblind AQCs, called \emph{nameblind} gate operators, manipulate all external addresses in the same manner, yet may do so in a non-trivial way. For instance, $S$ may choose to swap one for the other, or not, in a superposition, as shown by the general form of nameblind unitary matrices over two addresses:
\begin{equation*}
	U_{\pm}(\varphi,\vartheta)=e^{i\varphi}\left(\begin{array}{cc}
\cos(\vartheta)& \pm i \sin(\vartheta)\\
\pm i \sin(\vartheta) & \cos(\vartheta)
\end{array}\right)
\end{equation*}
understood as acting over $\textrm{Span}\{\ket{a\ b},\ket{b\ a}\}$ with any $a,b \in \mathcal{A}$, $a\neq b$.
Note that the data register of the gate space can readily be used to control these angles, e.g. we could let $S\ket{\overline{a}\ \overline{x}}=\left(U_+(\varphi_{\overline{x}},\vartheta_{\overline{x}})\ket{\overline{a}}\right)\otimes \ket{\overline{x}}$:

\begin{example}\label{example:controlledU}
The action of the \textsc{switch} gate $S^{\{2,5\}}$ in the cases shown in Fig. \ref{fig:qsaqc:2.5} 
is described by $U_+(\varphi_{\overline{x}},\vartheta_{\overline{x}})\ket{\overline{a}}$ with control angles:
$\vartheta_{\overline{x}} = \varphi_{\overline{x}} = 0 $ when the control qubit in $\overline{x}$ is $0$, and
$\vartheta_{\overline{x}} = \varphi_{\overline{x}} = \frac{\pi}{2} $ when the control qubit in $\overline{x}$ is $1$,
followed by a swap of input and output data spaces.
\end{example}

Keeping this in mind, we now turn our attention towards matrices over $n$ addresses, which commute with all renamings. 
Such matrices act on the vector space whose \emph{canonical basis} is identified with the words of length $n$ on the alphabet $\{1\dots n\}$, i.e. $n$ addresses between $1$ and $n$. Moreover, we sort this canonical basis in ascending lexicographic order, i.e.  
$\{1\,2\cdots (n-1)\,n,\quad 1\,2\cdots n\, (n-1),\quad\dots,\quad n\,(n-1) \cdots 2\,1.\}$

\begin{theorem}[$(n,n)$-nameblind matrices] \label{thm:nameblind}
Matrices on states containing all $n$ addresses which commute with all renamings, i.e. $(n,n)$\emph{-nameblind matrices}, are of the following form: 
\begin{equation*}
	\begin{pmatrix}
        D&B&M_1B&M_2M_1B&M_3M_2M_1B&\dots\\
        B&D&M_1BM_1&M_2M_1BM_1&\dots&\\
        BM_1&M_1BM_1&D&&& \\
        BM_1M_2&M_1BM_1M_2&&\ddots&& \\
       BM_1M_2M_3&\vdots&&&&\\
       \vdots&&&&&
    \end{pmatrix}
\end{equation*}
where the blocks act on $n-1$ addresses, $D$ is any $(n-1,n-1)$-nameblind matrix, $M_i$ is the renaming $i\leftrightarrow i+1$, and $B$ (see Sect. \ref{annex:partially-nameblind}) commutes with $M_i$ (see Sect. \ref{annex:rkktranspo}) for each $i=2,3,...,n-2$.
\end{theorem}

By forcing the partially-nameblind matrix $B$ to be nameblind, inductively so, one obtains a convenient subfamily of $(n,n)$-nameblind matrices, built from exactly two smaller matrices of the same subfamily. Even though they are of size $n!\times n!$, these have only $\calO(2^n)$ degrees of freedom. This makes these matrices easy to construct in practice.

\begin{corollary}[Pure nameblind matrices]\label{prop:charac_general}
A pure nameblind matrix on $n$ addresses is a matrix of the following form :
    \begin{equation*}
        \begin{pmatrix}
        D&B&M_1B&M_2M_1B&M_3M_2M_1B&\dots\\
        B&D&B&M_2B&\dots&\\
        M_1B&B&D&&& \\
        M_1M_2B&M_2B&&\ddots&& \\
       M_1M_2M_3B&\vdots&&&&\\
       \vdots&&&&&
        \end{pmatrix},
    \end{equation*}
    where $D$ and $B$ are pure nameblind matrices on $(n-1)$ addresses and the $M_i$ are as in Thm.~\ref{thm:nameblind}. All pure nameblind matrices are also nameblind.
\end{corollary}

Our main goal is to characterize gate operators, which do not necessarily act on the full address space of an AQC. 
As gate operators never create nor suppress addresses, we can consider separately all lengths of address words. 
For instance, with $\calA=\{1,2,3\}$, we consider independently the action of a gate operator on words of length 2, i.e. on the subset $\{12,21,13,31,23,32\}$

Let $M$ be a nameblind matrix acting on words in $\calA^-$ of size $m\leq n = |\calA|$, meaning that $M$ is a nameblind matrix on a set of external addresses, some of which may not be present in the basis state. 
Such matrices, which we additionally require to neither create nor suppress addresses are called a $(m,n)$-nameblind matrix. They are characterized as follows:

\begin{proposition}[$(m,n)$-nameblind matrices]\label{prop:nameblind}
Let $M$ be a $(m,n)$-nameblind matrix. 
When sorting the basis first by which addresses are present and then by lexicographical order, $M$ can be written as $\bigoplus_{\calA' \subset_m \calA} D$
    where $D$ is a ($m!$)-dimensional $(m,m)$-nameblind matrix as in Thm. \ref{thm:nameblind}.
\end{proposition}

\subsection{Gate Operators of a nameblind AQC}\label{nameblindqs} 

In practice, gate operators do not manipulate basis states solely composed of addresses. Basis states also contain data, separators between (\calT, $\calI$ and \calO) spaces and between the different sectors. Moreover, recall from Thm. \ref{thm:isomorphisme} that the gate operators of a nameblind AQC need only commute with external renamings. Still, they can be organized as blocks of nameblind matrices acting on well-identified subspaces, as in Fig. \ref{fig:subspaces:intext}, \ref{fig: nameblind block} and \ref{fig:subspaces}.

\begin{figure}[t]
    \centering
$\ket\eps_\calT\ket{1a}_\calI\ket\eps_\calO, \hspace{11pt} \ket\eps_\calT\ket{a1}_\calI\ket\eps_\calO,\hspace{11pt}\ket\eps_\calT\ket a_\calI\ket{1}_\calO,\hspace{11pt} \ket\eps_\calT\ket1_\calI\ket{a}_\calO,$\\ \phantom{.}\\
$\ket\eps_\calT\ket{\eps}_\calI\ket{1a}_\calO, \hspace{11pt} \ket\eps_\calT\ket{\eps}_\calI\ket{a1}_\calO,\hspace{11pt}\ket 1_\calT\ket a_\calI\ket{\eps}_\calO,\hspace{11pt} \ket1_\calT\ket\eps_\calI\ket{a}_\calO,$\\ \phantom{.}\\
$\ket a_\calT\ket{1}_\calI\ket\eps_\calO ,\hspace{11pt} \ket a _\calT\ket\eps_\calI\ket 1_\calO$
    \caption{Consider a gate operator $S^{\{1\}}$ acting over one external address $a$ plus its own address. Its state space can be organized into the above ten `nameblind' subspaces which should be interpreted as $V_1=\textrm{Span}_{a\in\mathcal{A}\backslash\{1\}} \{\ket\eps_\calT\ket{1a}_\calI\ket\eps_\calO\}$, $V_2=\textrm{Span}_{a\in\mathcal{A}\backslash\{1\}} \{\ket\eps_\calT\ket{a1}_\calI\ket\eps_\calO\}$, etc. Knowing in which of these subspaces $V_i$ we are, $a$ fully determines a basis state. That is, basis states can be identified with the tuple $(V_i,a)$. 
    All states within one of such subspaces evolve in the same manner under the action of $S^{\{1\}}$.
    Therefore, it suffices to describe the action of the gate operator on a general state $\ket{V_i}$ to describe its action on all the true basis states $\ket{V_i}\otimes\ket a$.}
    \label{fig:subspaces:intext}
\end{figure}

A nameblind gate operator $S$ can be written using only blocks of the form given by Thm. \ref{thm:nameblind}:
\begin{itemize}[noitemsep,topsep=0pt,parsep=0pt,partopsep=0pt]
    \item First, we decompose the large gate space into $|Q^{2|g|}|$ subspaces for each possible basis state of the dataspace. For each couple of basis states ($\ket{q_0}, \ket{q_1}) \in {\calH}_{\calQ^{2|g|}}^2$ we denote $S_{q_0\to q_1}$ the block of $S$ which maps states with data $\ket{q_0}$ to states with data $\ket{q_1}$.
    \item Then we regroup states according to their number of external addresses $m$ and the set of internal addresses $\calA_{int}$ contained in their address spaces.
    Fig. \ref{fig:subspaces:intext} exemplifies the subspace defined by $\calA_{int}=\{1\}$ and $m=1$.
    Since both parameters are preserved by $S$ and $R$, any $S_{q_0\to q_1}$ is block-diagonal in this basis.
    
    \item By also fixing positions $P_{int}$ and $P_{ex}$ for internal and external addresses, we can define even smaller subspaces preserved by renamings (but not by S). The blocks corresponding to this subspaces are $(m,n-|g|)$-nameblind matrices (Prop. \ref{prop:nameblind}) ---see Appendix \ref{annex: from GO to (n,n) nameblind matrices} . They are illustrated in Fig. \ref{fig:block:nameblind} and correspond to the different subspaces described in Fig. \ref{fig:subspaces:intext}. 
    
    \item Finally we know that these $(m,n-|g|)$-nameblind matrices are block diagonal over subspaces defined by a set of $m$ external addresses (Prop. \ref{prop:nameblind}), and only contain the same  $(m,m)$-nameblind matrix in each block (in Fig. \ref{fig: nameblind block}, each block $A_{i,j}$ denotes such a matrix). Indeed, these operators act the same way whatever the names of external addresses are, because we made such addresses indistinguishable.
\end{itemize}

These results are proved in full detail in Appendix \ref{annex: from GO to (n,n) nameblind matrices} together with Prop. \ref{prop:gateopdecomp}, which states formally that a gate operator can be decomposed into 
$(k,k)$-nameblind matrices for some $k$.

\begin{proposition}[Decomposition of gate operators into $(m,m)$-nameblind matrices]\label{prop:gateopdecomp}
Let $S^g$ be a gate operator on a set of addresses $\calA$. Suppose that it commutes with all renamings on external addresses. Consider any words on quantum data $q_0$ and $q_1$, internal addresses $\calA_{int} \subseteq g$ with positions $P_{int_{1}}$ and $P_{int_{2}}$, external addresses  $\calA_{ex} \subseteq  \calA \backslash g$ with positions $P_{ex_{1}}$ and $P_{ex_{2}}$
. The action of $S^g$ between states corresponding to these parameters $S^g_{\calA_{int},\calA_{ex},q_0\to q_1,P_{int_{1}}\to P_{int_{2}},P_{ex_{1}}\to P_{ex_{2}}}$, is a $(m,m)$-nameblind matrix, where $m = |\calA_{ex}|$.
\end{proposition} 

\begin{figure}[htbp]
    \centering
    \savebox{\bigmat}{
    	\hbox{\scalebox{0.65}{
    		$B_i = 
    		\left(
    		\begin{array}{ccc|ccc|ccc|ccc}
    		A_{1,1} & 0 & 0 & A_{1,2} & 0 & 0 &&&& A_{1,N} & 0 & 0\\
    		0 & \ddots & 0 & 0 & \ddots & 0 && \dots &&0 & \ddots & 0\\
    		0 & 0 & A_{1,1} & 0 & 0 & A_{1,2} &&&& 0 & 0 & A_{1,N} \\\hline
    		A_{2,1} & 0 & 0 & A_{2,2} & 0 & 0 &&&& A_{2,N} & 0 & 0 \\
    		0 & \ddots & 0 & 0 & \ddots & 0 && \dots && 0 & \ddots & 0 \\
    		0 & 0 & A_{2,1} & 0 & 0 & A_{2,2} &&&& 0 & 0 & A_{2,N} \\\hline
    		&&&&&&&&&&\\
    		&\vdots&&&\vdots&&&\ddots&&&\vdots&\\
    		&&&&&&&&&&\\\hline
    		A_{N,1} & 0 & 0 & A_{N,2} & 0 & 0 &&&& A_{N,N} & 0 & 0\\
    		0 & \ddots & 0 & 0 & \ddots & 0 && \dots && 0 & \ddots & 0\\
    		0 & 0 & A_{N,1} & 0 & 0 & A_{N,2} &&&& 0 & 0 & A_{N,N}
    		\end{array}
    		\right)
    		$}
    	}
    }
    \subcaptionbox{\label{fig:block:data}Block $S_{q_0\to q_1}$ of a gate operator}[0.39\textwidth]{
	    \centering
	    \raisebox{\dimexpr\ht\bigmat-.5\height}{
	    \scalebox{0.8}{$S_{q_0\to q_1}= 
	    \left(\begin{array}{c|c|c|c}
	    B_1 & 0 & \dots & 0\\\hline
	    0 & B_2 & \dots & 0\\\hline
	    \vdots & \vdots & \ddots &\vdots\\\hline
	    0 & 0 & \dots & B_{f(n,m)}
	    \end{array}\right)
	    $}
	}
    }
	\subcaptionbox{\label{fig:block:nameblind}Detailed representation of the blocks $B_{i}$ of Fig. \ref{fig:block:data}. If $B_i$ corresponds to subspaces with $k$ external addresses, the $A_{i,j}$ are $(k,k)$-nameblind matrices.}[0.60\textwidth]{
     \centering
     \usebox{\bigmat}
    }
    \caption{
     To each pair of values $q_0,q_1$ in the data space corresponds a block $S_{q_0\to q_1}$ of the gate operator. 
     This block $S_{q_0\to q_1}$ may be divided into sub-blocks $B_i$ acting on subspaces having a given set of internal addresses $\calA_{int}$ and a given number of external addresses $m$ ---see Fig. \ref{fig:block:data}.
     Non-diagonal sub-blocks are zero because gate operators preserve these parameters. The diagonal sub-blocks are divided into  $(m,n)$-nameblind matrices on more addresses than the ones present---see Prop. \ref{prop:nameblind} for details. Those $(m,n)$-nameblind matrices are then split into $(m,m)$-nameblind matrices, yielding Prop. \ref{prop:gateopdecomp}.}
    \label{fig: nameblind block}
\end{figure}

\begin{example}
In Sect. \ref{whatcando}, all gate operators are nameblind.
As an example, let us argue that the gate operator of the $\textsc{Switch}$ $S^{\{2,5\}}$ from Sect. \ref{sec:qs} commutes with any renaming over $\calA = \{1,2,3,4,5,6\}$. 
$S^{\{2,5\}}$ is controlled by many inputs, therefore proving renaming invariance in all cases specified in Appendix \ref{annex:QS} would be cumbersome. Instead, we focus on the first case - the one leading to Fig. \ref{fig:qsaqc:1}.
. It explains how states of the subspace $\calK  = \{\ket {ta\ovvap}_{\calW_\calO}^2 \ket{\ovvq c}_{\calQ_\calI}^2 \ket {b}_{\calT}^5 \in \calH^{\{2,5\}}\}$
evolve:
\begin{equation}
S^{\{2,5\}} \ket {ta\ovvap}_{\calW_\calO}^2 \ket{\ovvq c}_{\calQ_\calI}^2 \ket {b}_{\calT}^5   = \begin{cases}\ket{t}_\calT^2\ket{a\ovvap,\ovvq c}_\calO^2 \ket {b}_{\calT}^5&\text{if }c=0\\
    \ket{a}_\calT^2\ket{t\ovvap,\ovvq c}_\calO^2 \ket {b}_{\calT}^5&\text{if }c=1\end{cases}
\end{equation}
To see that this operator is nameblind, we consider an element of the form $\ket {ta\ovvap}_{\calW_\calO}^2 \ket{\ovvq 1}_{\calQ_\calI}^2 \ket {b}_{\calT}^5 $ in $\calK$ and a renaming $R$:
\begin{align*}\label{eq: qs commute}
RS^{\{2,5\}}\ket {ta\ovvap}_{\calW_\calO}^2 \ket{\ovvq 1}_{\calQ_\calI}^2 \ket {b}_{\calT}^5
    &= \ket{R(a)}_\calT^2\ket{R(t)R(\ovvap),\ovvq c}_\calO^2 \ket {R(b)}_{\calT}^5 \\
    &= S^{\{2,5\}} \ket {R(t)R(a)R(\ovvap)}_{\calW_\calO}^2 \ket{\ovvq 1}_{\calQ_\calI}^2 \ket {R(b)}_{\calT}^5\\
    &=S^{\{2,5\}}R \ket {ta\ovvap}_{\calW_\calO}^2 \ket{\ovvq 1}_{\calQ_\calI}^2 \ket {b}_{\calT}^5.
\end{align*}
Similar calculations for the other cases show the nameblindness of $S^{\{2,5\}}$.
\end{example}

\section{Related works}\label{annex:related}

In the following, we discuss links between our model and other approaches to
indefinite causal orders in the literature.

\medskip
\noindent
\textit{Quantum Causal Graph Dynamics (QCGD).}~
Another approach to represent indefinite causal orders is to shun away from the circuit formalism and take inspiration from cellular automata. The QCGD framework \cite{ArrighiQCGD} features quantum superpositions of graphs evolving synchronously according to unitary local rules. Each vertex has a local quantum state and communicates with other vertices via ports. Compared to AQCs, vertices have replaced gates and edges have replaced wires. Like QCGD, our model is able to represent a quantum superposition of different graphs (the graphs of connectivity between gates). 
The interdiction of non-causal evolution in QCGD may appear as contradictory with the sudden appearance of an edge in the AQC between nodes of different gates --- for instance, when an address suddenly goes from the input address space to the target. 
However, by encoding address spaces into port names, one may simulate this seemingly non-causal edge appearance at the cost of having a larger number of ports. We can therefore encode any AQC into a QCGD (see Appendix~\ref{section:QCGD}).

\noindent
\textit{Quantum circuits with quantum control of causal order.}~
One way is to allow higher-order circuits, with supermaps encoding of quantum control \cite{wechs2021quantum}, in a way that generalizes the quantum switch. All thus far known physically realizable coherent control of orderings can be expressed within this QC-QC formalism. Notice that these supermaps cannot reuse the same black box twice, as causality is not handled as in AQCs.

\noindent
\textit{PBS calculus.}~
A second proposal uses a graphical language inspired by PBS to express causal ordering of general, non unitary quantum channels \cite{PBScalculus,PBScalculus:coh}.
The PBS diagram formalism \cite{PBScalculus} is a graphical language modeling coherent control of "purified channels" \cite{PBScalculus:coh}, i.e. unitary matrices enhanced with ancill\ae~ from an environment, to represent CPTP maps. It provides completeness theorems for coherent quantum control. This control is obtained by having the PBS as a primitive of this graphical language. Similar to AQC, the physical location of quantum data is itself quantum. However, unlike AQC in the PBS calculus wires are fixed.

\noindent
\textit{Routed quantum circuits (RQC).}~
A third idea consists in building a process theory \cite{coecke_kissinger_2017} which allows for superposition of paths \cite{RoutedQC}. RQC \cite{RoutedQC} is an extension of the quantum circuit formalism which allows for communication in a superposition of paths \cite{Quantum-Shannon-Theory}. In this work, the authors describe sectorial constraints to study causal decomposition \cite{Causal-decomposition} of the used unitaries. In this purpose they partition Hilbert spaces into sectors, over which they describe with a relation ---called a route--- which sectors are causally connected. This is reminiscent of our model as each sector's target may be understood as a specific sectorial constraint. However the formalism of \cite{RoutedQC} does not allow for a dynamical rewiring of the circuit.

\noindent
\textit{Causal Boxes.}~
A fourth way is to admit superpositions of temporal orders of messages with \emph{Causal Boxes} \cite{CausalBoxes}. Causality is ensured by enforcing that each new message depends on messages labelled with lower time tags. This allows for implementations of systems as the quantum switch or a the control of an unknown unitary. Our model bears similarity to \cite{CausalBoxes}, for instance the input and output spaces of sectors are analogous to the input and output wires of causal boxes.  Nevertheless, unlike \cite{CausalBoxes} where circuits are fixed, our model is inherently causal \emph{by construction}, despite the dynamical aspect of the wiring structure.

\noindent
\textit{Quantum Shannon Theory with superposition of trajectories.}~
The paper \cite{Quantum-Shannon-Theory} extends the usual Quantum Shannon Theory by allowing trajectories of the information carriers to be superposed. A separation between trajectory and internal quantum information is enforced similar to our separation between addresses and data. If "phase kickback" \cite{phasekickback} were to be used to encode supplementary information in the trajectories, they would become part of the message, and the standard Quantum Shannon Theory would be sufficient to express the communication protocol. In the AQC formalism, all external addresses are handled in the same way. This ensures that the phase kickback only arises from things such as the number of addresses and their position, but never from their value.

\section{Conclusion} \label{part:comparison}

{\em Summary.} We introduced Addressable Quantum Circuits (AQC), which consist of ``gates'' that are themselves divided into ``sectors''. Each sector contains ``data'' and ``addresses''. The evolution decomposes into two steps. During the `scattering step', each gate applies a local unitary operator over the data and addresses of its constituent sectors. During the `transport step', each sector holding an address within its `target space', swaps the content of its `output space' with that of the `input space' of the target sector. 
We showed how indefinite causal orders and superpositions of paths emerge naturally in this formalism, by encoding the quantum switch and the  polarizing beam splitter. 
Operations that manipulate addresses raised the question of their fundamental nature. If we require addresses to represent purely geometrical data, we must restrict such operations by demanding that they are indistinguishable. Hence our focus on `nameblind' AQCs, i.e. those whose scattering step commute with `renamings' of addresses. Through a detour via the sub-formalism of Named AQC, we derived local conditions on the gates of AQCs for it to be nameblind. The gates must commute with renamings over `external addresses', i.e. those addresses which do not pertain to the gate. We were able to give the general form of the $(n,n)$-`nameblind matrices' (i.e. acting only over lists of addresses and commuting with every renaming) first, and then show that gates are blockwise decomposable into these.
In the appendices we studied composability, and encoding within Quantum Causal Graph Dynamics (QCGD).\\
Altogether this provides a concrete, composable model of distributed quantum computing, featuring quantum evolutions of connectivity in the spirit of classical markovian graph dynamics \cite{MarkovianDynGraphs}, and a thorough understanding of the limitations brought by renaming invariance in the spirit of classical identifier-obliviousness \cite{FraigniaudOblivious0}.

{\em Perspectives.} The structure of nameblind matrices and gates might simplify even further if combined with the unitary condition. Further comparison with related models could also be fruitful: Are AQCs equivalent, in terms of their expressiveness, to quantum channels extended with vacuum states \cite{Quantum-Shannon-Theory}? Could an expressive process language be devised that would encompass AQC, including their dynamical aspects \cite{BruknerDynamics}?
We proved that each AQC can be simulated by a QCGD, but is the converse simulation always possible? This last question would help us establish a form of completeness of AQCs within in-principle physically realisable dynamics over quantum causal orders.\\

\begin{acks}
We thank Cyril Branciard and Michel Quercia for insightful discussions and Pascal Baßler for carefully reading the manuscript. This work was funded by the Deutsche Forschungsgemeinschaft (DFG, German
Research Foundation) -- 441423094 and the German Federal Ministry of Education and Research
(BMBF) within the funding program ``quantum technologies -- from basic
research to market'' via the joint project MIQRO under the grant number
13N15522. 
This publication was made possible through the support of the ID\# 61466 grant from the John Templeton Foundation, as part of the “The Quantum Information Structure of Spacetime (QISS)” Project (\href{qiss.fr}{qiss.fr}). The opinions expressed in this publication are those of the author(s) and do not necessarily reflect the views of the John Templeton Foundation».
\end{acks}

\bibliography{biblio}

\newpage

\appendix
\section{Addressable Quantum Circuits}\label{annex:defaqc}
\subsection{Details on the model}\label{annex:addresses}
Several notions are introduced in this paper. For a summary of notations ---see Sect. \ref{annex:notations}. This section provides details on Sect. \ref{Section:model}.

\subsubsection*{Addresses}
In the core of the paper the stored address spaces and the target spaces were defined in words. More formally, we have
\begin{definition}[Addresses]\label{def:addresses}
The \emph{set of addresses} $\calA$ is a finite subset of $\mathbb{N}$. Moreover, we denote by $\calW = \calA^-=\{a_1\dots a_n|a_i\in\calA,a_i\neq a_j\:\forall i,j\}$ the set of words of \emph{non-repeating addresses}, and by $\calT = \calA^?=\{a \mid a\in\calA^1\}\cup\{\eps\}$ the set of words of length at most one.
\end{definition}

\noindent When $n$ addresses are in a sector space, a gate operator can map them to any superposition of the $(2n+1)n!$ different states 
: 
\begin{itemize}
    \item if the target space is empty, there are $n+1$ possibilities for the number of addresses in the input space (the others are in the output space) and $n!$ different ways to order them.
    \item if the target space is non-empty, the $n-1$ remaining addresses are to be distributed between input and output space: any number of addresses from $0$ to $n-1$ can be in the input space, and as above, there are $n!$ different ways to order them order.
\end{itemize}

When $n$ addresses are in a sector space, and additionally the gate operator is \emph{nameblind}, there are only $2n+1$ different possibilities counted as above but up to reordering. It suffices to specify the behaviour of the gate operator on one of each of those $(2n+1)$ possibilities, to have them specified on all the corresponding $n!$ reordered possibilities, by renaming invariance ---see Fig. \ref{fig:subspaces:intext}. 

\subsubsection*{Flipping gates}\label{annex:gates}

To prepare the transport to a targeted sector, many gate operators push some of their content to the output spaces 
after manipulating it. This is achieved by flipping operations, which are in some sense necessary since gates whose scattering unitary is the identity $I$ are merely as a ``mirror'', i.e. they reflect incoming data. In the example of the Bell state creation circuit in \eqref{eq:eprgateoperator} the input and output spaces are flipped by
\begin{equation*}
    F:\ket{t}_\calT\ket{\ovva,\ovvq}_\calI\ket{\ovvap,\ovvqp}_\calO\longmapsto \ket{t}_\calT\ket{\ovvap,\ovvqp}_\calI\ket{\ovva,\ovvq}_\calO.
\end{equation*}
Other possibilities are flipping only the data spaces, i.e.
\begin{equation*}
    F_{\text{data}}:\ket{t}_\calT\ket{\overline{a},\ovvq}_\calI\ket{\ovvap,\ovvqp}_\calO\longmapsto \ket{t}_\calT\ket{\overline{a},\ovvqp}_\calI\ket{\ovvap,\ovvq}_\calO,
\end{equation*}
or more sophisticated flips that are conditioned on the state of the sector.

\subsubsection*{Evolution}

Like in Def. \ref{def:Evolution}, the example below uses the following convention on the states : $\ket{x}^y_z$ indicates that the space $z \in \{\calT,\calI,\calO\}$ of the sector $y \in \calA$ is in the state $\ket{x}$.

\begin{example}\label{example:transport}
Effect of the transport step $T$ with $\calA = \{1,2\} $:

Without basis reordering : 
\begin{equation*}
	\begin{array}{rcl}
    	\ket{2}^1_\calT \ket{\eps,q_1}^1_\calI \ket{\eps,\eps}^1_O & \multirow{2}{*}{$\longrightarrow_T$}& \ket{2}^1_\calT \ket{\eps,q_1}^1_\calI \ket{1,q_2}^1_O   \\[1mm]
     	\ket{\eps}^2_\calT \ket{1,q_2}^2_\calI \ket{\eps,\eps}^2_O  &  & \ket{\eps}^2_\calT \ket{\eps,\eps}^2_\calI \ket{\eps,\eps}^2_O
	\end{array}
\end{equation*}

With basis reordering of Sect. \ref{sect:address}: 
\begin{align*}
&\ket{2,\eps}_\calT  \ket{(\eps,q_1),(1,q_2)}_\calI \ket{(\eps,\eps),(\eps,\eps)}_O \\
\longrightarrow_T 
&\ket{2,\eps}_\calT \ket{(\eps,q_1),(\eps,\eps)}_\calI \ket{(1,q_2),(\eps,\eps)}_O
\end{align*}
\end{example}

The address $1$ and data $q_2$ were in the input space of the sector 2, and $T$ transports them to the output space of sector $1$, because sector $1$ targeted sector $2$.

\begin{example}\label{example:epr:ts}
The Bell state creation circuit evolves very simply --- missing the obvious coefficient here and not writing target spaces (as they do not change) and empty spaces :
\begin{equation*}
	\begin{array}{@{~~}r@{~~}l@{~~}l@{~~}l@{~~}l@{~~}l@{~~}l@{~~}l@{~~}l@{}}
	\ket{00}_{\calI}^1 &\to_S  & \ket{00}_{\calO}^1  &
    \to_T & \ket {00}_\calI^2 &
    \to_S & (\ket{0} + \ket 1)\ket0_\calO^2 &
    \to_T & (\ket0 + \ket1)\ket 0_\calI^3 \\[1em]
    & \to_S & \ket{00}_\calO^3 + \ket{11}_\calO^3 &
    \to_T & \ket{00}_\calI^4 + \ket{11}_\calI^4 &
    \to_S & \ket{00}_\calO^4 + \ket{11}_\calO^4  &
    \to_T & \ket{00}_\calO^4 + \ket{11}_\calO^4,
	\end{array}
\end{equation*}
where the lack of a target address in sector $4$ causes the data to stay there during the last transport step.
\end{example}

\subsection{Details of the Quantum Switch}\label{annex:QS}
\subsubsection{Definition}\label{annex:m34}
The description of the quantum switch in Sect. \ref{sec:qs} can be completed by specifying some gate operators $S^g$. Their action on the addresses is as follows:
\begin{equation*}
    S^{\{1\}} = F = S^{\{6\}}, \qquad S^{\{3\}} = M_{34}\cdot L_U, \qquad  S^{\{4\}} = M_{34} \cdot L_V.
\end{equation*}
Here, $L_g$ applies the gate $g$ 
to the first $m$ bits of data, i.e. on all data except the control qubit, and 
\begin{align*}
    M_{34} = \left(
    \begin{array}{l@{\,}l@{\,}l@{~~\longleftrightarrow~~}l@{\,}l@{\,}l}
    \ket{t}_\calT&\ket{\ovva}_{\calW_\calO} &\ket{\ovvq}_{\calQ_\calO} &  \ket{\eps}_\calT&\ket{t\ovva}_{\calW_\calI}&\ket{\ovvq}_{\calQ_\calI} \\
     \ket t_\calT&\ket\eps_{\calW_\calO}&\ket{\ovvq}_{\calQ_\calO} & \ket\eps_\calT& \ket t_{\calW_\calI}&\ket{\ovvq}_{\calQ_\calI}
     \end{array}
     \right)
\end{align*}
The gate operator $S^{\{2,5\}}$, the effect of which is detailed in \ref{annex:QS:correctness}, acts as follows (where $a,b,t \in \calA$, $\ovvap \in \calA^-$, $c\in\calD$, $\ovvq \in \calD^m$ and, as usual, states that are omitted correspond to $\eps$):

\begin{align}\label{eq:S25:s2}
S^{\{2,5\}} \ket {ta\ovvap}_{\calW_\calO}^2 \ket{\ovvq c}_{\calQ_\calI}^2 \ket {b}_{\calT}^5 &= \begin{cases}\ket{t}_\calT^2\ket{a\ovvap,\ovvq c}_\calO^2 \ket {b}_{\calT}^5&\text{if }c=0\\
    \ket{a}_\calT^2\ket{t\ovvap,\ovvq c}_\calO^2 \ket {b}_{\calT}^5&\text{if }c=1\end{cases}\\[2mm]
    \label{eq:S25:s5}
S^{\{2,5\}} \ket{b}_\calT^2\ket{t}_\calT^5\ket{\ovvq c}_{\calQ_\calI}^5 &= \ket{b}_\calT^2\ket{t}_{\calW_\calO}^5\ket{c}_{\calQ_\calI}^5\ket{\ovvq}_{\calQ_\calO}^5 \\[2mm]
\label{eq:S25:s2s5}
S^{\{2,5\}}\ket{t}_\calT^2\ket{b\ovvap}_{\calW_\calO}^2 \ket{c}_{\calQ_\calO}^2\ket{a,\ovvq}_\calO^5 &=\begin{cases}\ket{tb\ovvap}_{\calW_\calO}^2 \ket a_\calT^5\ket{\ovvq c}_{\calQ_\calO}^5 &\text{if }c=0 \\\ket{bt\ovvap}_{\calW_\calO}^2\ket a_\calT^5\ket{\ovvq c}_{\calQ_\calO}^5 &\text{if }c=1   \end{cases}
\end{align}
We ensure unitarity of $S^{\{2,5\}}$ by specifying a reciprocal version of the equations above:
\begin{itemize}
    \item reciprocal version of \eqref{eq:S25:s2}: $ \begin{cases}S^{\{2,5\}}\ket{t}_\calT^2\ket{a\ovvap,\ovvq 0}_\calO^2 \ket {b}_{\calT}^5 =
    \ket {ta\ovvap}_{\calW_\calO}^2 \ket{\ovvq 0}_{\calQ_\calI}^2 \ket {b}_{\calT}^5 \\
    S^{\{2,5\}}\ket{a}_\calT^2\ket{t\ovvap,\ovvq 1}_\calO^2 \ket {b}_{\calT}^5 = \ket {ta\ovvap}_{\calW_\calO}^2 \ket{\ovvq 1}_{\calQ_\calI}^2 \ket {b}_{\calT}^5  \end{cases}$
    \item of \eqref{eq:S25:s5}: $ S^{\{2,5\}} \ket{b}_\calT^2\ket{t}_{\calW_\calO}^5\ket{c}_{\calQ_\calI}^5\ket{\ovvq}_{\calQ_\calO}^5\ket\eps_\calW^2 = \ket{b}_\calT^2\ket{t}_\calT^5\ket{\ovvq c}_{\calQ_\calI}^5\ket\eps_\calW^2$
    \item and of \eqref{eq:S25:s2s5}: $\begin{cases}\ket{tb\ovvap}_{\calW_\calO}^2 \ket a_\calT^5\ket{\ovvq 0}_{\calQ_\calO}^5  = S^{\{2,5\}}\ket{t}_\calT^2\ket{b\ovvap}_{\calW_\calO}^2 \ket{0}_{\calQ_\calO}^2\ket{a,\ovvq}_\calO^5 \\\ket{bt\ovvap}_{\calW_\calO}^2\ket a_\calT^5\ket{\ovvq 1}_{\calQ_\calO}^5 = S^{\{2,5\}}\ket{t}_\calT^2\ket{b\ovvap}_{\calW_\calO}^2 \ket{1}_{\calQ_\calO}^2\ket{a,\ovvq}_\calO^5  \end{cases}$
\end{itemize}
These three reciprocal cases are not useful in usual executions of the quantum switch, where they have zero amplitude. 
On all other basis states, $S^{\{2,5\}}$ acts as the identity. This fully defines $S^{\{2,5\}}$.

\subsubsection{Correctness}\label{annex:QS:correctness}

The evolution of the quantum switch is displayed in Fig. \ref{fig:qsaqc}. Since the operator $S^{\{2,5\}}$ is the most interesting gate operator, let us describe its action in detail. In the scattering after step $2$, it acts as in \eqref{eq:S25:s2}. In step $5$, it acts as in \eqref{eq:S25:s5}, which allows the remaining data to wait in the output space of sector 5. Most importantly, the control qubit flows back into the circuit, allowing the addresses to return to their initial positions in steps 6-9. In step 8, they ``merge'' back together, according to \eqref{eq:S25:s2s5}. The data then flows out of sector 5 into the outgoing buffer sector 6, and can be processed further by concatenated AQCs (see Sect. \ref{annex:concatenation}) as the quantum switch is not entangled with the data anymore ---see e.g. \ref{fig:qsaqc:8.5}.

That is, at the end of Fig. \ref{fig:qsaqc}, the data has been delivered, and the quantum switch has folded back into its initial configuration. In Step 9 the final state of the data space is $UV\ket{\psi}\ket1+VU\ket{\psi}\ket0$ (when omitting scalar coefficients). This matches the behaviour described in \eqref{eq:qswitch}, which proves correctness.

As a final remark, note that according to Thm. \ref{thm:isomorphisme} the nameblindness of $S^{\{2,5\}}$ only would have required commutation with renamings over the set of \emph{external} addresses $\{1,3,4,6\}$. Yet, we proved in Sect. \ref{nameblindqs} that it also holds for the addresses $\{2,5\}$.

\section{Named AQC}\label{annex:naqc}

To prove the results on the nameblindness of the main text, we disentangle the notions of physical locations and ``names'' as arbitrary labels of sectors. This amounts to extending the Hilbert spaces we work with by a ``name space'', i.e. an additional address space, which allows us to carefully distinguish between instances of AQCs with differently chosen labelings. There are also other ways to describe a geometry than resorting to arbitrary names, e.g. working modulo isomorphism which, however, results in signalling problems \cite{ArrighiNamesInQG}.

\subsection{Definition of an NAQC}\label{annex:naqc:def}
Here we define formally every notion used for \emph{Named AQC} (NAQC) which differs from that used for AQC.

\begin{definition}[Named sector space]\label{def:sectorspace NAQG}

Let $i$ be the index of a sector. The Hilbert space ${\hat \calH}^i
= \calH_\calT\otimes \calH_\calI\otimes \calH_\calO \otimes\calH_\calN$ is referred to as a \emph{named sector space}, with $\calH_\calN$ a copy of $\calH_\calA$ referred to as the \emph{name space}.
\end{definition}

\begin{definition}[Named circuit space]\label{def: circuit space NAQG}
The circuit space ${\hat \calH}$ of an NAQC is the subspace of $\bigotimes_{i \in \calA} {\hat \calH}^i$ such that each address appears at most once in the combined target, input and output spaces, and exactly once in the name spaces.
\end{definition}

In the NAQC formalism, target spaces $\calH_\calT$ no longer target physical locations (a position as factor in the tensor product  $\bigotimes_i {\hat \calH}^i$). Instead, it targets the name $a\in\calA$ of that physical location, as occurring in some name space $\calH_\calN$ of some sector. A gate $g=\{i_1,i_2,\ldots,i_k\}$, on the other hand, consists of the sectors at the physical locations $\{i_1,i_2,\ldots,i_k\}$.

\begin{definition}[Named gate space]\label{def NAQC: gate operator}
The \emph{named gate space} of a gate $g$ consists of $|g|$ named sector spaces, i.e. it is given by ${\hat \calH}^g=\bigotimes_{i\in g} {\hat \calH}^i$.
\end{definition}

\begin{definition}[Named gate operator]\label{def: gate operator NAQG}

A \emph{named gate operator} on $g\in\mathcal G$ is a unitary $U:{\hat \calH}^g\to{\hat \calH}^g$ which preserves the addresses of the target and stored spaces as in Def. \ref{def:gates}, and separately 
the name space stays unchanged.
\end{definition}

For any name $\mu\in\calA^-$, we denote by ${\hat\calH}_{\mu}^g := {\calH}^g\otimes\ket{\mu}_\calN$ the subspace of ${\hat\calH}^g$ spanned by the basis states that give the name $\mu$ to the gate $g$. Clearly, ${\hat \calH}^g$ is the smallest Hilbert space containing $\bigoplus\nolimits_{\mu}{\hat \calH}^g_{\mu}$. Moreover, since gate operators by definition preserve the name of the gate space, they are of the form $U^g=\bigoplus\nolimits_{\mu} U^g_\mu$, i.e. named gate operators are block-diagonal with blocks $U^g_\mu$ acting on and preserving ${\hat \calH}^g_\mu$.

\begin{definition}[NAQC transport]\label{def: Transport NAQG}
The named transport $T$ maps synchronously $\ket{a}^i_\calT  \ket{x}^i_\calO \otimes \ket{y}^j_\calI\ket{a}^j_\calN$ to  $\ket{a}^i_\calT \ket{y}^i_\calO  \otimes \ket{x}^j_\calI\ket{a}^j_\calN$, leaving the other sector constituents unchanged. That is, $T$ swaps the output space of a sector with the input space of the sector with the \emph{name} it targets, i.e. whenever a sector $i$ targets a name $a$, $i$'s output space is swapped with the input space of the sector with name $a$.
\end{definition}

\subsection{Renaming NAQCs}\label{par: useful subspaces}

Let $\mu$ be a name, i.e. an element in ${\cal A}^{-}$. We call a renaming $R$, i.e. a permutation of $\calA^-$
\begin{itemize}
    \item \textbf{external} for $\mu$, if $R$ acts as the identity over each address in $\mu$.
    \item \textbf{internal} for $\mu$, if $R$ acts as the identity over each address in $\cal A \setminus \mu$.
    \item \textbf{mixed} for $\mu$, if $R$ is a composition of disjoint swaps fully specified by two equal sized sets $(M_1,M_2)$ where $M_1$ takes its elements in $\mu$ and $M_2$ takes its elements in $\cal A \setminus \mu$. $R$ swaps the first element of $M_1$ for the natural order over $\mathbb{N}$ with the first element of $M_2$, the second with the second and so on). On elements of $\calA\setminus(M_1\cup M_2)$, $R$ acts as the identity.
\end{itemize}

\begin{figure}[H]
    \centering
    \begin{subfigure}[t]{\textwidth}
        \centering
        \resizebox{0.7\textwidth}{!}{\renewcommand{\eps}{\epsgr}\definecolor{box_color}{rgb}{1,1,1}
\definecolor{data_color}{RGB}{0,0,255}
\definecolor{dark_box_color}{rgb}{0,0,0}
\def\bleu#1{{\textcolor{data_color}{#1}}}

\newcommand*{\xMin}{0}%
\newcommand*{\xMax}{10}%
\newcommand*{\yMin}{0}%
\newcommand*{\yMax}{2}%

\newcommand{\adresseCircuitGrid}[4]{
    \draw[fill=box_color] (#1-0.45,#2-0.45) rectangle (#1+0.45,#2+0.45);
    \node[draw,fill=dark_box_color!10] at (#1,#2+0.18) {
     #3};
    \node (#3) at (#1,#2-0.24) {$\ket{ #4}_\calW$};
}

\begin{tikzpicture}
    \foreach \i in {\xMin,...,\xMax} {
        \draw [very thin,gray,dashed] (\i,\yMin) -- (\i,\yMax)  node [below] at (\i,\yMin) {$\i$};
    }
    \foreach \i in {\yMin,...,\yMax} {
        \draw [very thin,gray,dashed] (\xMin,\i) -- (\xMax,\i) node [left] at (\xMin,\i) {};
    }
    \adresseCircuitGrid{2}{1}{9}{\eps};
    \adresseCircuitGrid{4}{1}{2}{89};
    \adresseCircuitGrid{5}{1}{5}{24};
    \adresseCircuitGrid{9}{1}{8}{\eps};
    \adresseCircuitGrid{8}{1}{4}{\eps}
\draw[dashed,red] (3.5,0.5) rectangle (5.5,1.5);
    
\end{tikzpicture}\renewcommand{\eps}{\epsgr}}
        \caption{An NAQC $\dots$}
        \label{figs/NAQC_1}
    \end{subfigure}
    
    \begin{subfigure}[t]{\textwidth}
        \centering
        \resizebox{0.7\textwidth}{!}{\renewcommand{\eps}{\epsgr}\definecolor{box_color}{rgb}{1,1,1}
\definecolor{data_color}{RGB}{0,0,255}
\definecolor{dark_box_color}{rgb}{0,0,0}
\def\bleu#1{{\textcolor{data_color}{#1}}}

\newcommand*{\xMin}{0}%
\newcommand*{\xMax}{10}%
\newcommand*{\yMin}{0}%
\newcommand*{\yMax}{2}%

\newcommand{\adresseCircuitGrid}[4]{
    \draw[fill=box_color] (#1-0.45,#2-0.45) rectangle (#1+0.45,#2+0.45);
    \node[draw,fill=dark_box_color!10] at (#1,#2+0.18) {
     #3};
    \node (#3) at (#1,#2-0.24) {$\ket{ #4}_\calW$};
}

\begin{tikzpicture}
    \foreach \i in {\xMin,...,\xMax} {
        \draw [very thin,gray,dashed] (\i,\yMin) -- (\i,\yMax)  node [below] at (\i,\yMin) {$\i$};
    }
    \foreach \i in {\yMin,...,\yMax} {
        \draw [very thin,gray,dashed] (\xMin,\i) -- (\xMax,\i) node [left] at (\xMin,\i) {};
    }
    \adresseCircuitGrid{2}{1}{2}{\eps};
    \adresseCircuitGrid{4}{1}{4}{92};
    \adresseCircuitGrid{5}{1}{5}{48};
    \adresseCircuitGrid{9}{1}{9}{\eps};
    \adresseCircuitGrid{8}{1}{8}{\eps}
\draw[dashed,red] (3.5,0.5) rectangle (5.5,1.5);

\end{tikzpicture}\renewcommand{\eps}{\epsgr}}
        \caption{$\dots$ renamed so that each gate space is named with the canonical name.}
        \label{figs/NAQC_2}
    \end{subfigure}
    
    \caption{An NAQC represented over a grid. Here $\calA=\{2,4,5,8,9\}$. Inside of each sector (big white boxes) we can find the content of the stored address space. Gray boxes represent the content of the name space.}
    \label{figs/circuit_on_a_grid}
\end{figure}
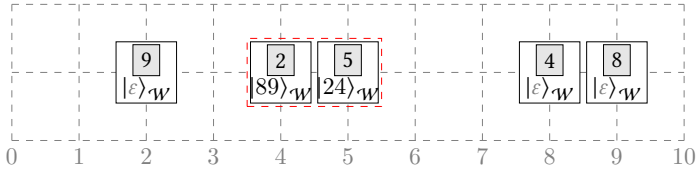
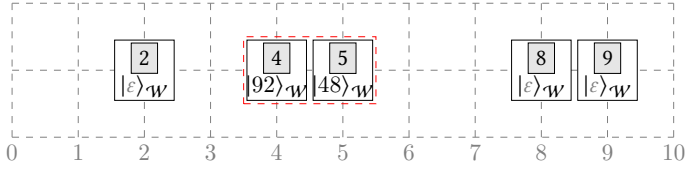

\begin{example}
    Consider the circuit in Fig. \ref{figs/NAQC_1}. We focus here on the gate space \{4,5\} which has the name $25$. Then
    \begin{itemize}
        \item $R_{4\leftrightarrow 8}$ is \textbf{external} because it only swaps external addresses.
        \item $R_{2 \leftrightarrow 5}$ is the only non-trivial \textbf{internal} renaming for this state.

        \item $R_{5\leftrightarrow 9,2\leftrightarrow 4}$ is \textbf{mixed} because it swaps external with internal addresses in an ordered fashion.
        \item $R_{5\leftrightarrow 4,2\leftrightarrow 9}$ is neither internal nor mixed nor external.
    \end{itemize}
    
\end{example}

Distinguishing these particular types of renamings is useful because any renaming $R$ can be uniquely decomposed into a mixed, an internal and an external renaming with respect to a given name $\mu$. The basic idea to construct this decomposition is to first apply a mixed renaming to bring all the letters of $R(\mu)$ into the gate space, then use an internal renaming to ensure that the gate space is named $R(\mu)$, and finally adjust the addresses that are not in $\mu$ with an external renaming. Let us define $R_{\mu\to\nu}$ the renaming changing the word $\mu$ to $\nu$, and not changing other letters.

\begin{proposition}[Decomposition of renamings]\label{prop: decomposition}~
\begin{enumerate}
\item Let $\mu$ and $\nu$ be names of same length. There is a unique renaming $R_{\mu\to\nu}=IM$ such that $M$ is mixed for $\mu$ and $I$ is internal for $M(\mu)$.
\item Let $R$ be a renaming and $\mu$ be an element of ${\cal A}^-$ ---see Def. \ref{def:addresses}. There exists exactly one pair of renamings such that:
\begin{equation*}
	R = ER_{\mu\to\nu}
\end{equation*}
where $\nu=R(\mu)$ and $E$ is external for $\nu$.
\end{enumerate}
\end{proposition}
\begin{proof}\label{proof: decomposition} For any name $\mu=\mu_1\mu_2\ldots\mu_k\in {\cal A}^-$ we write $\SET(\mu)=\{\mu_1,\mu_2,\ldots,\mu_k\}$ and note that it is invariant under renamings that are external and internal for $\mu$.

(1) Let $R:\cal A \to \cal A$ be a renaming, and take $M$ as the renaming that swaps elements of 
\begin{equation*}
	\SET(\mu)\setminus \SET(R(\mu))
\end{equation*}
with elements of 
\begin{equation*}
	\SET(R(\mu))\setminus \SET(\mu)
\end{equation*}
and is mixed for $\mu$. 
Note that $\SET(M(\mu))=\SET(R(\mu))$, and that $M$ is unique since we have that $\SET(M(\mu))=\SET(EIM(\mu))$. Similarly, take $I$ as the renaming such that for every integer $0<k<|\mu|+1$
\begin{align*}
  I \colon \calA \to \calA,\qquad M(\mu)_k  \mapsto R(\mu)_k,
\end{align*}
and which is internal for $M(\mu)$. Since $IM(\mu)=EIM(\mu)$ for all $E$ external for $ IM(\mu)$, $I$ is the only possible choice. Together, this gives a unique decomposition $R_{\mu\to\nu}=IM$.

(2) This follows from defining $E=RM^{-1}I^{-1}$ which is unique since it is the only renaming that guarantees $R=EIM$. Moreover, $E$ is (trivially) external for $IM(\mu)$, since $EIM(\mu)=R(\mu)=IM(\mu)$, where the last equality is by construction of $I$ in the proof of (1).
\end{proof}

\begin{example}[Decomposition of renamings]
    Consider the renaming $R_{2\rightarrow 9 \rightarrow 8,4\leftrightarrow 5}$ on the circuit of Fig. \ref{figs/NAQC_1}, which maps the gate space state $\ket{2}_\calN\ket{89}_\calW^{4}\ket{5}_\calN\ket{24}_\calW^{5}$ to $\ket{9}_\calN\ket{28}_\calW^{4}\ket{4}_\calN\ket{95}_\calW^{5}$. We decompose it with respect to the name $\mu=25$:
    \begin{itemize}
        \item Since $R(\mu)=94$, the mixed renaming constructed in the above proof is $M=R_{5\leftrightarrow 9,2\leftrightarrow 4}$, which results in the state $\ket{4}_\calN\ket{85}_\calW^{4}\ket{9}_\calN\ket{42}_\calW^{5}$.
        \item Then, $I=R_{9\leftrightarrow 4}$ is the only possible internal renaming that reorders the name of the gate space, which is thus in the state 
        $\ket{9}_\calN\ket{85}_\calW^{4}\ket{4}_\calN\ket{92}_\calW^{5}$.
        \item Finally, $E=R_{2\rightarrow 5 \rightarrow 8}$ reorders the external addresses to match the desired ones.
    \end{itemize}
\end{example}

Next, we relate AQCs to NAQCs with specifically chosen gate names, called the canonical name ---see Fig. \ref{figs/NAQC_2}. 
To this end, consider a named gate space ${\hat \calH}^g$ with $k=|g|$ sectors. Its basis states are of the form 
\begin{equation*}
    \bigotimes_{\substack{i\in g\\i\hspace{1pt}\text{increasing}}}\ket{\overline{a}_i\overline{q}_i}\ket{\mu_i}_\calN^i
\end{equation*}
This state gives $k$ names $\mu_i$ to the sectors inside the gate space, thus we say that this state  gives the \textbf{name} $\mu=\mu_{i_1}\mu_{i_2}\ldots\mu_{i_k}$  to the gate $g={i_1 i_2\ldots i_k}$ with $i_1<i_2<\ldots<i_k$.

\begin{definition}[Canonical name]\label{def: canonical name}
Consider a gate $g=\{i_1,i_2,i_3,\dots,i_k\}$. The \emph{canonical name} of $g$ is $e=i_1i_2i_3\dots i_k$ with $i_1<i_2<\dots <i_k$. The basis states of ${\cal H'}^g_e$ take the form:
\begin{equation*}
    \bigotimes_{i\in g}\ket{\overline{a}_i   \overline{q}_i}\ket{i}^i
\end{equation*}
\end{definition}

By this definition, the basis states of ${\cal H'}^g_e$ are in one-to-one correspondence with those of the gate space $\calH^g$ of an AQC (Sect. \ref{Section:model}), i.e.
\begin{equation*}
	\bigotimes_{i\in g}\ket{\overline{a}_i   \overline{q}_i}^i.
\end{equation*}
Choosing a canonical name means choosing a representative for each equivalence class w.r.t. renamings --- see \eqref{eq:equivalence}. 
Thus the subspace $\bigotimes_{i\in g} {\cal H'}^g_{e(g)}$ of the named circuit space ${\hat \calH}$ is isomorphic to the circuit space $\calH$.

Finally, we describe those gate operators which act the same way whatever the name of the gate is. This reflects the idea that the action at a physical location (a house) should not depend on the arbitrary choice of name for that location (the number on the house), which is fiducial. 
Since such gate operators do not depend on the internal addresses in the name spaces of gate space, we call them ``internally-blind'':

\begin{definition}[Internally-blind gate operators]\label{def: Internally-blind gate operators}
Let $g$ be a gate with $k$ sectors and canonical name $e\in {\cal A}^-$. The named gate operator $U^g$ is called \textbf{internally-blind} if for any $\mu\in {\cal A}^-$ with $|\mu|=k$ and for any state $\ket{\phi}\in{\cal H'}^g_{\mu}$:
\begin{equation*}
    U^g_{\mu}\ket{\phi} = R_{e\to\mu}U_{e}^g R_{e\to\mu}^{-1}\ket{\phi}.
\end{equation*}
\end{definition}

Internally-blind named gate operators are completely determined by their block $U^g_{e}$. In the next section we show that renaming invariance imposes further constraints on this block.

\subsection{Proof of Theorem \ref{thm:isomorphisme}}

The concept of NAQC introduced above allows us to prove that nameblind AQCs and nameblind NAQCs are isomorphic. To do so, we show that nameblindness of an NAQC is equivalent to it being internally-blind, i.e. fully determined by its action over the canonical name (see Def. \ref{def: Internally-blind gate operators}), and commuting with renamings that are external for the canonical name.

We show that internally-blind operators that commute with external renamings are nameblind:

\begin{lemma}[Commutation with mixed and internal renamings]\label{lemme: commutation with mixed and internal renamings}
Let $U^g$ be an internally-blind gate operator, $e$ the canonical name of $g$, and $\mu\in{\calA}^-$ a name of size $|g|$. If $U^g_{e}$ commutes with every renaming external to $e$, then for all $\nu\in\calA^-$ with $|\nu|=|\mu|$ and any $\ket{\phi}\in{{\hat \calH}}_\mu^g$
\begin{equation*}
	U^g R_{\mu\to\nu}\ket{\phi}=R_{\mu\to\nu}U^g\ket{\phi}.
\end{equation*}
\end{lemma}

\begin{proof}
Let $R_{\mu\to\nu}$ as in Prop. \ref{prop: decomposition}. Let $E$ be the renaming defined by
\begin{equation*}
	E=R_{e\to\nu}^{-1}R_{\mu\to\nu}R_{e\to\mu}
\end{equation*}
Since $Ee=e$, it is external for $e$.
We have:
\begin{center}
$\begin{array}{rlll}
    U^{g}R_{\mu\to\nu}\ket{\phi} & =U_\nu^g R_{\mu\to\nu}\ket{\phi} &\hspace{1pt} &\text{(Diagonal blockwise)}\\
    & =R_{e\to\nu}U_e^g R_{e\to\nu}^{-1}R_{\mu\to\nu}\ket{\phi} & & \text{(Internally-blind)}\\
    & =R_{e\to\nu}U_e^g E R_{e\to\mu}^{-1}\ket{\phi} &  &\text{(Definition of $E$)}\\
    & =R_{e\to\nu}EU_e^g R_{e\to\mu}^{-1}\ket{\phi} & & \text{(Hypothesis)}\\
    & =R_{\mu\to\nu}R_{e\to\mu}U_e^g R_{e\to\mu}^{-1}\ket{\phi} & &\text{(Definition of $E$)}\\
    & =R_{\mu\to\nu}U_\mu^g \ket{\phi} &  &\text{(Internally-blind)}\\
    & =R_{\mu\to\nu}U^{g}\ket{\phi}& & \text{(Diagonal blockwise)}
\end{array}$\\
\end{center}

\end{proof}

\begin{lemma}[Commutation with external renamings]\label{lemme: commutation with external renamings}
Let $U^g$ be a internally-blind gate operator, $e$ the canonical name of $g$, and $\mu\in\calA^-$ a name of size $|g|$. If $U_{e}^g$ commutes with every external renaming, then for all external renamings $E$ and all $\ket{\phi}\in{{\hat \calH}}_\mu^g$:
\begin{equation*}
	U^gE\ket{\phi}=EU^g\ket{\phi}.
\end{equation*}
\end{lemma}

\begin{proof}
Let $E$ be a renaming external for $\mu$. Let $E'$ be the renaming defined by:
\begin{equation*}
	E'=R_{e\to\mu}^{-1}ER_{e\to\mu}
\end{equation*}
$E'$ is external for $e$ so we have for each state $\ket{\phi}\in {\cal H'}_\mu^g$:

\begin{alignat*}{2}
    U^{g}E\ket{\phi} & =U_\mu^g E\ket{\phi} &\\
    & =R_{e\to\mu}U_e^g R_{e\to\mu}^{-1}E\ket{\phi} &\qquad\text{(Internally-blind)}\\
    & =R_{e\to\mu}U_e^gE' R_{e\to\mu}^{-1}\ket{\phi} &\qquad\text{(Definition of $E'$)}\\
    & =R_{e\to\mu}E'U_e^g R_{e\to\mu}^{-1}\ket{\phi} &\qquad\text{(Hypothesis)}\\
    & =ER_{e\to\mu}U_e^g R_{e\to\mu}^{-1}\ket{\phi} &\\
    & =EU^{g}_\mu\ket{\phi}& & 
\end{alignat*}
\end{proof}

\begin{theorem}[Characterization of renaming invariance]\label{Theorem: Characterization renaming invariance}
Let ${\hat \calH}^g$ be a named gate space, $U^g$ be a named gate operator and $e$ the canonical name of $g$. $U^g$ is renaming invariant if and only if $U^g$ is internally-blind and $U_{e}^g$ commutes with every renaming $E$ external for $e$, i.e $U_{e}^gE=EU_{e}^g$.
\end{theorem}
\begin{proof}
$[\Rightarrow]$ : 
Let $U^g$ be a nameblind gate operator. Since by hypothesis $U^g$ commutes with all renamings, it obviously commutes with the external ones, and so does the block $U^g_e$. To show that $U^g$ is internally-blind, consider a renaming $R_{e\to\mu}$. Renaming invariance of $U_g$ implies that $U^{g}R_{e\to\mu} = R_{e\to\mu}U^{g}$, and hence for every name $\mu$ of $g$ and for each state $\ket{\phi}\in{\hat \calH}_\mu$ we have
\begin{equation*}
	U_{{\mu}}^g\ket{\phi} = R_{e\to\mu}U_{{e}}^gR_{e\to\mu}^{-1}\ket{\phi},
\end{equation*}
which is a restatement of internal-blindness.

$[\Leftarrow]$ : 
Let $U^g$ be an internally-blind gate operator such that $U^g_{e}$ commutes with all external renamings. Consider a renaming $R$ and a name $\mu\in\calA^-$ of size $|g|$. Then, for every state $\ket{\phi}\in{\hat \calH}_\mu$:
\begin{alignat*}{2}
    U^gR\ket{\phi} & =U^gEIM\ket{\phi} &\qquad\text{(by Prop. \ref{prop: decomposition})}\\
    & =EU^gIM\ket{\phi} &\qquad\text{(by Lem. \ref{lemme: commutation with external renamings})}\\
    & =EIMU^g\ket{\phi} &\qquad\text{(by Lem. \ref{lemme: commutation with mixed and internal renamings})}\\
    & =RU^g\ket{\phi}, &
\end{alignat*}
i.e. $U^g$ is renaming invariant.
\end{proof}

Thm. \ref{Theorem: Characterization renaming invariance} states that the named gate operator $U^g$ of a nameblind NAQC is fully characterized by its action $U^g_{e}$ over the gates with canonical names, and that these must commute with external renamings. From Def. \ref{def: canonical name} it follows that the evolution of an NAQC over gates with canonical names is isomorphic to that of an AQC. Thus, the set of nameblind NAQCs is isomorphic to the set of AQCs whose gate operators commute with external renamings. Therefore, Thm. \ref{thm:isomorphisme} follows from Thm. \ref{Theorem: Characterization renaming invariance}. 

\begin{proof}[Proof of Thm. \ref{thm:isomorphisme}]
Let us consider a nameblind NAQC  $((\calA,\calG,\calQ,S,\calN),\ket{\psi}\ket{e})$, its scattering $S$ commutes with every renaming $R$, i.e. $RS=SR$. The isomorphic nameblind Addressable Quantum Circuit is $((\calA,\calG,\calQ,U),\ket{\psi})$, where $U$ is defined as a scattering where each gate operator $U^g$ is of the form $\bigoplus_\calN U^g_e$ of Thm. \ref{Theorem: Characterization renaming invariance}. Reciprocally, given a nameblind AQC, as each gate operator $S^g$ commutes with external renamings, $U^g = \bigoplus_\calN S^g$ is a nameblind named gate operator that can be used to specify a nameblind NAQC.
\end{proof}

\section{Characterization of nameblind operators}\label{annex:nameblind}

Nameblind gate operators consist of blocks of nameblind matrices (see Prop. \ref{prop:gateopdecomp}), which will be characterized inductively in this appendix.
Below, when representing such operators over a list of $n$ external addresses as matrices, we implicitly use the basis of words sorted in lexicographically ascending order, i.e.
\begin{equation}\label{eq:basis}
    \{12\dots(n-1)n,12\dots n(n-1),\dots,n(n-1)\dots21\}.
\end{equation}
In this basis, the first $(n-1)!$ words begin with letter $1$, the next $(n-1)!$ words begin with letter $2$ etc. We often partition matrices acting on this $n!$-dimensional Hilbert space into $n^2$ blocks, each acting between two $(n-1)!$-dimensional Hilbert spaces, each being spanned by words starting with the same letter.

\subsection{Adjacent transpositions}\label{annex:rkktranspo}

Let $n$ be an integer and let us consider $k\in\{1,\dots,n-1\}$. We call the renaming $R_k$ which swaps $k$ and $k+1$ and acts as the identity otherwise an \emph{adjacent transposition} on words of $n$ addresses at $k$. Adjacent transpositions generate the group of renamings, such that any renaming can be written as a product of such adjacent transpositions. Commuting with all adjacent transpositions is therefore equivalent to commuting with all renamings.
\begin{lemma}\label{lemme:rkktranspo}
In the basis \eqref{eq:basis} the matrix of an adjacent transposition $R_{k}$ has the form 
\begin{equation}\label{eq:Rk}
    R_{k}=
        \begin{blockarray}{ccccccccc}
 & & & k & k+1 & && \\
 & & & \downarrow & \downarrow & && \\
\begin{block}{(cccccccc)l}
        M_{k-1} &0&\dots&&&&&&\\
       0&\ddots &\ddots&&&&&&\\
        \vdots&\ddots&M_{k-1}&&&&&& \\
        &&&0&I&&&&\leftarrow k \\
        &&&I&0&&&&\leftarrow k+1 \\
        &&&&&M_k&\ddots&\vdots &\\
        &&&&&\ddots&\ddots&0& \\
        &&&&&\hdots&0&M_k&\\
\end{block}\end{blockarray}
\end{equation}
where $M_i$ is the adjacent transposition matrix on  words of  $n-1$ addresses at $i$. 
\end{lemma}
\begin{proof}

We consider the lines before and after the $k$th line separately: The addresses after the $k$th block line correspond to $n$ letter words that begin with a letter bigger than $k$. Amongst the remaining $n-1$ letters, $k$ is still the $k$th letter in ascending lexicographic order. Thus, to transpose $k$ and $k+1$ we must apply $M_k$.

The lines before the $k$th line, correspond to $n$ letter words starting with a letter smaller than $k$. Thus, amongst the remaining $n-1$ letters, $k$ is still the $(k-1)$-th letter in ascending lexicographic order. Therefore, the correct matrix to apply is $M_{k-1}$.
\end{proof}

\subsection{Nameblind operators}

We now prove Theorem \ref{thm:nameblind}, which we restate here for the convenience of the reader :
\begin{manualtheorem}{\ref*{thm:nameblind}}\label{nameblind matrices}
$(n,n)$-nameblind matrices have exactly the form
    \begin{equation}\label{eq:nameblind}
        A = \begin{pmatrix}
        D&B&M_1B&M_2M_1B&M_3M_2M_1B&\dots\\
        B&D&M_1BM_1&M_2M_1BM_1&\dots&\\
        BM_1&M_1BM_1&D&&& \\
        BM_1M_2&M_1BM_1M_2&&\ddots&& \\
       BM_1M_2M_3&\vdots&&&&\\
       \vdots&&&&&
        \end{pmatrix}
    \end{equation}
        where $D$ is a $(n-1,n-1)$-nameblind matrix 
        and $B$ is any matrix which commutes with $M_i$ for each $i$ in $\{ 2,3,...,n-1\}$, whose characterization is given in Prop. \ref{thm:partially-nameblind}
\end{manualtheorem}

\begin{proof}
Let $A$ be a matrix of size $n!\times n!$ written blockwise with blocks of size $(n-1)!\times (n-1)!$ as
\begin{equation*}
	A = 
        \begin{pmatrix}
        A_{1,1}& A_{1,2}&\dots & A_{1,k}&  A_{1,k+1} &\dots&A_{1,n}\\
        \vdots&\vdots&\ddots&{\vdots}&{\vdots}&&\vdots \\
        {A_{k,1}}&{A_{k,2}}&{\dots}&{A_{k,k}}&{A_{k,k+1}}&{\dots}&{A_{k,n}} \\
        {A_{k+1,1}}&{A_{k+1,2}}&{\dots}&{A_{k+1,k}}&{A_{k+1,k+1}}&{\dots}&{A_{k+1,n}} \\
        \vdots&&&{\vdots}&{\vdots}&\ddots&\vdots \\
        A_{n,1}&A_{n,2}&\hdots&{ A_{n,k}}&{ A_{n,k+1}}&\hdots&A_{n,n}
        \end{pmatrix}.
\end{equation*}
Since every renaming can be written as a product of adjacent transpositions, commuting with all renamings is equivalent to commuting with all adjacent transpositions. Thus, $A$ is nameblind iff $R_kA=AR_k$ for all $k\in\{1 \dots n-1\}$ with $R_k$ as in \eqref{eq:Rk}, i.e. iff the following matrices are equal for every $k$:
\begin{equation*}
	R_{k}A=
        \begin{pmatrix}
        M_{k-1} A_{1,1}&M_{k-1} A_{1,2}&\dots&\lowgreen{M_{k-1} A_{1,k}}&\lowgreen{M_{k-1} A_{1,k+1}}&\dots&M_{k-1}A_{1,n}\\
        M_{k-1} A_{2,1}&M_{k-1} A_{2,2}&\dots&\lowgreen{M_{k-1} A_{2,k}}&\lowgreen{M_{k-1} A_{2,k+1}}&&\vdots\\
        \vdots&\vdots&&\lowgreen{\vdots}&\lowgreen{\vdots}&&\vdots \\
        \lowgreen{A_{k+1,1}}&\lowgreen{A_{k+1,2}}&\lowgreen{\dots}&\blacier{A_{k+1,k}}&\blacier{A_{k+1,k+1}}&\lowgreen{\dots}&\lowgreen{A_{k+1,n}} \\
        \lowgreen{A_{k,1}}&\lowgreen{A_{k,2}}&\lowgreen{\dots}&\blacier{A_{k,k}}&\blacier{A_{k,k+1}}&\lowgreen{\dots}&\lowgreen{A_{k,n}} \\
        \vdots&&&\lowgreen{\vdots}&\lowgreen{\vdots}&&\vdots \\
        M_kA_{n,1}&\hdots&\hdots&\lowgreen{M_k A_{n,k}}&\lowgreen{M_k A_{n,k+1}}&\hdots&M_kA_{n,n}
        \end{pmatrix}
\end{equation*}        

\begin{equation*}
	AR_{k }=
        \begin{pmatrix}
        A_{1,1}M_{k-1}&A_{1,2}M_{k-1}&\dots& \lowgreen{A_{1,k+1}}& \lowgreen{A_{1,k}}&\dots&A_{1,n}M_k\\
        A_{2,1}M_{k-1}&A_{2,2}M_{k-1}&\dots&\lowgreen{A_{2,k+1}}& \lowgreen{A_{2,k}}&&\vdots\\
        \vdots&\vdots&&\lowgreen{\vdots}&\lowgreen{\vdots}&&\vdots \\
        \lowgreen{A_{k,1}M_{k-1}}&\lowgreen{A_{k,2}M_{k-1}}&\lowgreen{\dots}&\blacier{A_{k,k+1}}&\blacier{A_{k,k}}&\lowgreen{\dots}&\lowgreen{A_{k,n}M_k} \\
        \lowgreen{A_{k+1,1}M_{k-1}}&\lowgreen{A_{k+1,2}M_{k-1}}&\lowgreen{\dots}&\blacier{A_{k+1,k+1}}&\blacier{A_{k+1,k}}&\lowgreen{\dots}&\lowgreen{A_{k+1,n}M_k} \\
        \vdots&&&\lowgreen{\vdots}&\lowgreen{\vdots}&&\vdots \\
        A_{n,1}M_{k-1}&\hdots&\hdots&\lowgreen{A_{n,k+1}}&\lowgreen{A_{n,k}}&\hdots&A_{n,n}M_k
        \end{pmatrix}
\end{equation*}
This gives three types of conditions which we distinguish by their color in the matrices above:
        
\paragraph*{\blacier{Blue conditions}}\label{blue}
For each  $i\in\{1 ,\dots, n-1\}$:
\begin{itemize}
 \item[] $A_{i,i} = A_{i+1,i+1}$, i.e. all diagonal blocks are equal.
 \item[] $A_{i,i+1} = A_{i+1,i}$, i.e. the matrix is blockwise-symmetric near the diagonal.  
\end{itemize}
We will henceforth write the diagonal blocks as $D:=A_{1,1}$ and we set $B:=A_{1,2}$.

\paragraph*{\lowgreen{Green conditions}}\label{green}
Consider a row $i$, a column $j$ and an adjacent transposition $R_{k}$. The green conditions from commutation with $R_k$ are equivalent to:

\begin{itemize}
\item[] Above the diagonal, i.e. $i<k<k+1<j$:
\begin{itemize}
 \item $A_{i,k+1}=M_{k-1}A_{i,k}$, i.e one step to the right corresponds to right multiplication with .
 \item $A_{k+1,j}=A_{k,j}M_k$, i.e the $(k+1)$-th element of column $j$ is equal to $M_{k}$ times the $k$-th element.
\end{itemize}
\item[] Below the diagonal, i.e. $j<k<k+1<i$:
\begin{itemize}
 \item $A_{i,k+1}=M_kA_{i,k}$, i.e. the $(k+1)$-th element of line $i$ is equal to $M_{k}$ times the $k$-th element.
 \item $A_{k+1,j}=A_{k,j}M_{k-1}$, i.e. the $(k+1)$-th element of column $j$ is equal to $M_{k-1}$ times the $k$-th element.
\end{itemize}
\end{itemize}
This has to hold for all $k$ which gives the general form of $A$ given in Thm. \ref{nameblind matrices}, because we can deduce all blocks from $A_{1,1}$ and $A_{1,2}$, but the conditions on the blocks are yet to be derived.
         
\paragraph*{Black conditions}\label{black}
From the diagonal conditions in the top-left and bottom-right quadrant we immediately deduce that $D$ commutes with all $M_i$ for $0<i<n-1$, i.e. $D$ is nameblind. Similarly, $B$ commutes with all $M_i$ for $2<i<n-1$. Thus, in the basis \eqref{eq:basis} $A$ has the form \eqref{eq:nameblind}.

\indent Let us now show the converse, i.e. that any matrix $A$ as specified in Thm. \ref{nameblind matrices} is nameblind. By definition, it satisfies the blue and green conditions, and it only remains to check that it also meets the requirements of the black conditions. To this end, consider a block $A_{i,j}$ above the diagonal, i.e. $i<j$, which is symmetric with the case where it is below. By \eqref{eq:nameblind} this block is given by
\begin{equation*}
	A_{i,j}=M_{j-2}\dots M_1BM_1\dots M_{i-1}.
\end{equation*}
To satisfy the top-left and bottom-right conditions, this block has to commute with $M_{k-1}$ for all $k$ such that $j< k <n$, and with $M_{k'}$ for all $k'$ such that $k'+1 < i$, i.e. it has to commute with $M_k$ for all $k\in K:=\{1,\dots,i-2\}\cup\{j,\dots,n-2\}$. Explicitly, we need to show that for $k\in K$
\begin{equation*}
    M_{j-2}\dots M_1BM_1\dots M_{i-1}\RC{M_k} = \blacier{M_{k}}M_{j-2}\dots M_1BM_1\dots M_{i-1}.
\end{equation*}
Rearranging this we obtain
\begin{align*}
    B&=M_1\dots M_{j-2}\blacier{M_{k}}M_{j-2}\dots M_1BM_1\dots M_{i-1}\RC{M_k}M_{i-1}\dots M_1\\
    \Leftrightarrow \qquad B&=M_1\dots M_{j-2}\blacier{M_{k}}M_{j-2}\dots M_1M_1\dots M_{i-1}\RC{M_k}M_{i-1}\dots M_1B\\
    \Leftrightarrow \qquad B&=M_1\dots M_{j-2}\blacier{M_{k}}M_{j-2}\dots M_{i}\RC{M_k}M_{i-1}\dots M_1B\\
    \Leftrightarrow \qquad B&=M_1\dots M_{j-2}\blacier{M_{k}}\RC{M_k}M_{j-2}\dots M_{i}M_{i-1}\dots M_1B\\
    \Leftrightarrow \qquad B&=M_1\dots M_{i-1}M_{i-1}\dots M_1B
\end{align*}
where the first equivalence comes from premise that $B$ commutes with renamings that preserve the letter $1$ (since $\{M_2,\dots,M_{n-1}\}$ generates this family of renamings) which $M_1\dots M_{i-1}\RC{M_k}M_{i-1}\dots M_1$ does since $k$ is different from $i$ and $i-1$, and the third comes from the fact that two transpositions with disjoint support commute. Moreover, since $M_jM_j=I$ the right side of the last equation indeed gives $B$.
        
To verify the top-right and bottom-left conditions, consider a block $A_{i,j}$ with $i<j$, and a $i<k<j-1$. Since the conditions for this quadrant are symmetric, this is sufficient. This block has to satisfy $A_{i,j}\RC{M_{k}}=\blacier{M_{k-1}}A_{i,j}$. We calculate:
\begin{align*}
	A_{i,j}\RC{M_{k}}&=M_{j-2}\dots M_1BM_1\dots M_{i-1}\RC{M_{k}}&(1)\\
	&=M_{j-2}\dots M_1B\RC{M_{k}}M_1\dots M_{i-1}&(2)\\
	&=M_{j-2}\dots M_{k+1}M_{k}M_{k-1}\RC{M_{k}}M_{k-2}\dots M_1BM_1\dots M_{i-1}&(3)\\
	&=M_{j-2}\dots M_{k+1}\blacier{M_{k-1}}M_{k}M_{k-1}M_{k-2}\dots M_1BM_1\dots M_{i-1}&(4)\\
	&=\blacier{M_{k-1}}M_{j-2}\dots M_{k+1}M_{k}M_{k-1}M_{k-2}\dots M_1BM_1\dots M_{i-1}&(5)\\
	&=\blacier{M_{k-1}}A_{i,j}&(6)
\end{align*}
where as above we used that transpositions with disjoint support commute and that $BM_k=M_kB$ since $k>1$. The non-trivial step in the forth equality is to realize that ${M_{k-1}}M_{k}M_{k-1}=M_{k}M_{k-1}M_{k}=(k-1\quad k+1)$.
       
Thus, all black conditions are satisfied by the blocks of $A$ of the general form given in Thm. \ref{nameblind matrices}, which implies that $A$ indeed is nameblind.
\end{proof}

\subsection{Partially-nameblind operators}\label{annex:partially-nameblind}

The matrix $B$ of Thm. \ref{thm:nameblind} is ``partially-nameblind'' in the sense that it commutes with all renamings which preserve the first $p$ addresses in lexicographical order. We denote the set of such matrices by $\calB^n_{p}$, where $n$ indicates the total number of addresses.

\begin{proposition}[Characterization of partially-nameblind matrices]\label{thm:partially-nameblind}
partially-nameblind matrices $A\in\calB^n_{p}$ are exactly those matrices which can be written as $A=$
\begin{equation*}        
	\begin{pmatrix}
        A_{1,1}& 
        \dots &A_{1,p}&  B_{1}^l &M_pB_{1}^l&M_{p+1}M_pB_{1}^l&\dots\\
        A_{2,1}& \dots&A_{2,p}&B_{2}^l&M_pB_{2}^l&M_{p+1}M_pB_{2}^l&\dots\\
        \vdots& &\ddots&&\vdots&&& \\
        {A_{p,1}}& \dots&A_{p,p}&B_p^l&M_pB_p^l&M_{p+1}M_pB_p^l&\dots \\
        {B_{1}^c}& \dots&B_p^c&D&B&M_{p+1}B&\dots  \\
        {B_{1}^c}M_p& \dots&B_p^cM_p&B&D&M_{p+1}BM_{p+1}&{\dots} \\
        {B_{1}^c}M_pM_{p+1}& \dots&B_p^cM_pM_{p+1}&BM_{p+1}&M_{p+1}BM_{p+1}&D& \\
        \vdots&&&BM_{p+1}M_{p+2}&\vdots&&\ddots\\
        &&&\vdots&&&
     \end{pmatrix}
\end{equation*}
where $A_{i,j}\in \calB^{n-1}_{p-1}$, $B_i^l$, $B_i^c$ and $D$ are in $\calB^{n-1}_{p}$, and $B$ is in $\calB^{n-1}_{p+1}$.

\end{proposition}

\begin{proof}
    As above, we use blue, green and top-left black conditions, which for matrices in $\calB^n_{p}$ however have to be satisfied only for $p<k$.

    Conversely, we have to check that matrices of this form satisfy the black conditions. First, we consider blocks $A_{i,j}$ in the top-left or bottom-right corner. If $i<p+1$ and $j<p+1$ we just have to check that $A_{i,j}$ is in $\calB^{n-1}_{p-1}$, which is the case by definition. 
        
    If $i>p$ and $j>p$, the proof of Thm. \ref{nameblind matrices} applies. Thus, we only have to consider the case $i<p+1$, $j>p$ since the other one follows by symmetry. 
    In this case the block is given as
    \begin{equation*}
       	A_{i,j}=M_{j-2}\dots M_p B_i^l,
    \end{equation*}
    and we have to verify that it commutes with $M_{k-1}$ for all $k$ such that $j<k<n$, and with $M_{k'}$ for all $k'$ such that $p+1<k'+1<i$, i.e. that it commutes with $M_k$ for each $k\in \{j,n-2\}$ as $\{p+1,\dots,i-2\}$ is empty. Since these transpositions have disjoint support, this is straightforward to see. 
        
    Next, pick $A_{i,j}$ in the top-right or bottom-left corner. 
        Since the top-right and bottom-left conditions are symmetric it is sufficient to consider the case $i<j$. Moreover, if $i>p$ again the previous proof applies, so we only consider $i<p+1<p+2<j$. We have to show that $A_{i,j}\RC{M_{k}}=\blacier{M_{k-1}}A_{i,j}$ for all $k<j-1$. Indeed,
        \begin{align*}
            A_{i,j}\RC{M_{k}}&=M_{j-2}\dots M_p B_i^l\RC{M_{k}}\\
            &=M_{j-2}\dots M_k M_{k-1}\RC{M_{k}}\dots M_p B_i^l\\
            &=M_{j-2}\dots \blacier{M_{k-1}}M_{k}M_{k-1}\dots M_p B_i^l\\
            &=\blacier{M_{k-1}}M_{j-2}\dots M_{k}M_{k-1}\dots M_p B_i^l\\
            &=\blacier{M_{k-1}}A_{i,j}
        \end{align*}
        
        This ends the proof that all black conditions are satisfied by matrices of the general form given in Prop. \ref{thm:partially-nameblind}, and thus we have proved the desired equivalence.
\end{proof}

\subsection{\texorpdfstring{$(m,n)$}{(m,n)}-nameblind matrices}

\subsubsection{Proof of Prop. \ref{prop:nameblind}}
Consider a $(m,n)$-nameblind matrix $M$;
$n=|\calA_{ex}|$ is the number of external addresses, and $m \leq n$ is the number of addresses present in any state handled by $M$. $M$ is an operator over words on $m$ addresses chosen among $n$. In the following, we will use the basis of such words, not sorted in lexicographic ascending order as in $\{12\dots (m-1) m, 12\dots (m-1) (m+1), \dots , n \dots (n-m+2) (n-m) , n \dots (n-m+2) (n-m+1)\}$, but instead sorted first by what addresses are present, and then by lexicographical ascending order, i.e.
\begin{equation*}
	\{12\dots (m-1) m, 12\dots m (m-1) , \dots , n \dots (n-m+1) (n-m+2) , n \dots (n-m+2) (n-m+1)\}.
\end{equation*}

Note that if $m<n$, $M$ acts on basis states with less addresses than all external addresses. Let us show that $M$ can be written blockwise with $(m,m)$-nameblind matrices (see Thm. \ref{thm:nameblind}). 

This state space is partitioned according to the set $\calA'$ (of size $m$) of what addresses are actually present, each subspace being called $\hat\calH_{\calA'} $ --- for some ${\calA'}\subset\calA_{ex}$. As $M$ does neither create nor suppress addresses, it is also block diagonal (in the particular ordering we chose). Lemma \ref{prop:equalblocks} proves that all these blocks are equal.

Next, we show that all blocks are nameblind. To this end, consider a set of addresses $\calB$ with corresponding block $M_{\cal B}$ which maps words from $\hat\calH_{\calB}$ to $\hat\calH_{\calB}$, where  $\hat\calH_{\calB}$ is the corresponding subspace. Then, for any renaming $R$ manipulating addresses within $\hat\calH_{\calB}$, and any $w \in \hat\calH_{\calB}$:
\begin{equation*}
	RMw = R M_{\calB} w, \qquad MRw = M_{\calB}Rw,
\end{equation*}
as $Rw \in \hat\calH_{\calB}$. The left sides of these equalities agree since $M$ is nameblind. Therefore, $M_{\calB}$ is nameblind, and as it is acting on words of size $m$ on $m$ addresses, it is $(m,m)$-nameblind.

\subsubsection{Lemma - Block equality}

The basis states of a $(m,n)$-nameblind matrix $M$ are ordered according to which addresses $\calA'$ are present. For the sake of simplicity of demonstrating Lem. \ref{prop:equalblocks}, we further specify the order between such address sets. Let $\calA'$ be a subset of $\calA_{ex}$ of size $m$. We define its \emph{size} as 
the integer corresponding to the indicator function (understood as little-endian binary number), and note it $\mathfrak B (\calA')$. For instance, with $\calA_{ex} = \{1,2,\dots,6\}$:
\begin{gather*}
	\mathfrak B ( \{2,3,4\}) = Int(011100) = 2^1 + 2^2 + 2^3 = 14	\\
	\mathfrak B ( \{1,2,5\}) = Int(110010) = 2^0 + 2^1 + 2^4 = 17
\end{gather*}

Note that only the binary numbers with $m$ bits set to $1$ will be represented. The order between finite address sets is then simply the usual order between their sizes.

\begin{figure}[H]
    \centering
    \begin{subfigure}[t]{0.49\textwidth}
    $\begin{array}{r@{~~<~~}l}
     \dots& \ket {011100}~~\leftrightarrow~~ \{2,3,4\}\\
        & \ket {110010}~~\leftrightarrow~~ \{1,2,5\}\\
        & \ket {101010} ~~\leftrightarrow~~ \{1,3,5\}\\
        & \ket {011010} ~~\leftrightarrow~~ \{2,3,5\}\\
        & \ket {100110} ~~\leftrightarrow~~ \{1,4,5\}\\
        & \ket {010110} ~~\leftrightarrow~~ \{2,4,5\}\\
        & \ket {001110} ~~\leftrightarrow~~ \{3,4,5\}\\
        & \ket {110001} ~~\leftrightarrow~~ \{1,2,6\}\\
        & \dots \\
    \end{array}$
    \caption{Writing address sets as binary numbers to determine their size.}
    \end{subfigure}
    \begin{subfigure}[t]{0.49\textwidth}
    $\begin{array}{r@{~~<~~}l}
        \dots              & \ket {\RC{01}1100}\\
        \ket {\blacier{01110}0} & \ket {\blacier{11001}0}\blacier{ -Carry }\\
        \ket {1\RC{10}010} & \ket {1\RC{01}010}\RC{ -Incrementation } \\
        \ket {\RC{10}1010} & \ket {\RC{01}1010}\\
        \ket {\blacier{0110}10} & \ket {\blacier{1001}10}\\
        \ket {\RC{10}0110} & \ket {\RC{01}0110}\\
        \ket {0\RC{10}110} & \ket {0\RC{01}110}\\
        \ket {\blacier{001110}} & \ket {\blacier{110001}}\\
        \ket {1\blacier{10}001} & \dots \\
    \end{array}$
    \caption{Effect of the renaming : when is each case used.}
    \end{subfigure}
    \caption{Order and types of successive basis blocks, by writing the image of $\calA$ by the Indicator function of $\calA_i$ as a binary number. The pattern identified by $\calB_{i,i+1}$ or $\calB'_{i,i+1}$ is then represented as colored numbers.}
    \label{fig:order}
\end{figure}

\begin{lemma}\label{prop:equalblocks}
Let $g$ be a gate of an AQC with address space $\calA$. Write $S^g$ in the basis sorted by:
\begin{enumerate}
    \item  First, according to the size $\mathfrak B$ of address sets ---see above.
	\item Second, for each group of addresses, the words are sorted in lexicographical increasing order.
\end{enumerate}
Then S is blockwise diagonal with blocks corresponding to which addresses are present, and each of those blocks are equal.
\end{lemma}
\begin{proof}

Consider the set of external addresses $\calA_{ex} = \{x_1, \dots, x_{n - |g|}\} \subset \calA $ of the gate operator $S^g$, and two successive sets of addresses $\calA_i = \{y_1 \dots y_{p}\} \subseteq \calA_{ex}$ and $\calA_{i+1} = \{y_1' \dots y_{p}'\} \subseteq \calA_{ex}$ (both ordered as specified in the lemma and as seen in Fig. \ref{fig:order}). Then, extract the smallest continuous varying parts $\calB_{i,i+1}  = \{y_l, y_{l+1}, \dots, y_{l+m}\}$ and $\calB_{i,i+1}' = \{y_{l}', \dots, y_{l+m}'\} $ --- also ordered the same way. Note that if $y_l = x_j$:
\begin{equation*}
	\begin{cases}\calB_{i,i+1} = \{x_j,x_{j+1},\dots,x_{j + m} \} \\ y_{l+m}' = x_{j+m+1} \end{cases}
\end{equation*}

We define now a renaming $R$ so that $\forall k, R(y_{l+k}) = y_{l+k}'$. As a consequence, $R$ always results in a identity block in the block column corresponding to the first state and in the line of the second state ---as this renaming preserves the relative order of such words on addresses. 

The renaming $R$ is defined with the following case distinction:
\begin{description}
    \item[- Incrementation] If $\calB_{i,i+1}$ is of the form $\{x_j\}$, then $\calB_{i,i+1}'$ is $\{x_{j+1}\}$. $R$ is defined as the transposition swapping $x_{j}$ and $x_{j+1}$.
    \item[- Carry] Else, $l=1$. Then $\calB_{i,i+1} = \{y_l = x_j, x_{j+1}, x_{j+2}, \dots  , x_{j+m-1}, x_{j+m} \}$ consists of $m$ successive numbers in $\calA_{ex}$ with $m > 0$, and $\calB_{i,i+1}' = \{y_1' = x_1, x_2,  \dots, x_{m-1}, x_{j+m+1} = y_{l+m}'\}$. Therefore, $R$ is defined as :
    \begin{enumerate}
        \item $\forall k \in \{1,\dots,{j-1}\}, R(x_k) = x_{k+m}$
        \item $\forall k \in \{j,\dots,{j+m-1}\}, R(x_{k}) = x_{k-j+1}$
        \item $R(y_{l+m}) = R(x_{j+m}) = x_{j+m+1} = y_{l+m}'$, $R( y_{l+m}')= x_{j+m}$
        \item $R(a) = a$ for any other $a\in\calA$
    \end{enumerate}
\end{description}
The property $ \forall k, R(y_{l+k}) = y_{l+k}'$ entails that $R\Pi_i = \Pi_{i+1}$, therefore : $\Pi_{i+1} R S \Pi_{i} = S_i $  and   $ \Pi_{i+1} S R \Pi_{i} = S_{i+1} $. As $RS=SR$ we have $S_i = S_{i+1}$.
\end{proof}

\subsection{From gate operators to \texorpdfstring{$(n,n)$}{(n,n)}-nameblind matrices}\label{annex: from GO to (n,n) nameblind matrices}
\begin{figure}[H]
    \centering
$\ket\eps_\calT\ket{abc}_\calI\ket\eps_\calO, \hspace{11pt} \ket\eps_\calT\ket{ab}_\calI\ket c_\calO,\hspace{11pt}\ket\eps_\calT\ket a _\calI\ket{b  c }_\calO,\hspace{11pt} \ket\eps_\calT\ket\eps_\calI\ket{abc}_\calO,$\\ \phantom{.}\\
$\ket a_\calT\ket{bc}_\calI\ket\eps_\calO ,\hspace{14pt} \ket a _\calT\ket b _\calI\ket c _\calO,\hspace{14pt} \ket a _ \calT\ket\eps_\calI\ket{bc}_\calO.$
    \caption{The 7 subspaces with 3 external addresses of a one sector gate.   $a,b,$ and $c$ represent a position of an address, spanned by any possible external address like in Fig. \ref{fig:subspaces:intext}. Each one of the 49 blocks delimited by these subspaces contains a nameblind matrix.}
    \label{fig:subspaces}
\end{figure}

\begin{manualprop}{\ref{prop:gateopdecomp}}
[Decomposition of gate operators into $(m,m)$-nameblind matrices]
Consider $S^g$ a gate operator on addresses $\calA$ which commutes with every external renaming. Consider any words on quantum data $q_0$ \& $q_1$, internal addresses $\calA_{int} \subseteq g$ with positions $P_{int_{1}}$ \& $P_{int_{2}}$, external addresses  $\calA_{ex} \subseteq  \calA \backslash g$ with positions $P_{ex_{1}}$ \& $P_{ex_{2}}$. The action of $S^g$ between states corresponding to these parameters $S^g_{\calA_{int},\calA_{ex},q_0\to q_1,P_{int_{1}}\to P_{int_{2}},P_{ex_{1}}\to P_{ex_{2}}}$ is a $(m,m)$-nameblind matrix with $m = |\calA_{ex}|$.
\end{manualprop}

\begin{proof}

Let $g$ be a gate and $S = S^g$ a nameblind gate operator. First we decompose the large gate space $\calH^g$ into $|\calQ^{2|g|}|$ subspaces for each possible basis state of the dataspace. For each couple of basis states $\ket{q_0}, \ket{q_1} \in {\calH}_{\calQ^{2|g|}}$ we denote by $S_{q_0\to q_1}$ the block of $S$ which maps states with data $\ket{q_0}$ to states with data $\ket{q_1}$. Since $S$ is nameblind we have :
\begin{equation*}
    \Pi_{q_1}SR\ket{\overline{a}}\ket{q_0}=S_{q_0\to q_1}R\ket{\overline{a}}\ket{q_0}.
\end{equation*}
On the other hand:
\begin{align*}
    \Pi_{q_1}SR\ket{\overline{a}}\ket{q_0}&=\Pi_{q_1}RS\ket{\overline{a}}\ket{q_0}\\
    &=RS_{q_0\to q_1}\ket{\overline{a}}\ket{q_0}
\end{align*}
where $\Pi_{q_1}$ is the projector onto the subspace with data state $q_1$. These equalities imply that each block must commute with renamings.

Furthermore, each $S_{q_0\to q_1}$ is block diagonal over subspaces $B_i$ defined by a number of external addresses $m \in \{0,\dots, |\calA\setminus g|\}$ and a set of internal addresses $\calA_{int}\subseteq g$, because both parameters are preserved by $S$. We cannot derive constraints between those diagonal blocks, since the subspaces $B_i$ are preserved by renamings and gate operators.

Then, we further divide the $B_i$ subspaces into subspaces characterized by a fixed position for each internal addresses of $\calA_{int}$, and $m$ "slots" for external addresses. We enumerate these subspaces  (see Figs. \ref{fig:subspaces:intext} and \ref{fig:subspaces}) and denote them by $\subspace i$. Moreover, we write $S_{[i,j]}=\Pi_j S_{\subspace i}$ where $S_{\subspace i}$ is projection of $S$ onto $\subspace i$. Its coefficients are the same as the block of $S$ in (block) column $i$ and (block) line $j$. Let us show that this block is nameblind.

Consider two subspaces $ \subspace i$ (corresponding to the block column/line $i$ in $S$) and  $\subspace j$ (column/line $j$), and a renaming $R$ that is external for $g$. We may for instance take  $\subspace i$ with states of the form $\ket{ axy }_I$, where $a\in g$ and $x,y$ span all external addresses. The block of $S$ specifying what part of $\subspace i$ goes to $\subspace j$ is $S_{[i,j]}$. Let $\ket{\mathbf w} \in \subspace i$ be a state. On one hand, as $\ket{\mathbf w}$ and $R\ket{\mathbf w}$ both belong to $\subspace i$, they are handled by the same block column $i$ of $S$ (they are in the same subspace, as $R$ does not modify or change the position of $a$, and make no internal address appear): 
\begin{equation*}
	\Pi_j (SR\ket{\mathbf w}) = S_{[i,j]} R\ket{\mathbf w}
\end{equation*}
On the other hand :
\begin{align*}
    \Pi_j (SR \ket{\mathbf w}) &= \Pi_j (RS\ket{\mathbf w}) \\
    &= \Pi_j ( \sum_{w\in \subspace j} \alpha_w \ket{R(w)} + \sum_{w\not\in \subspace j} \alpha_w\ket{R(w)} ) \\
    &=  R(\sum_{w\in \subspace j} \alpha_w \ket{w})\\
    &= RS_{[i,j]}\ket{\mathbf w},
\end{align*}
where in the third line we used that $R$ preserves subspaces and thus commutes with $\Pi_j$. As a consequence, for all $\ket{\mathbf w}, R$ we have $ S_{[i,j]} R \ket{\mathbf w}= R S_{[i,j]} \ket{\mathbf w}$, and therefore $S_{[i,j]}$ is nameblind. 

Therefore, if a gate operator commutes with every renaming on external addresses, its matrix representation can be written blockwise with nameblind matrices.

Vice versa, if a gate operator $S$ can be partitioned blockwise as above into nameblind matrices, we can write 
\begin{equation*}
	RS = R\sum_i\sum_j S_{[i,j]} = \sum_i\sum_j  \Big(S_{[i,j]}R \Big) = \Big(\sum_i\sum_j S_{[i,j]}\Big)R = SR.
\end{equation*}
\end{proof}

When handling the nameblind blocks $S_{[i,j]}$, we may simply consider that $S$ is acting on a state of the form $\ket {x_1\dots x_{m}}$ where $x_i$ are external addresses, and ignore the internal addresses and the specific subspaces in which are the external addresses, as those does not change. Therefore, by Prop. \ref{prop:nameblind} these operators are isomorphic to $(m,n)$-nameblind matrices and can hence be written blockwise as $(n,n)$-nameblind matrices by Thm. \ref{nameblind matrices}.

\section{Composing AQCs}\label{annex:operations}

Several ways of composing AQCs are possible. Here we define: 
\begin{description}
\item[Parallel composition] of AQCs corresponds to executing them in parallel without interaction.
\item[Concatenation] of AQCs corresponds to successively applying them to the quantum data.
\item[Connection] of AQCs is also defined, and allows for a more general connectivity without assuming an intrinsic (temporal) order between AQCs.
\end{description}

\subsection{Parallel composition}\label{annex:parallel}

The simplest combination of AQCs is parallel composition, which intuitively means executing them both side by side, separately:

\begin{definition}[Parallel composition] \label{def:parallel}
Let $\calA,\calB$ be two disjunct address sets, and let $C_\calA$ and $C_\calB$ be the corresponding AQC with skeletons $K_\calA = ( \calA,\calG_\calA,\calD_\calA^{\leq n},S_\calA)$ and $K_\calB = (\calB,\calG_\calB,\calD_\calB^{\leq m},S_\calB)$ and states $\ket {\psi_\calA}$ and $\ket {\psi_\calB}$. Their \emph{parallel composition} $C_\calA \mid \mid C_\calB$ is defined as $((\calA \cup \calB,\calG_\calA \cup \calG_\calB,\calQ_\calC,S_\calC),\ket{\psi_\calC})$, where $\calQ_\calC=(\calD_\calA\cup\calD_\calB)^{\leq m + n}$ and the scattering $S_\calC$ applies to $g \in \calG_\calA$ the trivial extension ${S_\calA^g}'$ of $S_\calA^g$ from $\calA$ to $\calA\cup\calB$ and from $\calQ_\calA$ to $\calQ_\calC$, and similarly for $g\in\calG_\calB$.
\end{definition}

By trivial extension, we mean that if $g\in\calG_\calA$, then $S^g_\calC$ acts as the identity on addresses of $\calB$ 
and on data of $\calQ_\calC \setminus \calQ_\calA$.

\subsection{ Connection }\label{annex:connection}

Explicitly including the sectors 1 and 4 in the Bell state creation circuit (see Fig. \ref{fig:epr}) or sectors $1$ and $6$ of the quantum switch (see Fig. \ref{fig:QS}) helps distinguishing where the data is coming from and where it is going to.  
AQC are infinitely evolving systems
and, as such, may not be perfectly fitted to express finite computations. However, these sectors, which at first sight seem superfluous, provide a unified template 
to portray finite quantum circuits, possibly followed by a measurement. In other words, they are useful to determine beginning and endpoint of a finite computation, and we use them as sectors where we place our initial data in and take the final data out.

Moreover, the presence of these sectors becomes crucial as soon as we want to compose AQCs and make them interact. Indeed, we will do so by identifying the outgoing sectors of an AQC, with the ingoing sectors of another. 

The following definition paves the road to composing AQCs, by making it explicit which sectors can be identified. 
\begin{definition}[Ingoing and outgoing buffer sectors]
Let $((\calA,\calG,\calQ,S),\ket{\psi})$ be an AQC and $a \in \calA$ be an address. We call the sector at $a$ an \emph{ingoing buffer sector} if
\begin{itemize}
    \item $a$ does not occur in $\ket\psi$, i.e. the sector cannot be targeted during the evolution,
    \item $\{a\}\in\calG$, and
    \item $S^{\{a\}} = F$ (the flip -- see \ref{annex:gates}).
\end{itemize}
Similarly, we call the sector at $a$ an \emph{outgoing buffer sector} if
\begin{itemize}
    \item $\ket{\psi}=\ket\eps^a \otimes \ket{\phi}$, i.e. the state of the sector $a$ is initially empty,
    \item $\{a\}\in\calG$, and
    \item $S^{\{a\}} = F$.
\end{itemize}
\end{definition}
As a consequence, it is possible to ``merge'' outgoing with ingoing sectors into a single sector having the address of the outgoing buffer sector and the internal state of the ingoing buffer sector, for an example see Fig. \ref{fig:inoutbs:2}.

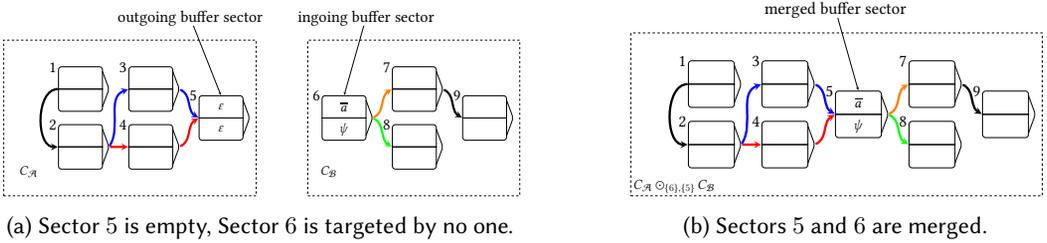
\begin{figure}[htb]
    \centering
    \begin{subfigure}[t]{0.5\textwidth}
        \centering
        \resizebox{\textwidth}{!}
        {\newcommand{\AQC}[5]{
\draw[fill=white,dashed,opacity=1] (#1,#2) rectangle (#3,#4) node[pos=0.1,above] {\huge #5};;
}
\newcommand{\aqcx}{12}

\begin{tikzpicture}[thick]
\AQC{-2.5}{-2.5}{9}{4.5}{$C_\calA$};
\AQC{\aqcx - 0.65}{-2.5}{\aqcx + 9}{4.5}{$C_\calB$};
\sector{0}{1.3}{}{1}{}{}
\sector{0}{-1.3}{}{2}{}{}
\sector{3.2}{1.3}{}{3}{}{}
\sector{3.2}{-1.3}{}{4}{}{}
\sector{6.4}{0}{$\eps$}{5}{$\eps$}{$\eps$}
\sector{\aqcx}{0}{$\psi$}{6}{}{$\ovva$}
\sector{\aqcx + 3.2}{1.3}{}{7}{}{}
\sector{\aqcx + 3.2}{-1.3}{}{8}{}{}
\sector{\aqcx + 6.4}{0}{}{9}{}{}
\node[inner sep=0pt] (joe) at (6,5.5) {\Huge outgoing buffer sector};
\draw[stealth-] (topmid5) --  (joe);
\node[inner sep=0pt] (joe2) at (14,5.5) {\Huge ingoing buffer sector};
\draw[stealth-] (topmid6) -- (joe2);
\path[-stealth,line width=1mm] (inarrow1) edge[bend right=90]  (inarrow2);
\fleche[out=60, in= 200]{outarrow2}{inarrow3}{blue}
\fleche[out=-30, in= 150]{outarrow3}{inarrow5}{blue}
\fleche{outarrow2}{inarrow4}{red}
\fleche[out=30, in= 210]{outarrow4}{inarrow5}{red}
\fleche[out=30, in= 210]{outarrow6}{inarrow7}{orange}
\fleche[out=-30, in= 150]{outarrow6}{inarrow8}{green}
\fleche[out=-30, in= 150]{outarrow7}{inarrow9}{black}
\end{tikzpicture}}
        \caption{Sector $5$ is empty, Sector $6$ is targeted by no one.}
        \label{fig:inoutbs:1}
    \end{subfigure}
    \hfill
    \begin{subfigure}[t]{0.4\textwidth}
        \centering
        \resizebox{\textwidth}{!}
        {\newcommand{\AQC}[5]{
\draw[fill=white,dashed,opacity=1] (#1,#2) rectangle (#3,#4) node[pos=0,above right] {\huge #5};;}
\newcommand{\aqcx}{6.4}

\begin{tikzpicture}[thick]
\AQC{-2.5}{-2.5}{\aqcx + 9}{4.5}{$C_\calA \odot_{\{6\},\{5\}} C_\calB$};
\sector{0}{1.3}{}{1}{}{}
\sector{0}{-1.3}{}{2}{}{}
\sector{3.2}{1.3}{}{3}{}{}
\sector{3.2}{-1.3}{}{4}{}{}
\sector{6.4}{0}{$\psi$}{5}{$\eps$}{$\ovva$}
\sector{\aqcx + 3.2}{1.3}{}{7}{}{}
\sector{\aqcx + 3.2}{-1.3}{}{8}{}{}
\sector{\aqcx + 6.4}{0}{}{9}{}{}
\node[inner sep=0pt] (joe) at (6.4 ,5.5) {\Huge merged buffer sector};
\draw[stealth-] (topmid5) -- (joe);
\path[-stealth,line width=1mm] (inarrow1) edge[bend right=90]  (inarrow2);
\fleche[out=60, in= 200]{outarrow2}{inarrow3}{blue}
\fleche[out=-30, in= 150]{outarrow3}{inarrow5}{blue}
\fleche{outarrow2}{inarrow4}{red}
\fleche[out=30, in= 210]{outarrow4}{inarrow5}{red}
\fleche[out=30, in= 210]{outarrow5}{inarrow7}{orange}
\fleche[out=-30, in= 150]{outarrow5}{inarrow8}{green}
\fleche[out=-30, in= 150]{outarrow7}{inarrow9}{black}
\node at (3.6,1.5) (secn3l) {};
\node at (3.6,0.5) (secn4l) {};
\node at (5,1.5) (secn3r) {};
\node at (5,0.5) (secn4r) {};
\node at (8.4,-0.3) (secn6) {};

\end{tikzpicture}}
        \caption{Sectors $5$ and $6$ are merged.}
        \label{fig:inoutbs:2}
    \end{subfigure}
    \caption{Concatenating two AQCs $C_\calA$ and $C_\calB$ along  $\{6\}$ and $\{5\}$. In the end, sector $5$ is no longer an outgoing buffer sector and sector $6$ does not exist anymore.}
    \label{fig:inoutbs}
\end{figure}

\begin{definition}[Connection of AQCs]\label{def:connection}
Let $k \in\mathbb{N}^*$ and $C_1 \dots C_k$ be AQCs with corresponding sets of addresses $\calA_1 \dots \calA_k $ so that $\forall i\neq j ,\calA_i \cap \calA_j = \varnothing$ and skeletons $K_{j} = ( \calA_j,\calG_{j},\calQ_{j},S_{j})$. Moreover, let $\ket {\psi}$ be the tensor product of their (initial) states, and $\In_{j}$ and $\Out_{j}$ the sets of their ingoing and outgoing buffer sectors. Let $ \In = \{i_1, \dots ,i_n\} \subseteq \bigcup\nolimits_{j=1}^k\In_j$ and $\Out = \{o_1, \dots ,o_n\} \subseteq\bigcup\nolimits_{j=1}^k \Out_{j}$ such that $\In\cap\Out=\varnothing$. The \emph {connection of $ C_{1} \dots C_{k}$ along $(\In,\Out)$} is denoted 
\begin{equation*}
    C = \Conn(\In,\Out,\{C_1\dots C_k\})  =  ((\calA_\calC,\calG_\calC,\calQ_\calC,S_\calC),\ket{\psi_\calC})
\end{equation*}
where
\begin{itemize}
    \item $\calA_\calC= \bigg(\bigcup\limits_{j=1}^k \calA_j\bigg) \setminus \In $ and $\calG_\calC= \bigg(\bigcup\limits_{j=1}^k \calG_j\bigg) \setminus \{ \{i\} ,i \in \In\}$ is the set of addresses of $C$ with corresponding gate set,
    \item $\ket{\psi_\calC} = \bigotimes\limits_{a\in\calA_\calC\setminus \Out} \ket{\psi_a}^a\otimes\bigotimes\limits_{j=1}^n\ket{\psi_{i_j}}^{o_j}$ when $\ket{\psi}$ is of the form $\bigotimes \ket{\psi_{a}}^a$, with the general case following by linear extension.
    \item $\calQ_C=\Large(\bigcup\limits_{j=1}^k\calD_j\Large)^{\leq N}$, where $N=\sum\limits_{i=1}^k n_i$ and those $n_i$ are such that $\calQ_i = \calD_i^{\leq n_i}$. 
    \item $S_C$ is the trivial extension of all $S_{j}$ as in Def. \ref{def:parallel}.
\end{itemize}
\end{definition}

The connection of an AQC with the trivial AQC, which consists only of a single empty sector, is the identity operation, given how ingoing and outgoing buffers are merged.

Moreover, note that parallel composition is a special type of connection, where the connecting sectors \emph{\In} and \emph{\Out} are empty, i.e. 
\begin{equation*}
    C_\calA\mid \mid C_\calB = \Conn (\varnothing,\varnothing, \{C_\calA, C_\calB\}).
\end{equation*}

With such a general notion of connection, one may connect an AQCs to form a loop, or even connect an AQC to itself ($\Conn (\In,\Out, \{C_\calA\})$). The resulting AQC may be an evolving structure with neither ingoing nor outgoing buffers.

Such constructions are of course not very physical, since they lack a way for quantum data to enter and or exit, and thus cannot be connected to any other AQC: For an AQC to correspond to a quantum circuit, it has to include at least one ingoing and one outgoing buffer. However, the AQC model as presented in this paper allows for the definition of more general systems, without an input or an output node.  Intuitively, such AQC are not physically implementable as a circuit in a satisfying way, as they lack a way for quantum data to enter and to exit it.
They are still causal structures, as the evolution is driven by local reversible rules and gates may be applied on no data, i.e. vacuum states \cite{Quantum-Shannon-Theory}; no causality problem arises if some data enters the same sector twice during the infinite evolution.

The alternative to composition with ingoing and outgoing buffers would have been to manipulate ``open AQC'', i.e. featuring addresses that correspond to no sector, and then plugging these open AQCs together. This causes several problems however, as the transport step may be undefined, and renaming only one of the open AQCs may change the nature of the composition. 

\subsection{Concatenation}\label{annex:concatenation}
Another type of connection is a concatenation : $C_\calA \odot_{\In,\Out} C_\calB = \Conn(\In,\Out,\{C_\calA, C_\calB\})$, where $\In$ only contains addresses from $\calB$ and $\Out$ only address from $\calA$.

Concatenation satisfies some kind of associativity property: for all input (resp. output) sets $\alpha,\beta_1,\beta_2,\gamma$ : 
\begin{equation*}
	\beta_1\cap\beta_2=\varnothing \qquad \Rightarrow \qquad (A \odot_{\beta_1,\alpha} B) \odot_{\gamma,\beta_2} C = A \odot_{\beta_1,\alpha} (B \odot_{\gamma,\beta_2} C).
\end{equation*}

Without the condition to the left, the equality on the right holds, but may also be meaningless, as some sectors might have disappeared during the first operation to be applied, or simply not be outgoing or ingoing anymore. Notice that similar to causal boxes \cite{CausalBoxes} and contrary to quantum combs, the concatenation of two AQCs is an AQC.

\begin{example}\label{example:concatenation}
Let us concatenate the quantum switch $A$ defined in Sect. \ref{sec:qs} and the Bell state creation circuit $B$ of Fig. \ref{fig:epr} with address sets $\calA=\{1,2,3,4,5,6\}$ and $\calB = \{7,8,9,0\}$, respectively. As unitaries in the quantum switch, take the Pauli gates $U=Y$ and $V=Z$ on a single qubit. Moreover, to facilitate matters we take $\ket{\psi_A}=\ket\eps$ as initial state of the quantum switch so that after concatenation data is only present at the beginning of the Bell state creation circuit. We show in detail the steps of the evolution of $A$ and $B$ concatenated along $\alpha = \{ 0 \}$ and $\beta = \{ 1 \}$ in Fig. \ref{fig:eprqs}, where $B \odot_{\alpha,\beta} A = \Conn(\alpha,\beta,\{B,A\})$.

\begin{figure}[H]
    \begin{center}
    \begin{subfigure}[ht!]{\textwidth}
        \centering
        \resizebox{0.65\textwidth}{!}
        {\renewcommand{\eps}{\epsgr}\begin{tikzpicture}[thick]
\sector{-10.4}{0}{00}{7}{$\eps$}{$\eps$}
\sector{-7.2}{0}{$\eps$}{8}{$\eps$}{$\eps$}
\sector{-4}{0}{$\eps$}{9}{$\eps$}{$\eps$}
\sector{0}{1.3}{$\eps$}{1}{$\eps$}{$\eps$}
\sector{0}{-1.3}{$\eps$}{2}{$\eps$}{$345_\calO$}
\sector{3.2}{1.3}{$\eps$}{3}{$\eps$}{$\eps$}
\sector{3.2}{-1.3}{$\eps$}{4}{$\eps$}{$\eps$}
\sector{6.4}{1.3}{$\eps$}{5}{$\eps$}{$\eps$}
\sector{6.4}{-1.3}{$\eps$}{6}{$\eps$}{$\eps$}
\node (in12) at (0.1,2) {};
\node (in11) at (0.1,2.6) {};
\fleche[out= 20, in= 190]{outarrow9}{in11}{black}
\fleche[bend right=70]{in12}{inarrow2}{black}
\fleche{outarrow7}{inarrow8}{black}
\fleche{outarrow8}{inarrow9}{black}
\node at (3.6,1.5) (secn3l) {};
\node at (3.6,0.5) (secn4l) {};
\node at (5,1.5) (secn3r) {};
\node at (5,0.5) (secn4r) {};
\path[-stealth,line width=1mm] (outarrow5) edge[bend left=70]  (outarrow6); 
\end{tikzpicture}\renewcommand{\eps}{\epsgr}}
        \caption{Initial state of the concatenation $B \odot_{\alpha,\beta} A$ of the Bell state creation circuit and the quantum switch.}
        \label{fig:eprqs:1}
    \end{subfigure}
    
    \begin{subfigure}[ht!]{\textwidth}
        \centering
        \resizebox{0.65\textwidth}{!}
        {\renewcommand{\eps}{\epsgr}\begin{tikzpicture}[thick]
\sector{-10.4}{0}{$\eps$}{7}{$\eps$}{$\eps$}
\sector{-7.2}{0}{$\eps$}{8}{$\eps$}{$\eps$}
\sector{-4}{0}{$\eps$}{9}{$\eps$}{$\eps$}
\sector{0}{1.3}{\scalebox{0.82}{(\textcolor{blue}{00}+\textcolor{red}{11})}}{1}{$\eps$}{$\eps$}
\sector{0}{-1.3}{$\eps$}{2}{$\eps$}{$345_\calO$}
\sector{3.2}{1.3}{$\eps$}{3}{$\eps$}{$\eps$}
\sector{3.2}{-1.3}{$\eps$}{4}{$\eps$}{$\eps$}
\sector{6.4}{1.3}{$\eps$}{5}{$\eps$}{$\eps$}
\sector{6.4}{-1.3}{$\eps$}{6}{$\eps$}{$\eps$}
\node (in12) at (0.1,2) {};
\node (in11) at (0.1,2.6) {};
\fleche[out= 20, in= 190]{outarrow9}{in11}{black}
\fleche[bend right=70]{in12}{inarrow2}{black}
\fleche{outarrow7}{inarrow8}{black}
\fleche{outarrow8}{inarrow9}{black}
\node at (3.6,1.5) (secn3l) {};
\node at (3.6,0.5) (secn4l) {};
\node at (5,1.5) (secn3r) {};
\node at (5,0.5) (secn4r) {};
\path[-stealth,line width=1mm] (outarrow5) edge[bend left=70]  (outarrow6); 
\end{tikzpicture}\renewcommand{\eps}{\epsgr}}
        \caption{State of $B \odot_{\alpha,\beta} A$ after the Bell state circuit.}
        \label{fig:eprqs:2}
    \end{subfigure}

    \begin{subfigure}[ht!]{\textwidth}
        \centering
        \resizebox{0.65\textwidth}{!}
        {\renewcommand{\eps}{\epsgr}\begin{tikzpicture}[thick]
\sector{-10.4}{0}{$\eps$}{7}{$\eps$}{$\eps$}
\sector{-7.2}{0}{$\eps$}{8}{$\eps$}{$\eps$}
\sector{-4}{0}{$\eps$}{9}{$\eps$}{$\eps$}
\sector{0}{1.3}{$\eps$}{1}{$\eps$}{$\eps$}
\sector{0}{-1.3}{$\eps$}{2}{$\eps$}{$345_\calO$}
\sector{3.2}{1.3}{$\eps$}{3}{$\eps$}{$\eps$}
\sector{3.2}{-1.3}{$\eps$}{4}{$\eps$}{$\eps$}
\sector{6.4}{1.3}{$\eps$}{5}{$\eps$}{$\eps$}
\sector{6.4}{-1.3}{\scalebox{0.83}{(\textcolor{blue}{10}\scalebox{0.75}[1.0]{$-$}\textcolor{red}{01})}}{6}{$\eps$}{$\eps$}
\node (in12) at (0.1,2) {};
\node (in11) at (0.1,2.6) {};
\fleche[out= 20, in= 190]{outarrow9}{in11}{black}
\fleche[bend right=70]{in12}{inarrow2}{black}
\fleche{outarrow7}{inarrow8}{black}
\fleche{outarrow8}{inarrow9}{black}
\node at (3.6,1.5) (secn3l) {};
\node at (3.6,0.5) (secn4l) {};
\node at (5,1.5) (secn3r) {};
\node at (5,0.5) (secn4r) {};
\path[-stealth,line width=1mm] (outarrow5) edge[bend left=70]  (outarrow6); 
\end{tikzpicture}\renewcommand{\eps}{\epsgr}}
        \caption{``Final'' state of $B \odot_{\alpha,\beta} A$ after the Quantum Switch AQC. A global phase $i$ is omitted.}
        \label{fig:eprqs:3}
    \end{subfigure}
    \caption{Evolution of $B \odot_{\alpha,\beta} A$, concatenation of two AQCs of Ex. \ref{example:concatenation}, with initial state $00$}
    \label{fig:eprqs}
    \end{center}
\end{figure}
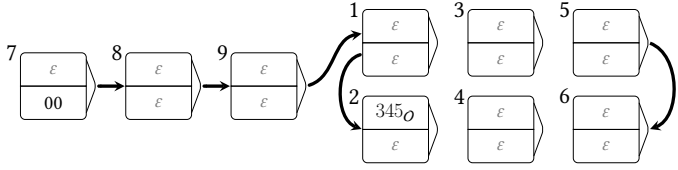
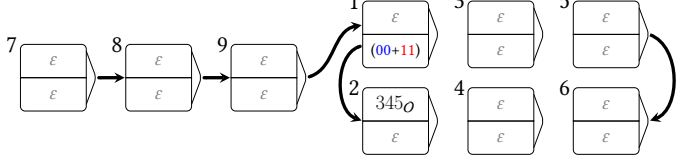
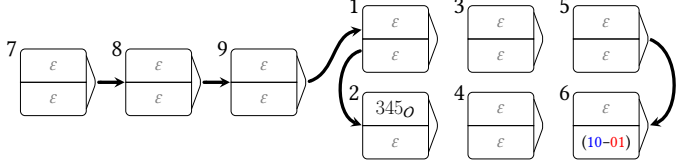
\end{example}

\begin{remark}
Outgoing buffer sectors correspond to the physical locations where the quantum data is measured. We do not discuss the possibility of measuring sectors here, yet it is straightforward to see how to do that. To represent a circuit with definite beginning and end, one would execute the AQC until quantum data reaches an outgoing sector.
\end{remark}

\section{Relation between AQC and QCGD}\label{section:QCGD}

Another approach to represent indefinite causal orders is to shun away from the circuit formalism and take inspiration from cellular automata. The QCGD framework \cite{ArrighiQCGD} features quantum superpositions of graphs evolving synchronously according to unitary local rules. Each vertex has a local quantum state and communicates with other vertices via ports. Compared to AQCs, vertices have replaced gates and edges have replaced wires.

Like QCGD, our model is able to represent a quantum superposition of different graphs (the graphs of connectivity between gates). 
The interdiction of non-causal evolution in QCGD may appear as contradictory with the sudden appearance of an edge in the AQC between nodes of different gates --- for instance, when an address suddenly goes from the input address space to the target.
However, by encoding address spaces into port names, one may simulate this seemingly non-causal edge appearance at the cost of having a larger number of ports ---see below
. We can therefore encode any AQC into a QCGD. The reverse simulation is a challenging open problem.

\subsubsection*{Encoding of AQC into the QCGD formalism} 
A gate operator can be encoded on local unitary rewritings, and multi-sector gates are implemented by connecting its sectors to a central common vertex, via a certain port $g$.

A minor difficulty is that in an AQC, a sector may suddenly target seemingly faraway sector, whereas in a QCGD, vertices may only connect to closeby vertices. However, when an AQC does this, it is because the faraway sector's address is a stored address. Therefore, we must encode stored address spaces as QCGD edges.

For clarity purposes, the translation presented here has $3n+1$ ports and $n+m$ vertices for an AQC with $n$ sectors and $m$ gates, but smaller, less trivial encodings are possible.

\begin{figure}[t]
	\centering
	\begin{minipage}[t]{0.38\textwidth}
		\centering
		\resizebox{\textwidth}{!}{
			\begin{tikzpicture}[thick]

\tikzset{arc text/.style={%
	decorate,
	decoration={%
		text effects along path,
		text={#1},
		text align=center,
		text effects/.cd,
		text along path,
		characters={anchor=mid}
	}
}}

\def\radius{2.5cm}
\def\dist{10pt}
\node[font=\Huge] (0,0) {$s_a$};
\draw (\radius,0) arc [start angle= 0, end angle= 360, radius= \radius]
	node[pos=.5, draw, fill, circle, inner sep=0pt, minimum size=4pt, xshift=-2pt, label={left:$e$}] {}
	node[pos=0 , draw, fill, circle, inner sep=0pt, minimum size=4pt, xshift=2pt, label={right:$\uptau$}] {}
	;

\foreach \ang in {65,215,305}{
	\path[arc text={...}] (\ang:\radius+\dist/2) arc (\ang:\ang+30:\radius+\dist/2); 
}
	
\foreach \ang/\lab [evaluate=\i as \m using {\ang < 180 ? \ang+5 : \ang-5 }, evaluate=\ang as \n using { \ang < 180 ? \ang-5 : \ang+5 }] 
	in {65/{$p_{a_n}$}, 95/{$p_{a_2}$}, 115/{$p_{a_1}$},
		195/{$in_1$}, 215/{$in_2$}, 245/{$in_n$},
		285/{$out_1$}, 305/{$out_2$}, 335/{$out_n$}
	} {
	\draw[fill] (\ang:\radius+2pt) circle (2pt);
	\path[arc text={\lab}] (\m:\radius+\dist) arc (\m:\n:\radius+\dist);
}

\end{tikzpicture}}
		\caption{Ports of a sector vertex $s_a$,\\ where~$\calA = \{a_1,a_2,\dots,a_n\}$. }
		\label{fig:QCGD:ports}
	\end{minipage}\hfill
	\begin{minipage}[t]{0.59\textwidth}
		\centering
		\resizebox{\textwidth}{!}{
			\pgfdeclarelayer{edgelayer}
\pgfdeclarelayer{nodelayer}
\pgfsetlayers{edgelayer,nodelayer}

\tikzstyle{gat}=[circle,fill=red,draw=black,line width=0.7 pt]
\tikzstyle{sec}=[circle,fill=white,draw=black,line width=0.7 pt]

\begin{tikzpicture}
	\begin{pgfonlayer}{nodelayer}
		\node [style=gat] (0) at (-2.5, -0) {$r_{\{2,5\}}$};
		\node [style=sec] (1) at (-4, -0) {$s_2$};
		\node [style=sec] (2) at (-1, -0) {$s_5$};
		\node [style=sec] (3) at (-3, 2) {$s_3$};
		\node [style=sec] (4) at (-3, -2) {$s_4$};
		\node [style=sec] (5) at (-6, -0) {$s_1$};
		\node [style=sec] (6) at (1, -0) {$s_6$};
		\node [style=gat] (7) at (-1.5, 2) {$r_{\{3\}}$};
		\node [style=gat] (8) at (-6, -1.5) {$r_{\{1\}}$};
		\node [style=gat] (9) at (-1.5, -2) {$r_{\{4\}}$};
		\node [style=gat] (10) at (1, -1.5) {$r_{\{6\}}$};
	\end{pgfonlayer}
	\begin{pgfonlayer}{edgelayer}
		\draw (1) to (5);
		\draw [color=blue] (1) to (3);
		\draw [color=blue] (1) to (4);
		\draw [color=black] (2) to (6);
		\draw [color=red] (0) to (1);
		\draw [color=red] (0) to (2);
		\draw [color=red] (4) to (9);
		\draw [color=red] (8) to (5);
		\draw [color=red] (10) to (6);
		\draw [color=red] (3) to (7);
		\draw [color = blue,bend right = 60] (1) to (2);
	\end{pgfonlayer}
\end{tikzpicture}}
		\caption{Simplified representation of the encoding as a QCGD of the quantum switch. Notice that ports are not represented. Red nodes are gate-vertices, white nodes are sector-vertices. Red edges use the ports $p_a$, black edges link the target port $\uptau$ to the entry port $e$, blue edges link output ports $out$ to entry ports $e$.}
		\label{fig:QCGD:QS}
	\end{minipage}
\end{figure}

\begin{definition}[Encoding into QCGD]\label{def:QCGD_enc}
Let $C = ((\calA,\calG,\calQ,S),\ket{\psi}) $ be an AQC.
In the canonical basis of the circuit space, $\ket{\psi}={\sum_k}\alpha_k \ket{c_k}$, where each $c_k$ is a basis state. The superposition of quantum graphs $e(C) = \ket{ \mathbf{G}}=\sum_k \alpha_k \ket{\mathbf{G}_{c_k}}$ encodes $C$, where each $\mathbf{G}_{c_k}$ is defined as follows: 

\begin{itemize}
    \item The set of vertices is :
    $\calV = \{ s_a \mid a \in \calA\} \cup \{r_g \mid g\in \calG\}$ where the vertices
    $\{ s_a  \mid a \in \calA\}$ represent the sectors, while the vertices $\{r_g \mid g\in \calG\}$ represent the gates and will be used to regroup all sectors of a same gate.
    \item Each vertex has ports $\pi = \{p_a  \mid a \in \calA\}\cup\{e\}\cup\{\uptau,in_1,in_2,\dots,in_{|\calA|}, out_1,out_2,\dots,out_{|\calA|}\}$.
    
    The ports $\{p_a \mid a \in \calA\}$ are used by the gate-vertices to connect to their sector-vertices and reciprocally.
    The ports $\{\uptau,in_1,in_2,\dots,in_{|\calA|}, out_1,out_2,\dots,out_{|\calA|}\}$
    are used to represent the target space, the input address space and the output address space of the sectors. Any edge from those ports is connected to the entry port $e$ of the corresponding addressed sector-vertex. 

    \item The set of labels depends on $\ket{c_k}$. Each sector $a$ having a quantum data space containing the state $\ket{\ovvq}_{\calQ_\calI}^a\ket{\ovvqp}_{\calQ_\calO}^a$ sees its corresponding vertex $s_a$ being labelled with $(\ovvq,\ovvqp)$.
    Each vertex $r_g$ is labelled with $S^g$.
    
    \item The set of edges depends on $\ket{c_k}$. Each sector $a$ having an address space containing the state $\ket t ^a\ket{\ovva}_{\calW_\calI}^a\ket{\ovvap}_{\calW_\calO}^a$
    sees its corresponding vertex $s_a$ have the following edges :
    \begin{gather*}
    \{\{s_a: in_i, a_i : e\} \mid\forall a_i \in \ovva \}\\
    \{\{s_a: out_i, a'_i : e\} \mid\forall a'_i \in \ovvap \}\\
    \{s_a:\uptau,t:e\}
    \end{gather*}
    
    Moreover, each sector $a$ is in a gate $g$, embodied by the edge :
    \begin{equation*}
    	\{s_a:p_a, r_g:p_a\}
    \end{equation*}
\end{itemize}
\end{definition}

\begin{remark}
The condition "address appear at most once in an AQC" in Def. \ref{def:AQC} ensures that at most one edge uses the port $e$ of any vertex.
\end{remark}

\begin{definition}[Evolution of Encoding]
The operator $E = E_T E_S$ is the following quantum causal graph dynamic acting on $\ket{\mathbf{G}_{c_k}}$ by changing the labels and edges of vertices according to the labels and edges of ``nearby'' vertices:
\begin{itemize}
    \item $E_T$ implements the transport step. It applies synchronously for every edge $\{s_a:\uptau , s_t : e\}$  the following evolution:
    \begin{itemize}
        \item \emph{(Swap of input and output address spaces)} For any $j,k \in [1\dots|\calA|]$, the edges $\{s_t:in_j,s_k:e\}$ disappear and are replaced by $\{s_a:out_j,s_k:e\}$, while edges $\{s_a:out_j,s_k:e\}$ disappear and become edges $\{s_t:in_j,s_k:e\}$
        
        \item \emph{(Swap of input and output data spaces)} Let $\ovvqp$ be the second part of the label of vertex $s_a$. Let $\overline{p}$ be the first part of the label of vertex $s_t$. 
        After applying $E_T$,
        the second part of the label of vertex $s_a$ becomes $\overline{p}$, and the first part of the label of vertex $s_t$ becomes $\ovvqp$.
    \end{itemize}

    \item $E_S$ implements the scattering step. It is synchronously applied on every $r_g\in\calV$ of label $S^g$, corresponding to the gate $g\in\calG$, and can be explained by the following substeps :
    \begin{enumerate}
        \item Remember that the vertices $\{s_a \mid a\in g\}$ are exactly the neighbours of vertex $r_g$.
        One builds the state of the equivalent AQC-gate space : 
        
        \qquad \qquad \qquad \quad $ \ket\psi^g = \bigotimes_{a \in g} \ket\psi^a$,

        where each $\ket\psi^a$ is build from the quantum labelled graph:
        \begin{itemize}
            \item The target space is $\ket t$ if there is an edge $\{s_a:\uptau,s_t:e\}$. Else, it is $\ket\eps$.
            \item The input address space contains the (potentially empty) word  $\ovva= a_1\dots a_{|\calA|}$ where, for each edge $\{s_a : in_i, s_z : e\}$, $a_i = z$. If there exists no such edge for some  $i\in[1\dots|\calA|]$, the corresponding $a_i$ is equal to $\eps$.
            \item The output address space is filled like the input address space, by replacing every $in_i$ by a $out_i$ in the line above.
            \item Let $(\ovvq,\ovvqp)$ be the label of vertex $s_a$. The input and output address spaces of sector $a$ are $\ket{\ovvq}_{\calQ_\calI}^a\ket{\ovvqp}_{\calQ_\calO}^a$.
        \end{itemize}
        \item $S^g$ (obtained from the label of vertex $r_g$) is applied to the gate space $ \bigotimes_{a \in g} \ket\psi^a$, which therefore becomes $ \ket{\psi'}^g = S^g\bigotimes_{a \in g} \ket\psi^a$
        \item The edges and labels are rewritten from the state $\ket{\psi'}^g$, using the construction explained in step 1 (and Def. \ref{def:QCGD_enc}). in the other direction---this time constructing the subgraph (or potentially a superposition of subgraphs) from the gate space.  
    \end{enumerate}
\end{itemize}
\end{definition}

\begin{proposition}[Equivalence of Evolution of Encoding]
Let $C$ be an AQC. 
Then $e(GC)=E e(C)$ where $E=E_TE_S$ implements the AQC evolution operator $G=TS$ and with $e$ as in Def. \ref{def:QCGD_enc}.
\end{proposition}

\section{Notations}\label{annex:notations}
\label{annex:trivnotations}

Let ${\cal S}$ be a finite set. We denote by $\overline s$ finite strings with alphabet $\mathcal S$, i.e. $\overline s=s_1s_2\dots s_n$ with $s_i\in\mathcal S$. Moreover, $\eps$ denotes the empty string. Words of addresses written as letters with overlines ($\ovva$) may be of length $0$. Letters without overline ($a$) represent words of length $1$. From a finite set $\mathcal S$ we construct the Hilbert space $\calH_{\cal S} = \mathbb{C}^{\cal S} = \{ \sum_{s\in {\cal S}} \alpha_s \ket{s} , \alpha_s \in \mathbb{C} \}$ with scalar product $\langle s|s'\rangle=\delta_{s,s'}$.
In other words, to each element $s\in\calS$ is associated a unit vector $\ket s$, such that the family $(\ket s)_{s \in \calS}$ is the canonical orthonormal basis of $\calH_{\cal S} $. If $\mathcal S$ were countably infinite, we would have ensured square-summability by taking $\calH=\ell^2(\mathcal S)$.

The other notations used in the paper are summarized in Tab.~\ref{tab:notations} on the next page.
\clearpage

\label{annex:conventions}

\setlength{\extrarowheight}{1pt}
\begin{longtable}{|>{\centering}p{35pt}|p{9cm}|p{43pt}|}
\hline
Notation & Definition& Occurrence\\
\hline
     \calA &  A set of addresses, namely, a finite subset of $\mathbb{N}$. &  Def. \ref{def:targetstored}\\
     \calG &  A set of gates, namely, a partition of some \calA.&Def. \ref{def:gates}\\
    \calT & A set containing all addresses and the empty address $\eps$, called set of targets.&Def. \ref{def:targetstored} \\
    \calW & The set of words on addresses where no address repeats. $\calW = \{a_1\dots a_n|a_i\in\calA,a_i\neq a_j\:\forall i,j\}$&Def. \ref{def:targetstored}\&\ref{def:addresses}\\
     $\overline{a}$ & A list of addresses. It usually denotes an element of $\calW$ or $\calW\otimes\calW$.&Def. \ref{def:targetstored}\\
     \calD & A finite set, called set of data values.&Def. \ref{def:dataspace} \\
    \calQ & The set of data words of length at most $n$ for some $n\in\mathbb{N}$-- i.e. $\calQ=\calD^{\leq n}$.& Def. \ref{def:dataspace} \\
    $\overline{q}$ & An element in \calQ.&Def. \ref{def:dataspace} \\
     $\forall \calS, \calH_\calS$& The Hilbert space constructed from any set $\calS$ as specified in Sect. \ref{annex:trivnotations}&Sect. \ref{annex:notations}\\
    $\calH_\calT$& Hilbert space on $\calT$ called target space.& Def. \ref{def:targetstored}\\
    $\calH_\calW$& Hilbert space on $\calW$ called stored address space.&Def. \ref{def:targetstored}\\
    $\calH_\calQ$& Hilbert space on $\calQ$ called data space.&Def. \ref{def:dataspace}\\
    $\calH_I$& Input space. It is the tensor product of $\calH_\calW$ and $\calH_\calQ$.&Def. \ref{def:sectors}\\
    $\calH_O$& Output space. It is the tensor product of $\calH_\calW$ and $\calH_\calQ$.&Def. \ref{def:sectors}\\
    $\calH^a$ & The sector space associated to $a$. It is the tensor product of $\calH_{\calT}$,$\calH_I$ and $\calH_O$.&Def. \ref{def:sectors}\&\ref{def:gates}\\
    $\calH^g$ & The sector space associated to a gate $g\subset\calA$ : $\calH^g=\bigotimes_{a\in g}\calH^a$&Def \ref{def:gates}\\
    $\calH$ & Circuit space. $\calH= \bigotimes_{a\in \calA}\calH^a$ &Def. \ref{def:AQC}\\
    
    $\ket{a}_\calT$& A basis state of $\calH_\calT$.&Eq. \ref{eq:equivalence}\\

    $\ket{\overline{a},\overline{q}}_\calI$& A basis state of $\calH_\calI$.&Ex. \ref{example:bell_exec}\\

    $\ket{\overline{a},\overline{q}}_\calO$& A basis state of $\calH_\calO$.&Eq. \ref{eq:S25:s2}\\

    $\ket{\overline{a},\overline{q}}_\calO^s$& A basis state of the output space of sector $s$.&Eq. \ref{eq:psi0}\\
    $\ket{\overline{a},\overline{a}'}_\calW$
    & A basis state of $\calH_{\calW_I}\otimes \calH_{\calW_O}$.&Eq. \ref{eq:equivalence}\\
    $\ket{\overline{q},\overline{q}'}_\calQ$& A basis state of $\calH_{\calQ_I}\otimes \calH_{\calQ_O}$.&Eq. \ref{eq:equivalence}\\
    $\ket{x}^i$ & A basis state of the i-th sector space. It can be decomposed either as~: $\ket{a}^i_\calT\ket{\overline{a},\overline{q}}^i_\calI\ket{\overline{a}',\overline{q}'}^i_\calO$ or as~: $\ket{a}^i_\calT\ket{\overline{a},\overline{a}'}^i_\calW\ket{\overline{q},\overline{q}'}^i_\calQ$&Def. \ref{def:Evolution}\\
    $\ket{x}_y^i$ & Where $ y \in \{\calT,\calW,\calQ\}$ is the corresponding space of the $i$-th sector.&Def. \ref{def:Evolution}\\
    $\ket{32}$ & The state $\ket {3\hspace{1em}2}$ (ket--three--two), while being distinct from the state $\ket {32}$ (ket--thirty--two), is always written as $\ket {32}$ . This causes no confusion, as only address sets with fewer than ten addresses are used in this paper.& Ex. \ref{example:ket32}\\
    \hline
    $S$ & The scattering step. &Def. \ref{def:Evolution}\\
    $T$ & The transport step. &Def. \ref{def:Evolution}\\
    $G$ & $G = TS$, the global evolution operator. &Def. \ref{def:Evolution}\\\hline
    \caption{\label{tab:notations}Notations used in the paper}
\end{longtable}

\end{document}